\documentclass[a4paper,reqno,11pt]{article}

\usepackage[a4paper,bindingoffset=0.2in,%
left=1in,right=1in,top=1in,bottom=1in,%
           footskip=.25in]{geometry}
\usepackage{graphicx}

\usepackage{authblk}
\usepackage{amssymb}
\usepackage{amsthm}
\usepackage{amsmath}
\usepackage[english]{babel}
\usepackage{color}
\usepackage[utf8]{inputenc}

\usepackage[makeroom]{cancel} 

\usepackage{tikz}
\usepackage[all]{xy}
\usepackage[bookmarksnumbered,colorlinks]{hyperref}
\usepackage{dsfont}
\usepackage[normalem]{ulem}

\usepackage{float} 

\usepackage{enumerate}
\def\br#1\er{\textcolor{red}{#1}} %
\newcommand{\cvd}{\ \rule{0.5em}{0.5em} \smallskip}
\newcommand{\be}{\begin{equation}}
\newcommand{\ee}{\end{equation}}

\newcommand{\N}{{\mathbb N}}
\newcommand{\Z}{{\mathbb Z}}
\newcommand{\R}{{\mathbb R}}
\newcommand{\LL}{{\mathbb L}}

\newcommand{\M}{\overline{M}} 
\newcommand{\pM}{{\partial M}}
\newcommand{\hp}{{\hat p}}
\newcommand{\bSigma}{\bar \Sigma}

\newcommand{\ben}{\begin{enumerate}}
\newcommand{\een}{\end{enumerate}}
\newcommand{\bit}{\begin{itemize}}
\newcommand{\eit}{\end{itemize}}
\newcommand{\edoc}{\end{document}}

\newcommand{\cambios}[1]{{\color{black} #1}} 
\newcommand{\ncambios}[1]{{\color{black} #1}}

\newcommand{\bcambios}[1]{\textcolor{black}{#1}}
\newcommand{\gcambios}[1]{\textcolor{black}{#1}}

\def\bb#1\eb{\textcolor{blue}
{#1}} %
\def\br#1\er{\textcolor{red}
{#1}} %
\def\bv#1\ev{\textcolor{olive}
{#1}} %
 \def\bm#1\em{\textcolor{magenta}
 {#1}} %

\renewcommand{\j}{\mathfrak{j}}

\newtheorem{thm}{Theorem}[section]
\newtheorem{prop}[thm]{Proposition}
\newtheorem{lemma}[thm]{Lemma}
\newtheorem{cor}[thm]{Corollary}
\newtheorem{conv}[thm]{Convention}
\newtheorem{defi}[thm]{Definition}

\newtheorem{rem}[thm]{Remark}

\title{Structure of globally hyperbolic \\ spacetimes-with-timelike-boundary
}
\author[1]{L. Ak\'e}
\author[1]{J.L. Flores}
\author[2]{M. S\'anchez}
\affil[1]{\small{{\it Departamento de \'Algebra, Geometr\'{i}a y Topolog\'{i}a\\ Facultad de Ciencias, Universidad de M\'alaga\\ Campus Teatinos, 29071 M\'alaga, Spain.}}}
\affil[2]{\small{{\it Departamento de Geometr\'{i}a y Topolog\'{i}a\\ Facultad de Ciencias, Universidad de Granada\\ Campus Fuentenueva s/n, 18071 Granada, Spain.}}}
\date{}
\begin{document}

\maketitle

	\begin{abstract}
\noindent
Globally hyperbolic spacetimes-with-timelike-boundary $(\M = M \cup  \pM, g)$ are the natural class of spacetimes  where regular boundary conditions (eventually asymptotic, if $\pM$ is obtained by means of a conformal embedding) can be posed. $\pM$
represents  the naked singularities and can be identified with a part of the intrinsic causal boundary.
Apart from general properties of $\pM$,  the splitting of any globally hyperbolic $(\M,g)$ as an orthogonal product $\R\times \bSigma$
with Cauchy \ncambios{slices-with-boundary}  $\{t\}\times \bSigma$ is proved. This is obtained by constructing  a Cauchy temporal function $\tau$ with gradient $\nabla \tau$ tangent to $\pM$
on the boundary.
To construct such a $\tau$, results on stability of both global hyperbolicity and Cauchy temporal functions
are obtained. Apart from having   their own interest, these results allow us to circumvent  technical difficulties introduced by $\pM$. \gcambios{The techniques also show that $\M$ is isometric to the closure of some open subset in a globally hyperbolic spacetime (without boundary).
As a trivial} consequence, the interior $M$
  both
  splits orthogonally and can be embedded isometrically in \gcambios{some} $\LL^N$,   extending so properties of globally spacetimes without boundary  to a  class of causally continuous ones.
\end{abstract}


\tableofcontents

\section{Introduction}
The celebrated results by Choquet-Bruhat \cite{CB} and Choquet-Bruhat and Geroch \cite{CBG} ensure that the metric of a globally hyperbolic \ncambios{(say, Ricci-flat)} spacetime $(M,g)$ is determined by Cauchy data on a Cauchy hypersurface $\Sigma$. Such data (a Riemannian  metric, the {\em a posteriori} second fundamental form) are subject to some constraints  (Gauss and Codazzi); \gcambios{when the Lorentzian metric $g$ is prescribed, further results can be obtained for other fields satisfying hyperbolic equations.}
The existence of a smooth and spacelike Cauchy hypersurface $\Sigma$ for any such $(M,g)$ ensures the well-posedness of the Cauchy problem; moreover, the existence of a splitting for the full spacetime $M$ as a orthogonal product $(\R\times \Sigma, g=-\Lambda d\tau^2 +g_\tau)$, where $\tau$ is a Cauchy temporal function, shows the consistency of the notion of predictability from each ``instant'' of time to another, \gcambios{ apart from other advantages. } These two existence results where obtained at a topological level by Geroch \cite{geroch} and at the differentiable and metric levels in \cite{BS03, BernalSanchez}; a full and self-contained development of the Cauchy problem has been carried out recently in the painstaking book by Ringstr\"om \cite{ringstrom}.
In the present article,  such a splitting as well as other causal properties, will be extended to the case of \ncambios{{\em globally hyperbolic spacetimes-with {\em (smooth, conformal)} -timelike-boundary}.}  The Cauchy problem in this class of spacetimes was studied by Friedrich and Nagy \cite{Friedrich} and the class of spacetimes was systematically   studied in Sol\'{\i}s' thesis \cite{didier}. Globally hyperbolic spacetimes
(without boundary) can be characterized as those strongly causal spacetimes whose intrinsic {\em causal boundary} points  are not naked singularities. In a natural way, the timelike boundary $\pM$ of
our class of spacetimes is composed by all the  naked singularities (see the \gcambios{Appendix A});   so,  the  these singularities  become the natural place to  impose boundary conditions.  Chrusciel, Galloway and Solís \cite{kaluzaklein} proved a version of topological censorship for this class of spacetimes; notably, such a class includes asymptotically anti-de Sitter ones.

More precisely, we will show that any globally hyperbolic spacetime-with-timelike-boundary $(\M, g)$ admits a Cauchy splitting $\R\times \bar{\Sigma}$ where $\bar{\Sigma}=\Sigma \cup \partial \Sigma$ is a Cauchy hypersurface-with-boundary $\partial \Sigma$.
From the PDE viewpoint, this suggests that the Cauchy problem will be consistently well-posed in this class of spacetimes as a {\em mixed} Cauchy problem for a  symmetric positive linear  differential equation \cite{Friedrich, friedrichs}.   Now, not only the Cauchy data on  $\bar{\Sigma}$ $(\equiv \{0\} \times \bar{\Sigma})$  but also boundary data on  $\partial M \equiv \R \times \partial \Sigma$  (under suitable compatibility constraints at $\{0\} \times \partial \Sigma$) must be provided, 
resembling  the behavior of the elementary wave equation.  
Background on this mixed problems has been developed recently by Valiente-Kroon
and Carranza \cite{valiente, valiente2, valiente3}; see also  the article by Enciso and Kamran \cite{encisokamran} (including its expanded version \cite{encisokamran2}).
Applications to wave equations and quantum field theory on curved spacetimes were obtained by Lupo  \cite{lupo, lupophd} in the general case and by Dappiagi, Drago and Ferreira \cite{DDF} in the static case, so extending works by B\"ar, Ginoux and Pf\"affle \cite{Bar}, among others.

From the technical viewpoint, the splitting will be obtained by constructing a Cauchy temporal function $\tau$ such that $\nabla \tau$ is tangent
to the boundary (everywhere on $\partial M$). This technical condition will turn out  essential to obtain a $\pM$-orthogonal splitting by flowing through the integral curves of   a suitable normalization of the
gradient  $\nabla \tau$. Such a  $\tau$ will be constructed by the following procedure: (a)  check the $C^0$-stability of Cauchy temporal functions in the set of all the Lorentzian metrics on $\M$  (this will make irrelevant  any $C^2$ details  of $g$ on $\pM$),  (b) reduce the problem to the simplified case of a metric $g^*$ with \gcambios{a product} structure close to $\pM$, (c)  extend  symmetrically $g^*$  to the  {\em double manifold} $\M^d$ (which becomes then a globally hyperbolic spacetime without boundary)  and (d)~obtain a Cauchy temporal  function for  $\M^d$ which is invariant by a reflection  on $\pM$. This last problem  can be regarded as a simple
 case of a result by  M\"uller \cite{muller2016} about Cauchy temporal functions invariant by a compact conformal group; the question (a) will be studied in Section \ref{s_Stab} and the others in Section \ref{s4}.

For comparisons with the case without boundary, we will consider the  approach by using locally defined smooth temporal functions in the original papers \cite{BS03, BernalSanchez}, which has shown to be very flexible for a variety of questions \cite{BS06, MS, muller2016}.
However, some different smoothability procedures with many other applications have been developed  since then, namely: Fathi and Siconolfi \cite{fathi}, using methods inspired by weak-KAM theory, applicable to cone structures; Chrusciel, Grant and Minguzzi \cite{CGM},
inspired by Seifert's approach to smoothability in spacetimes \cite{seifert}; and Bernard and Suhr \cite{bernardshur}, inspired by Conley theory,  applicable to possibly non-continuous closed cone structures. Similar conclusions seem to hold if any of  these  alternative approaches were considered.

Summing up, quite a few  properties about causality and the causal ladder will be revisited for 
spacetimes-with-timelike-boundary, with the following main aim :

\begin{thm}\label{t0} Any globally hyperbolic $n$ spacetime-with-timelike-boundary $(\M,g)$ \ncambios{(see Defn. \ref{defn})} admits a  Cauchy temporal function $\tau$ whose gradient $\nabla \tau$ is tangent to $\partial M$.

As consequence, $\M$ splits smoothly as a product $\R\times \bSigma$, where $\bSigma$ is a  ($n-1)$ \ncambios{manifold-with-boundary}, the metric can be written (with natural abuses of notation) as a parametrized orthogonal product

\begin{equation}
\label{e_split}
g= -\Lambda d\tau^2 + g_\tau \, ,
\end{equation}
where $\Lambda: \R \times \bSigma \rightarrow \R$ is a positive function,  $g_\tau$ is a Riemannian metric on each slice $
\{\tau\} \times \bSigma$ varying smoothly with $\tau$, and these slices are spacelike Cauchy \ncambios{hypersurfaces-with-boundary.}
Moreover, $\M$ can be isometrically embedded in  Lorentz-Minkowski  $\LL^N$ for some $N\in \N$, the interior $M$ of  $\M$ is always causally continuous, the boundary $\pM$ is a (possibly non-connected) globally hyperbolic spacetime without boundary, and:

(a) the restriction of $\tau$ to $M$  extends the known orthogonal splitting of globally hyperbolic spacetimes to this class of causally continuous spacetimes without boundary,

(b) the restriction of $\tau$ to $\pM$ provides a  Cauchy temporal function 
for the boundary whose levels are acausal in\footnote{\label{f0} The global hyperbolicity of $\M$ implies directly the global hyperbolicity of $\pM$ and, so, the existence of a Cauchy temporal function in each connected component of $\pM$.
However,
the global splitting ensures that the different connected components   can be syncronized in an acausal one when looked at $\M$. Such a property becomes natural from the PDE
viewpoint commented above (notice that
different connected components of $\pM$ 
are permitted, even if $\M$ will be assumed connected with no loss of generality).} $\M$.
\end{thm}
It is worth pointing out that all the procedures to be used will be conformally invariant, so,
the \gcambios{results } will be also applicable to the case of a conformal boundary.  Moreover,  the splitting of those spacetimes without boundary which can be seen as the interior $M$ of a globally hyperbolic one with boundary $\M$, has its own interest. In particular, it permits to extend
results about linking and causality by Chernov and coworkers to a bigger class of spacetimes, recall \cite{Ch}.

The paper is organized as follows.

In Section \ref{s2}, some basic preliminaries are introduced. Many of them are known or expected  from standard techniques; anyway, they will become relevant later. In Subsect.~\ref{s2.1},  the  double manifold is used for extensions of  \ncambios{manifolds-with-boundary}, and relations between the
time-orientations of $\M, \pM$ and $M$ are pointed out. Gaussian coordinates are introduced and shown to yield a local version of the splitting~\eqref{e_split} (Cor.~\ref{c_localsplitting}). In Subsect.~\ref{s2.2}, the causal ladder  is introduced. Our choices of the definitions allow us a reasonably self-contained development. In particular, the basic properties of the lower levels of the ladder (until stably causal) are quickly checked there. \bcambios{In
Subsect.~\ref{s2.3},  causal continuous curves  are discussed,
and an intrinsic approach for its regularity (locally Lipschitz or, more properly for some cases where the reparametrization is taken into account, $H^1$) is introduced following \cite[Appendix A]{CFS}.
 This permits to obtain intrinsically limit curves within the same class of curves (Prop. \ref{p2.11}) and will circumvent subtleties which appear for other notions of causal continuous curves (recall Remark \ref{r_Didier}). By completeness, in Subsect. \ref{s2.4}  this regularity is compared  with the usual piecewise smoothness considered in the literature. Even though this issue does not affect to the chronological relation (Prop. \ref{p_futuresmooth}, Remark \ref{r_futuresmooth}), an example in the Appendix B shows that the causal relation is different depending on the class of curves used. However, a spacetime which is globally hyperbolic by using piecewise smooth curves, it will be also globally hyperbolic by using $H^1$-ones (Prop. \ref{otro}). So, our choice of $H^1$ curves is not only important for technical consistency in proofs, but also because the class of globally hyperbolic spacetimes where our results apply becomes bigger.}

%

In Section \ref{s3}, the framework of Geroch's topological splitting is revisited in order to include   boundaries.
Even though most of the properties here are transplanted from the case without boundary, we make a fast review to emphasize some differences and provide a reasonably self-contained study.
In Subsect. \ref{s3.2}, after checking that the role of admissible measures can be extended to the case with boundary, we reconstruct the higher levels of the causal ladder (Thm. \ref{t_ladder}), determine the causal properties  inherited by the boundary $\pM$  and the interior $M$ at each level, and provide the necessary (counter-) examples (Remark \ref{r_counterexamples}).
 Technically,  Prop. \ref{p_dist} and \ref{p_creer} summarize the  main properties which can be transplanted directly from the case without boundary; for the remainder, short proofs are provided. In Subsect. \ref{s3.3}, we go over Geroch's technique  to find the required Cauchy time function (Thm. \ref{t_geroch}). With this aim, some properties of achronal sets and hypersurfaces with boundary are revisited (recall Defn. \ref{d_achronal} and, then, Prop. \ref{pf}).

 In Section \ref{s_Stab}, the stability of both global hyperbolicity and Cauchy temporal functions is proved in the case with boundary. This question has interest in its own right, and will be used to simplify the proof of the existence of orthogonal splitting in the case with boundary. Recall that, in the case without boundary, the  stability of a globally hyperbolic metric $g$ on $M$ was also studied by Geroch \cite{geroch}.
This question becomes equivalent to showing that there exists a metric $g'$ with strictly  bigger cones (i.e. $g<g'$) which is globally hyperbolic  (as all the metrics $g''$ with $g''\leq g'$ will be globally hyperbolic too). As  Geroch's time function $t$ may be a non-time function for any $g'>g$ (recall that even in the smooth case the levels of $t$ may be degenerate hypersurfaces), this question  was non-trivial at that moment; however, the problem was widely simplified when a Cauchy temporal function $\tau$ was proved to exist (see Section 3 of \cite{BM} in the arxiv version,  and   \cite{S}). Indeed, a stronger result holds because $\tau$ becomes stable as a Cauchy temporal function (i.e. $\tau$ is also Cauchy temporal for some $g'>g$ and, thus, any $g''<g'$), and, then, the stability of global hyperbolicity appears as a direct consequence, as will be checked here (see Remark \ref{r_A}). Anyway, different proofs of the stability of $g$ with interest in its own right were found by Benavides and Minguzzi \cite{BM} and by Fathi and Siconolfi \cite{fathi}. As explained in Subsect. \ref{s_Stab1}, we will prove   stability of  both global hyperbolicity (by means of a direct proof for the sake of completeness) and
 Cauchy temporal functions (assuming that they have been constructed with no restriction on $\pM$ as in \cite{BernalSanchez})\footnote{This was carried out in the unpublished arxiv paper \cite{S}, which is essentially incorporated here.}, i.e.:
 \begin{quote}
{\em  For a spacetime-with-timelike-boundary $(\M, g)$, global hyperbolicity is a stable property. What is more, for any Cauchy temporal function $\tau$ of $(\M, g)$ there exists a globally hyperbolic metric with wider cones $g'>g$ such that $\tau$ is Cauchy temporal for $g'$ and, thus, for any other Lorentz metric (necessarily globally hyperbolic) $g''\leq g'$. }
\end{quote}
  Apart from the interest in its  own,  stability will simplify widely the procedure to obtain the   tangency of $\nabla \tau$ to $\pM$  in Thm. \ref{t0}.
 Indeed,  our procedure will stress that  possible problems associated to, say, the bending of $\pM$ or its non-convexity, will have a ``higher order''  than the requirements for  $\tau$ (and, thus, will be negligible). In Subsect. \ref{s_Stab2} we will see how to construct globally hyperbolic perturbations of $g$ (eventually, maintaining the Cauchy temporal character of a prescribed one $\tau$ for $g$) and in Subsect. \ref{s_Stab3} such perturbations  are shown to yield the stability results.

In Section \ref{s4}, after an overall  explanation of the simplified procedure,
Thm. \ref{t0} is proved. A discussion on the type of problems which can be proven transplanting directly the techniques for the case without boundary is also carried out (Remark \ref{r_final}, \ref{r_A}).
Moreover, a simple application of the techniques shows that every globally hyperbolic spacetime-with-timelike-boundary can be regarded as the closure of an open subset included in a globally hyperbolic spacetime without boundary
(Corollary \ref{c_extension_globhip}).

 Finally, the \gcambios{Appendix A} justifies rigourously that, in a globally hyperbolic spacetime-with-timelike-boundary, the  boundary $\partial M$ is composed by all the naked singularities of the spacetime (Thm. \ref{naked}). As these singularities are regarded naturally as a subset of the intrinsic causal boundary $\partial^cM$ (the pairs $(P,F)$ with $P\neq\emptyset\neq F$, see Remark \ref{nakedsingu}), the essential ingredients of $\partial^cM$ are reviewed first for the sake of completeness. \bcambios{The Appendix B provides the counterexample about  piecewise smooth causal curves commented above. }

\section{Preliminaries on spacetimes-with-timelike-boundary}\label{s2}

\subsection{Generalities: boundaries, time-orientation and coordinates} \label{s2.1}

{\bf \ncambios{Manifolds-with-boundary}.} In what follows $\M$ will denote a {\em connected} $C^r$
{\em  $n$  \ncambios{manifold-with-boundary}},
being
$n\geq 2$.
Any function or tensor field  will be {\em smooth} when it is as
 differentiable as possible (compatible
 with the $C^r$ character of $\M$); along the paper,
 $C^1$ will be enough  for
 the metric $g$, and
 the elements to be obtained (as the Cauchy temporal function) will maintain the maximum differentiability. 
 $\M$ is then  locally diffeomorphic to (open subsets of) a closed half space of $\R^n$;
 $M$ will denote its {\em interior} and $\pM$ its 
 {\em boundary}.
  For any $p\in \M$,   $T_p\M$ will denote
 its $n$-dimensional tangent space
 while for $p\in\pM$,  $T_p\pM$ is the
 $(n-1)$-dimensional tangent space to the boundary.
 Such a $\M$ can be regarded as a closed subset of the so-called {\em double manifold} $\M^d$, a $C^r$ $n$-manifold (without boundary) obtained by taking two copies of $\M$ and identifying  homologous
 boundary points (in particular, partitions of  unity for $\M$ can be constructed from  $\M^d$ using $\M^d\setminus \M$ as an extra open subset).   Lorentzian and Riemannian metrics on $\M$ are particular cases of {\em semi-Riemannian metrics}, i.e. non-degenerate metric tensors (of constant index). 
 Background on \ncambios{manifolds-with-boundary} can be seen in \cite[Section 9]{Lee}, for example. Next, we emphasize a basic property (see \cite[Prop. 3.1]{LuisTesis} for a proof).

\cambios{\begin{prop}\label{pp} Regarding $\M$ as a closed subset of $\M^d$, any semi-Riemannian  metric $g$ on $\M$ can be extended to some open subset $ \widetilde{M} \subset \M^d$ with $\overline{M}\subset \widetilde{M}$.
\end{prop}


\begin{rem} {\em Recall:
(a) in general, the metric  defined on two copies of $(\M,g)$ cannot be extended as a smooth metric on  $\overline{M}^d$,
(b) in the Riemannian case, $g$ can be extended to the whole $\M^d$, but in
 the Lorentzian case this may be non-possible (for example, if $\M$ is an even-dimensional closed half-sphere, $\M^d$ admits no Lorentzian metric).
}\end{rem}}

\bigskip

\noindent {\bf Time-orientation and spacetimes.} Let us recall some basic notions for  spacetimes-with-timelike-boundary. Usual notions for Lorentzian manifolds without boundary such as causal or timelike vectors
(here, following conventions in \cite{O, MSCH}) are extended to the case with boundary with no further mention (see \cite{Galloway, didier} for further background).

\begin{defi}\label{defn}
A {\em Lorentzian manifold with timelike boundary} $(\overline{M},g)$, $\overline{M}=M\cup \partial M$, is a Lorentzian  \ncambios{manifold-with-boundary} such that the pullback $i^{*}g$, with $i:\partial M \hookrightarrow M$ the natural inclusion, defines a Lorentzian
metric on the boundary. A {\em spacetime-with-timelike-boundary} is a time-oriented Lorentzian manifold with timelike boundary.
\end{defi}
By {\em time-oriented} we mean that a time cone has been chosen continuously (i.e., locally selected by a continuous timelike vector field $X$) on all $\M$. \cambios{The pull-back $i^*$ will be dropped} when there is no possibility of confusion \cambios{and the time-orientation is assumed implicitly. If $g,g'$ are two Lorentzian metrics on $\M$, the notation $g<g'$ (resp. $g\leq g'$) means that any future-directed causal vector for $g$ is future-directed timelike (resp. causal) for $g'$.}
The following result ensures that no additional subtleties on time-orientations appear because of the presence of the boundary. Its proof uses standard background for the case  without boundary, see  \cite[Lemma 5.32, Prop. 5.37]{O}.

\begin{prop}
\label{extendfield} The following properties are equivalent for any Lorentzian manifold with timelike boundary $(\overline{M},g)$:

(i) $(\overline{M},g)$  is time-orientable.

(ii) $(M,g)$ and $(\pM,g)$ are time-orientable.

(iii) There exists a timelike vector field $T$ on all $\M$ tangent to $T_{\hp}\partial M$ at each $\hp\in\pM$.

(iv) There exists a timelike vector field $T$ on all $\M$.

\smallskip

\noindent Therefore, for any spacetime-with-timelike-boundary, $\pM$ is naturally a spacetime \ncambios{(without boundary and possibly non-connected, consistently with footnote \ref{f0}).}

\end{prop}

\noindent {\em Proof.} (i) $\Rightarrow$ (ii) Notice that if $X$ selects the time-orientation on some neighborhood $U\subset \M$, then its orthogonal projection on $T_{\hat{p}}\partial M$ for each $\hat{p}\in U\cap \partial M$ 
selects continuously a time orientation on $U\cap \partial M$.



(ii) $\Rightarrow$ (iii) As $(M,g)$ has no boundary, it admits a smooth timelike vector field $T^M$, and each
connected part $C$ of $\pM$ will admit also a smooth timelike vector field
$T^C$. For each $\hp \in C$, consider a coordinate chart $(U_{\hp},(x^{0},x^{1},...,s))$ adapted to the boundary, i.e. $s^{-1}(0)=U_{\hp} \cap C$, and extend the (restricted) vector field $T^{C} \mid_{U_{\hp} \cap C}$ to the coordinate chart $U_{\hp}$ by making the components of the vector field independent of the $s$ coordinate. Let $T^C[\hat p]$ be such an extension.
As the set of points $\hat{p}\in C$ for which the time orientation determined by $T^C[\hp]$ and $T^M$ agree on $U_{\hat{p}}\cap M$ is both open
and closed in $C$, we can choose $T^C$ so that both agree for all $\hp\in C$.
Repeating this for all the connected components of $\pM$,  considering
the covering of $\M$ provided by $M$ and all $U_{\hp}, \hp \in \pM$, and taking a
partition of  unity subordinate to this covering, one gets  a timelike
vector field $T^0$ defined on some neighbourhood $U$ of $\pM$ which is also tangent
to $\pM$.
So, if $\{\mu, 1-\mu\}$ is a partition of
the unity of $\overline{M}$ subordinate to the
covering $\{U,M\}$, the required vector
field is just $T=\mu T^0+(1-\mu) T^M$.

The  implications (iii) $\Rightarrow$ (iv),  (iv) $\Rightarrow$ (i) and the last assertion are trivial.
 \cvd



\bigskip

\noindent {\bf Gaussian coordinates.} The following coordinates specially well adapted  to the boundary will be useful.
Let $\hat{p} \in \partial M$ and  take a chart in the boundary
$(\hat{U},x^{0},x^{1},\ldots,x^{n-2})$, with $\hat{U} \subset \partial M$ connected and relatively compact, satisfying
 $g(\partial_{0},\partial_{0})=-1$,   $\partial_{0}$  future-directed  on  $\hat U$,
and  $\{\partial_{0},\ldots,\partial_{n-2}\}$ an orthonormal basis at $T_{\hat{p}}\partial M$.  Since $\partial M$ is timelike, there exists a unitary
spacelike vector field $N$
on $\pM$ which is orthogonal to the boundary \cambios{and
points out into $M$.} Extend the previous coordinate system to a chart of $\hat p$ in $\M$ by using the geodesics with
initial data $(q,N_{q})$, $q \in \hat{U}$, that is, consider the geodesic $\gamma_{q}(s)=\exp_{q}(s\cdot N_{q})$, $s\geq 0$, and regard  its affine parameter $s$ as a transverse coordinate. This
provides the required coordinate system $(\hat{U} \times  [0,s_+ ), (x^{0},x^{1},\ldots,x^{n-2}, x^{n-1}=s))$ of $\overline{M}$ for some $s_+>0$ small enough.
Since $\partial_{s}|_{\pM}=N$ is orthogonal to the boundary, $g(\partial_{s},\partial_{j})=0$ on $U:=\hat{U} \times [0,s_+)$ for all $j=0,\ldots,n-2$ (see for example \cite[pp. 42-43
]{Wald}), and the metric $g$ can be written as
\begin{equation}\label{e_Gaussian}
g=\cambios{\sum_{i,j=0}^{n-2}} g_{ij}(\hat x ,s)dx^{i}dx^{j}+ds^{2},
 \qquad \cambios{\hbox{where} \, \hat x=(x^{0},x^{1},\ldots,x^{n-2}).}
\end{equation}
  Moreover, $\hat U$ and $s_+$ are taken small enough so that   the gradient $\nabla x^0$ is timelike (i.e., the slices of $x^0$ are spacelike) on $U$.  Any coordinate system constructed as above will be called a  {\em Gaussian chart adapted to the boundary}, or just {\em Gaussian coordinates}. When necessary, the image of the coordinates will be a {\em cube}, that is, $(-\epsilon,\epsilon)^{n-1}\times [0,\epsilon)$ for some $\epsilon>0$. When $p\in M$, the name {\em Gaussian coordinates} will refer just a normal neighborhood of $p$, and the term {\em cube} to the subset $(-\epsilon,\epsilon)^{n}$.
 As a simple consequence, the local version  of the desired splitting for globally hyperbolic spacetimes follows:
\begin{cor}
\label{cor}\label{c_localsplitting} For each
 $\hat{p}\in\partial M$
there exists a product neighborhood
$V=(-\epsilon, \epsilon ) \times \bar{V}_0$, where $\bar{V}_0$ is a spacelike embedded hypersurface with boundary,
such that both factors are $g$-orthogonal and $g$ can be written as a parametrized product
\[
g=-  \Lambda  d\tau^2+g_{\tau},\qquad \tau: (-\epsilon, \epsilon ) \times \bar{V}_{0} \rightarrow(-\epsilon, \epsilon ) \; \; \hbox{(natural projection)},
\]
where $\Lambda=-1/g(\nabla \tau,\nabla \tau)$ is a function on $V$ and $g_{\tau}$ is a Riemannian metric on  $\{\tau\} \times \bar{V}_{0}$ depending smoothly on $\tau$.
\end{cor}
\begin{proof}
 Take any Gaussian coordinates of $\hat p$, put $\tau:=x^0$,  $\bar V_0=\tau^{-1}(0)$ and recall that, by  \eqref{e_Gaussian}, $\nabla \tau$ must be tangent to $\pM$ on the boundary. So (taken a smaller hypersurface $\bar V_0$ and  neighborhood)  the flow of the vector field
 $-\nabla \tau/|\nabla \tau|^{2}$ is well defined and moves $\bar V_0$ yielding the product neighboorhod $V=(-\epsilon,\epsilon)\times \bar{V}_0$ (the expression of the metric becomes  then standard, see the end of the proof of Prop.~ 2.4 in \cite{BernalSanchez}).
 \end{proof}

\subsection{Conditions on causality and lower levels of the causal ladder}\label{s2.2}
For any spacetime-with-timelike-boundary $(\M,g)$, the usual notation $\ll, \leq, I^\pm(p), J^\pm(p), $ will be used for the chronological and  causal relations and the chronological and causal future/past of any $p\in \M$; so, say, $I^+(p,U)$  will denote the chronological future obtained by using curves entirely contained in the subset $U\subset \M$.

There is, however, a subtlety regarding the degree  differentiability of the future-directed and past-directed causal curves necessary to compute $J^\pm(p)$. They will not be necessarily piecewise smooth but just locally Lipschitz up to a reparametrization (or $H^1$ taking into account the reparametrization); such curves must be  differentiable almost everywhere and the derivative must be  either future-directed causal or past-directed causal whenever it exists. This question  will be discussed in more detail in Section \ref{s2.3} and stated explicitly as Convention \ref{convention}.

In what follows $cl$ will denote closure.

\begin{prop}\label{p_opentrans} (a) The binary
relation $\ll$ is open (in particular, $I^\pm(p)$ are open in $\M$). (b) For any $p,q,r\in \M$, $p\ll q\leq  r \Rightarrow p\ll r$,
$p\leq q\ll  r \Rightarrow p\ll r$. (c)~$J^\pm(p)\subset cl(I^\pm(p))$. (d)  $I^{\pm}(p,M)=I^{\pm}(p) \cap M$ for all $p \in M$.
\end{prop}
\noindent {\em Proof.} Properties (a), (b), (c) are easy to check (see \cite[Prop. 3.5, 3.6, 3.7]{didier}). To prove $I^{+}(p,M)=I^{+}(p)\cap M$, the inclusion $\subset$ is trivial. So, let $q \in I^{+}(p) \cap M$, and take some (piecewise smooth) future-directed timelike $\gamma:[0,1] \rightarrow \overline{M}$ with $\gamma(0)=p, \gamma(1)=q \in M$.
Consider any vector field $N \in \mathfrak{X}(\overline{M})$ which extends the pointing-inward unit normal on $\partial M$  (this can be always done,
as $\pM$ is closed), any
smooth function $f:[0,1] \rightarrow \mathbb{R}^{+}$ vanishing only at $0,1$, and the vector field $V$ on $\gamma$ defined by $V(t)=f(t)N_{\gamma(t)}$ for all $t\in [0,1]$. For the fixed-endpoints variation of $\gamma$ corresponding to the variational vector $V$, longitudinal curves close to $\gamma$ are still timelike and cannot touch $\pM$. In conclusion, $q \in I^{+}(p,M)$.
\cvd

\begin{rem}\label{nuevo} {\em Property (d)  can be naturally extended to points at the boundary  as follows: for any $p\in\overline{M}$, $I^{\pm}(p)\cap M$ is the set of
$q\in M$ such that there exists a  future/past -directed timelike $\gamma$ with ${\rm Im}(\gamma)\setminus\{p\}\subset M$ joining $p$ with $q$. }
\end{rem}

In the case of spacetimes without boundary, there is a well-known causal ladder of spacetimes, each step admitting several characterizations (see \cite{Beem, MSCH}). Most of the ladder and characterizations can be transplanted directly to the case of spacetime-with-timelike-boundary. 
Here, we will focus just on  globally hyperbolic spacetimes, postponing  the systematic study of other causal subtleties for future work. So, we will make a fast summary of the standard steps of the ladder just making  simple choices on the definitions and
 properties to be used in a self-contained way.

\begin{defi}
 A spacetime-with-timelike-boundary $(\overline{M},g)$
is:
\bit\item
{\em chronological} (resp.
 {\em causal}) if it does not contain  closed timelike (resp.
causal) curves,
\item {\em future} (resp.  {\em past})  {\em distinguishing}.  if the equality  $I^{+}(p)=I^{+}(q)$ (resp. $I^{-}(p)=I^{-}(q)$) implies $p=q$, that is, if the set-valued map $I^+:\overline{M} \rightarrow P(\overline{M})$ (resp. $I^-:\overline{M} \rightarrow P(\overline{M})$), where $P(\overline{M})$ is the power set of $\overline{M}$,  is one-to-one. It is {\em distinguishing}, when it is both future and past distinguishing.
\item {\em strongly causal} if for all $p\in \M$ and  any neighborhood $U\ni p$ there exists another neighborhood $V \subset U$, $p\in V$, such that any causal curve with endpoints in
$V$ is entirely contained in $U$.
\eit
\end{defi}


\begin{defi} A subset $W$ of a spacetime-with-timelike-boundary $(\overline{M},g)$ is
{\em  causally convex 
 } if $J^{+}(x)\cap J^{-}(y)
\subset W$ for any $x,y \in W$ (equivalently, if any causal curve with endpoints in $W$ must remain in $W$).


Given an open neighborhood $U\subset \M$, a subset $W\subset U$ is {\em causally convex}
{\em in $U$} when $W$ is causally convex  
 as a subset of  $U$, regarding $U$ as a spacetime-with-timelike-boundary.
\end{defi}

\begin{lemma}\label{cc} Let $(\overline{M},g)$ be a spacetime-with-timelike-boundary. For any $p\in \M$ and any neighborhood $V\ni p$ there exists a sequence of nested neighborhoods $\{W_m\} \subset V$, $W_{m+1}\subset W_m$, $\{p\}=\cap_m W_m$, such that all $W_m$ are causally convex in $V$.
\end{lemma}

\begin{proof}  This can be proved as in the case without boundary \cite[Thm. 2.14, Lemma 2.13]{MSCH} (now requiring the nested neighborhoods just to be causally convex instead of globally hyperbolic). Namely, given $V$,  one takes (Gaussian) coordinates centered at $p$, $(V',x^i)$, $V'\subset V$,  a standard flat metric $g^+$ in these coordinates with $g<g^+$ (say, $g^+= -\epsilon (dx^0)^2+\sum_i (dx^i)^2 +ds^2$, with $\epsilon >0$ small and, eventually, choosing a smaller $V'$) and
$W_m:= I^+(x^0=-1/m, 0,\dots,0)\cap
I^-(x^0=1/m, 0,\dots,0)$ for large $m$. \end{proof}

\ncambios{
\begin{prop} \label{imprisoned}
$(\overline{M},g)$ is  strongly causal  if and only if each $p\in\M$ admits arbitrarily small causally convex
 neighbourhoods, that is,  for any neighbourhood
$U\ni p$ there exists a causally convex
neighbourhood $W\ni p$ contained in $U$.
\end{prop}
}

\begin{proof} We will focus on the the implication to the right (to the left is trivial). Since $(\overline{M},g)$ is strongly causal, given an open neighborhood $U$ of $p$, there exists a smaller neighborhood $V\subset U$, $p\in V$, such that any closed causal curve with extreme points in $V$ is totally contained in $U$. From Lemma \ref{cc}, there exists some neighborhood $W\subset V$ of $p$ which is causally convex in $V$. The property above satisfied by $V$ ensures that $W$ must be causally convex (in $\overline{M}$) as well.
%
%
%
%
\end{proof}
\ncambios{Then, following the case without boundary
(see \cite[Prop. 3.13]{Beem} or \cite[Lemma 14.13]{O}; see also \cite[Sect. 3.6.2]{MSCH}), full details are written in \cite[Prop. 3.7]{LuisTesis}, we deduce:	
	
\begin{cor}\label{imprisonedb} If $(\overline{M},g)$ is strongly causal then
causal curves are not partially imprisoned on compact sets, that is,  for any  future-directed causal curve $\gamma:[a,b) \rightarrow \overline{M}, a<b\leq \infty$ which cannot be extended continuously to $b$, and any compact set $K\subset \overline{M}$, there exists some $s_0\in [a,b)$ such that $\gamma(s)\not\in K$ for all $s\geq s_0$.
\end{cor}
}


\begin{prop}\label{p_ord_lowerlevels} In any spacetime-with-timelike-boundary: \ncambios{strongly causal implies distinguishing; moreover,   future or past distinguishing implies causal, which in turn implies chronological.}
\end{prop}
\begin{proof}
The first implication  follows as  in the case without boundary \cite[Lemma 3.10]{MSCH} (see \cite[Prop. 3.7(a)]{LuisTesis} for full details); this also happens for  the second one (a contradiction follows easily by applying the transitivity relations in Prop. \ref{p_opentrans}), and the last one is trivial.
\end{proof}


Following the ladder,
$(\M,g)$ is {\em stably causal}, when it admits a  {\em time function}, i.e., a continuous function $\tau$ which increases strictly on all future-directed causal curves. This step admits some classical characterizations which are stated next only for the sake of completeness. \ncambios{The equivalence among these characterizations can be done transplanting the techniques in the case without boundary, and the proof will be sketched in Subsection \ref{s4.1}, just after Remarks \ref{r_final}, \ref{r_A}.}
\begin{prop}\label{p_stable}
For a spacetime-with-timelike-boundary $(\overline{M},g)$ they are equivalent:
\ben
\item It admits a time function (i.e., the spacetime is stably causal).
\item It admits a temporal function $\tau$ (i.e., $\tau$ is smooth with timelike past-directed $\nabla \tau$).
\item There exists a strongly causal metric $g'$ with $g<g'$.
\een
In this case, the  spacetime-with-timelike-boundary is also  strongly causal.
\end{prop}

The higher levels of the ladder (related to Geroch's proof of the splitting) will be revisited in  Section \ref{s3}.  Its definitions (as optimized in \cite{BS07, MSCH} for the case without boundary) \ncambios{are given hereunder. First, recall that the natural topology in $P(\overline{M})$ is the one admitting as a basis the sets $\{U_K: K\subset \overline{M}$ is compact$\}$, where $U_K=\{A\subset \overline{M}: A\cap K=\emptyset\}$, see \cite[Def. 3.37 to Prop. 3.38]{MSCH}.}

\begin{defi} A spacetime-with-timelike-boundary $(\M,g)$ is:
\bit
\item  \ncambios{{\em causally continuous}, when the set valued functions $I^{\pm}:\overline{M} \rightarrow P(\overline{M})$ are both one to one (that is, the spacetime is distinguishing) and continuous (for the natural topology in
$P(\overline{M})$);}
\item {\em causally simple}, when it is causal and all $J^+(p), J^-(p), p\in \M$ are closed;
\item {\em globally hyperbolic}, when it is  causal and all $J^+(p)\cap J^-(q)$, $p,q\in \M$ are compact.
\eit
\end{defi}

\subsection{Continuous vs Lipschitz/$H^1$ causal curves}\label{s2.3}

Even though the basic definitions in Lorentzian Geometry are carried out with
smooth elements (in particular, causal curves are
regarded as piecewise smooth),
{\em continuous causal} curves are required for relevant purposes.
Indeed, a key result
is the {\em limit curve theorem} \cite[Prop. 3.31]{Beem} which, under some hypotheses, ensures the existence of a limit curve to a sequence of causal ones, being the limit only  continuous causal  (even if the  causal curves in the sequence are smooth).
In the case of distinguishing spacetimes (without boundary),
a continuous future-directed causal curve $\gamma: I\subset \R \rightarrow M$ is any continuous curve that preserves the causal relation, that is, satisfying: $t,t'\in I$ and $t<t'$ implies
$\gamma(t)<\gamma(t')$, see \cite[Prop. 3.19]{MS}; for non-distinguishing spacetimes, this property is required to be satisfied locally in arbitrarily small neighbourhoods, see \cite[Sect. 3.5]{MS}.
Continuous causal curves in a spacetime (without boundary) are known to satisfy a locally Lipschitz condition. This condition allows us to identify these curves (when conveniently reparametrized) with $H^1$ curves.
\bcambios{We prefer to  speak about  $H^1$ curves because, on the one hand,  parametrized curves are considered along this paper and, on the other,  $H^1$  is the natural regularity for the convergence of  curves, not necessarily causal,  in other problems,  as 
when spacelike geodesics are studied \cite{CFS, Masiello}.}
More precisely, recall that  a 
curve $\gamma:J \rightarrow \R^n$ defined on a compact interval $J$  is  $H^1$ when it is absolutely continuous (equally,  it satisfies both  differentiability almost everywhere and the Fundamental Theorem of Calculus, $  \gamma(t)=\gamma(0)+  \int_{0}^{t} \gamma'(s)ds$, $t\in J$) and $\gamma'$ is $L^2$ integrable; the set of all such curves is the {\em Sobolev space}
$H^{1}(J,\mathbb{R}^{n})$. \bcambios{Clearly, Lipschitz curves are $H^1$ and, conversely,}
 any $H^1$ curve $\gamma$ with  $|\gamma'|$ bounded a.e. by a constant is Lipschitz; \bcambios{so}, both conditions will be interchangeable for continuous causal curves. The following is
  a natural extension to  any interval $I$ and \ncambios{manifold-with-boundary} (see \cite[\S 3.1.3]{LuisTesis} for further details).

\begin{defi}\label{d_contcausalcurve}
Let $\M$ be an $n$  \ncambios{manifold-with-boundary} and $I\subset \R$ any interval. A continuous curve $\gamma:I \rightarrow \M$ is a
$H^1$-{\em curve} if, for any local chart $(U,\varphi)$,
$\varphi \circ \gamma\mid_{J}$ belongs to 
$H^{1}(J,\mathbb{R}^{n})$ for all compact  intervals  $J\subset \gamma^{-1}(U)$. The space of $H^1$-curves from $I$ to $\M$ will be denoted  $H^{1}(I,\M)$.

In the case that $(\M,g)$ is a spacetime-with-timelike-boundary, $\gamma$ is called {\em future (resp. past) -directed $H^1$-causal} if it is $H^1$ and its a.e. derivative is future (resp. past) -directed $H^1$-causal; $\gamma$ is {\em $H^1$-causal} if it is either future or past-directed $H^1$-causal.
\end{defi}
For manifolds without boundary, it is proven in \cite[Appendix A]{CFS} that a curve is $H^1$-causal if and only if it is continuous causal (in the sense described above) up to a reparametrization; moreover,
from the proof it is also clear that the reparametrization can be always carried out locally by using any temporal function (in this case, $\gamma$ can be
regarded as a Lipschitz function, according to the discussion above Defn. \ref{d_contcausalcurve}).
\ncambios{Recall that a curve $\gamma$ is a {\em limit curve} of a sequence of curves $\{\gamma_n\}$ on $M$ if there is a subsequence $\{\gamma_m\}$ such that for all $p\in$ Im$( \gamma)$, each neighborhood of $p$
intersects all but finitely many of the $\gamma_m$’s.}
Then, the classical theorem of limit curves is also valid in the framework of spacetimes-with-timelike-boundary and $H^{1}$-causal curves, see \cite[Prop. 3.31]{Beem} and \cite[Lemma 3.23]{didier}, namely, just extending  $(\M,g)$ to a manifold without boundary (Prop. \ref{pp}) and applying the results in this case (see \cite[Prop. 3.16]{LuisTesis} for details), i.e.:

\begin{prop}\label{p2.11}
Let $(\M,g)$ be a spacetime-with-timelike-boundary and $\{\gamma_m\}_m$ a sequence of future-directed inextensible $H^1$-causal curves. If $p \in \M$ is an accumulation point \ncambios{of $\{\gamma_m\}_m$}, then there exists a limit curve $\gamma$ of
$\{\gamma_{m}\}_{m}$ which is a future-directed inextensible $H^1$-causal curve  which crosses $p$.
\end{prop}

\begin{rem}\label{r_Didier} {\em For any strongly causal spacetime-with-timelike-boundary the following alternative notion of continuous causal curve was introduced in Solis' Ph.D. Thesis, see \cite[Def. 3.19 and below]{didier}:
a continuous curve $\gamma:I \rightarrow \M$ is {\em future-directed causal} if for any $t_{0} \in I$ there exists a $\widetilde{M}$-convex neighbourhood $U_{0}$ around $\gamma(t_{0})$ (where $(\widetilde{M},\widetilde{g})$ is a spacetime without boundary that extends $(\M,g)$, see Prop. \ref{pp}) and an interval $[a,b] \subset I$ such that for all $s,t \in [a,b]$ with $s \leq t_{0} \leq t$, one has $\gamma(t) \in J^{+}(\gamma(s),U_{0} \cap \M)$.
 This definition is shown to be independent of the chosen $(\widetilde{M},\widetilde{g})$, and causality relations in $(\M,g)$ defined by using such continuous curves become  equivalent to the classical ones with piecewise smooth ones, see \cite[Remarks 3.20 and 3.21]{didier}. However,
it is not obvious from the proof in  \cite[Lemma 3.23]{didier}, whether the limit curve theorem for such curves
will yield a limit curve which is also continuous causal according to previous definition. \bcambios{Indeed,  Appendix B shows a
spacetime with ($C^\infty$) timelike boundary containing two points which can be connected by means of a $H^1$-causal curve but not by a piecewise smooth one.}  
}\end{rem}


\bcambios{Such problems will be circumvented here by assuming  that the regularity of the  causal curves is $H^1$, and computing the causal futures and past consistently with these curves (Convention \ref{convention} below).
Indeed, reasoning  as in \cite{didier} with $H^1$ curves, we deduce the following key proposition, which is an extension to the case with boundary of the corresponding classical limit curves results in \cite[Prop. 3.34, Cor. 3.32]{Beem}.}

\begin{prop}\label{otro} (1) Let $(\overline{M},g)$ be a strongly causal spacetime-with-timelike-boundary. Suppose that $\{\gamma_n\}$ is a sequence of causal curves defined on $[a,b]$ such that $\gamma_n(a)\rightarrow p$, $\gamma_n(b)\rightarrow q$. A causal curve $\gamma:[a,b]\rightarrow \overline{M}$ with $\gamma(a)=p$ and $\gamma(b)=q$ is a limit curve of $\{\gamma_n\}$ \bcambios{(according to the definition recalled above Prop. \ref{p2.11})} iff there is a subsequence $\{\gamma_m\}$ of $\{\gamma_n\}$ which converges to $\gamma$ in the $C^0$ topology on curves.

(2) Let $(\overline{M},g)$ be a globally hyperbolic spacetime-with-timelike-boundary. Suppose that $\{p_n\}$ and $\{q_n\}$ are sequences in $\overline{M}$ converging to $p$ and $q$ in $\overline{M}$, resp., with $p\neq q$, and $p_n\leq q_n$ for each $n$. Let $\gamma_n$ be a future-directed causal curve from $p_n$ to $q_n$ for each $n$. Then there exists a future-directed causal limit curve $\gamma$ which joins $p$ to $q$.
\end{prop}

\subsection{Piecewise smooth vs $H^1$ causal curves}\label{s2.4}

\bcambios{The example in the Appendix B shows that the causal future (and, then, the notion of global hyperbolicity) is, in general, different, when computed either with piecewise smooth or with  $H^1$-curves. However, in this subsection (which is independent of the remainder up to the final Convention \ref{convention}), we will check:}
\begin{enumerate} \item \label{i1} \bcambios{ The chronological futures and pasts are equal when computed with piecewise smooth timelike curves or with $H^1$-curves which are  timelike in any reasonable sense  (Prop. \ref{p_futuresmooth}, Remark \ref{r_futuresmooth}). }
\item \label{i2} Any spacetime which is globally hyperbolic when its causal futures and pasts  are computed by means of piecewise smooth curves,  it is also globally hyperbolic  when they are computed with $H^1$ ones (Prop. \ref{l}), showing Appendix B a counterexample to the converse.
\end{enumerate}
\bcambios{Let us start with item \ref{i1}.
\begin{prop}\label{p_futuresmooth}
Let  $(\overline{M},g)$ be a spacetime-with-timelike-boundary, $p,q\in \M$ and $\gamma: [\tau_0,\tau_1] \rightarrow \overline{M}$ a future-directed  $H^1$-causal curve from $p$ to $q$ which satisfies that, whenever it is parametrized by a local temporal function $\tau$, there exists some $\nu>0$ such that
 $g(\gamma',\gamma')<-\nu$ a.e. Then, there exists a  future-directed timelike piecewise smooth curve from $p$ to $q$.
\end{prop}}
\begin{proof} Consider first the local case around $p$,
i.e.,  $\gamma$ is included in  some convenient coordinate neighborhood of $p$. We will
 focus on the case $p\in \partial M$, and our proof  will hold in   the case $p\in M$ with no modification (however,  the result is known in the case without boundary, where normal and convex neighborhoods can be used, see \gcambios{ \cite[Sect. 3.2]{Beem} \cite[Lemma 14.2]{O}, \cite[Appendix A]{CFS}}).  Take
$V=(-\epsilon,\epsilon) \times \bar V_{0}$
 as in Cor. \ref{cor}, being cl$(V)$  compact and included in a Gaussian chart with product coordinates
chart $(\tau,y^{1},..., y^{n-2}, y^{n-1}\equiv s) $  defined on a cube
(with
$\{s=0\}$  equal to $V\cap \partial M$).  As causality is conformally invariant, we can also assume   $g=-d\tau^{2}+g_{\tau}$.

 As a first simplification, the problem can been  reduced to the case when 
the metric is just the standard Lorentz-Minkowski one  $\langle\cdot, \cdot \rangle$ in the chosen coordinates.  To check this, narrow slightly the timelike cone at $p$ (that is,  consider the Lorentzian scalar product $(-1+\epsilon) d\tau^{2}+g_{\tau}$ at $p$ for small $\epsilon>0$) and extend it   as a flat Lorentzian metric $\langle\cdot, \cdot \rangle$  on $V$ with constant  coordinate components. Reducing $V$ if necessary, $\gamma$  satisfies $\langle \gamma', \gamma' \rangle < -\nu'$, for some  $\nu'>0$ smaller than $\nu$ ($\nu'$ will be relabelled as $\nu$).

   In order to compute $H^1$-norms and distances, we will consider  the natural Euclidean metrics  in these coordinates for both, $V$ and $\bar V_0$. In what follows, $|\cdot |_0$ will denote both the norm associated with the natural flat Riemannian metric  $\langle\cdot, \cdot \rangle_0$ on $\bar V_0$ (and the usual norm of $\R^{n-1}$).
In these coordinates, let $p=(0,y_0=0)$, $q=(\tau_1,y_1)$ and  put
\begin{equation}\label{e_cgamma}
 \gamma(\tau)=  (\tau,y^{1}_\gamma (\tau), \ldots ,y^{n-2}_\gamma(\tau),  y^{n-1}_\gamma(\tau)= s_\gamma(\tau))  \; \equiv (\tau,y_\gamma (\tau)), \; \tau \in [0,\tau_{1}];
\end{equation}
recall that $y_\gamma$ is  Lipschitz  
  with $|y'_\gamma(\tau)|_0^2 < 1-\nu$, a.e.

  As a second simplification,  it is sufficient  to ensure that $p$ and the points of a sequence $\{q_m\}\rightarrow q$ can be joined by means of smooth timelike curves. The reason is that, choosing some small $\delta\in (0,\tau_1)$ and picking the point $q_\delta$ with $\tau(q_\delta)=\tau_1-\delta$ in the same integral curve of $\partial_\tau$ as $q$, the following claim holds: $\gamma$ can be used to construct another future-directed $H^1$-causal curve $\gamma_\delta$ from $p$ to $q_\delta$ with $\langle \gamma'_\delta,\gamma'_\delta \rangle <-\nu/2$ a.e. Thus, all the points  in some neighborhood  of $q_\delta$ can be connected with
$q$ by means of a smooth timelike curve, and the problem is reduced to connect (by means of smooth timelike curves) $p$ and a sequence  $\{q_m\}\rightarrow q_\delta$ (then, relabel $q_\delta$ as  $q$ and $\nu/2$ as $\nu$). To check the claim, let $\delta:=  (\sqrt{1+\nu/2} -1) \tau_1/\sqrt{1+\nu/2}$ and $ \tau_\delta :=\tau_1-\delta$. Clearly, the  curve
$$\gamma_\delta: [0,\tau_\delta] \longrightarrow V, \qquad \gamma_\delta(\tau):=\left(\tau, y_\gamma \left( \sqrt{1+\nu/2}\, \cdot \tau  \right)\right),$$
connects $p$ and $q_\delta:=(\tau_\delta, y_1)$, and writing
$\gamma_\delta (\tau)=(\tau,
y_\delta (\tau))$ as in \eqref{e_cgamma}:
\begin{equation}\label{e_nu}
\langle \gamma'_\delta, \gamma'_\delta \rangle
= -1 + \langle y'_\delta, y'_\delta \rangle_0
= -1 + (1+\nu/2) \, |y'_\gamma |_0^2 < -1 + (1+\nu/2) (1-\nu)
<-\nu/2  \quad \hbox{a.e.}
\end{equation}

Taking into account the previous two simplifications, the problem can be reduced to the following analytic one. Let $y_\gamma$ as above \eqref{e_cgamma} and  $\epsilon >0$ small, then find some smooth  ($C^{\infty}$ in our coordinates) curve  $x_{\epsilon}$
satisfying:
\begin{equation}\label{e4bis1}
x_{\epsilon}: [0, \tau_1]\longrightarrow  \bar V_0
\quad (\hbox{thus,} \;   s\circ x_{\epsilon} \geq 0),
\quad x_{\epsilon}(0)=0, \quad |x_{\epsilon}(\tau_{1}  )-y_{1}|_0\leq  \epsilon,
\end{equation}
and
 \begin{equation}\label{e4bis11}
 l_\epsilon(\tau_1)
\leq \tau_1,
\qquad \hbox{where} \quad l_\epsilon(\tau)=\int_{0}^{\tau} \sqrt{|x_{\epsilon}'(\bar \tau)|^2_0   + \frac{\nu}{2}  } \; \, d\bar \tau , \quad \forall \tau\in [0,\tau_1].
\end{equation}
  Indeed,  in this case the curve $\gamma_{\epsilon}(\tau)=(l_\epsilon( \tau),x_{\epsilon}(\tau))$ will be
$C^{\infty}$ and timelike  from $p=(0,0)$ to
$q'_\epsilon:=(l_\epsilon(\tau_1), x_{\epsilon}(\tau_1) )$.
As
$q_\epsilon:=(\tau_1, x_{\epsilon}(\tau_1 ))$ lies in the same integral curve of $\partial_\tau$ as $q'_\epsilon$ with $\tau(q'_\epsilon)=l_\epsilon(\tau_1) \leq \tau_1= \tau(q_\epsilon)$, a piecewise smooth timelike curve from $p$ to $q_\epsilon$ also exists. Thus, the result  follows from $\lim_{\epsilon \rightarrow 0} q_{\epsilon} =q$.


To solve this analytic problem, each component $x_\epsilon^i$ of $x_\epsilon$ will be obtained by means of a  covolution
$x_\epsilon^i(\tau)=(y^i_\gamma \star \rho_\epsilon)(\tau)$, where $\rho_\epsilon$ is  some mollifier.
Indeed, put $\epsilon_m:=1/m$ and $x_m:=x_{\epsilon_m}$.
 For each  Lipschitz function
 $y^i_\gamma: [0,\tau_1]\rightarrow \R$, covolution allows us to find a sequence  of functions $\{x^i_m\}_m\subset C^{\infty}([0,\tau_1])$ such that: (a)  $|y^i_\gamma-x^i_m|<1/m$, (b)~for some $L>0$,  all  $x^i_m$ are
$L$-Lipschitz  and (c) the sequence of derivatives $\{(x^i_m)'\}_m$ converge pointwise to $(y^i_\gamma)'$ a.e.  in $[0,\tau_1]$ (for the last property, use for example
\cite[Th. 4.1, p. 146]{EG}).
Taking into account that   $x_m^i(0)=0$  can be assumed additionally (otherwise, recall (a) and add a suitable constant to $x_m^i$),
 a sequence of curves $\{x_m\}_m$  satisfying the properties in \eqref{e4bis1} is obtained (to ensure $s\circ x_m (=x_m^{n-1})\geq 0$, choose a positive mollifier $\rho_\epsilon$). What is more, the conditions (a) and (b) imply that all the  sequences $\{x^i_m\}_m$ and $\{(x^i_m)'\}_m$ are bounded by constants, while  (c) implies the pointwise convergence of the corresponding integrands in \eqref{e4bis11} to the limit (expressed in terms of $y'_\gamma$).
So, applying the theorem of dominated convergence,  the sequence $\{l_m(\tau_1):=l_{\epsilon_m}(\tau_1)\}_m$  converges to:
 $$
\int_{0}^{\tau_1} \sqrt{|y_\gamma'|_0^2(\tau) + \frac{\nu}{2}} \;  d\tau <  \; \tau_1,
 $$
(recall $|y_\gamma'|_0^2<1-\nu$ a.e., see below \eqref{e_cgamma}),
and \eqref{e4bis11} holds, as required.

 Finally, to go from the previous local result to the  global one, a standard procedure is followed.  Namely,
given $\gamma$, for each $\tau\in I:=[\tau_0,\tau_1]$ the local case ensures the existence of some
$\delta_\tau>0$ such that, if $\tau'\in (\tau-\delta_\tau,\tau) \cap I$ (resp.  $\tau'\in (\tau ,\tau+\delta_\tau ) \cap I$ then there exists a  a  piecewise smooth future-directed timelike curve starting at $\gamma(\tau')$ and ending at $\gamma(\tau)$ (resp. starting at $\gamma(\tau)$ and ending at $\gamma(\tau')$). Thus, the result follows by taking a Lebesgue number $\delta$ for the obtained open covering of $I$, a partition of $I$ with diameter smaller than $\delta$, and, then, joining each pair of consecutive points for the partition by means of such a smooth  timelike curve.
\end{proof}

\begin{rem}\label{r_futuresmooth}
{\em \bcambios{Notice that the conditions imposed in Prop \ref{p_futuresmooth} on $\gamma$ would be a natural definition of being a {\em (future-directed) $H^1$-timelike curve} defined on a compact interval (for non-compact intervals, the definition would be that the restriction of $\gamma$ to any compact interval were $H^1$-timelike in the previous sense). Indeed, clearly, {\em a piecewise smooth curve is timelike in the usual sense if and only if it is $H^1$-timelike as above}.}

One could consider  a more general definition of $H^1$-timelike curve, by relaxing the inequality $g(\gamma',\gamma')<-\nu$ a.e. into $g(\gamma',\gamma')<0$ a.e. Such a definition would have the drawback that, even if $\gamma$ is smooth,
it might be $H^1$-timelike in this generalized sense but not in the usual one. Indeed, this would happen
when $\gamma'$ is timelike everywhere but in some points, where it is either lightlike or zero; for example, when $\gamma'$ is lightlike either at a finite number of points, or at a sequence of points and its limit.  Even though the proof of Prop.~\ref{p_futuresmooth} might be extended to this general case,
we do not feel that this generalized notion of $H^1$-timelike (or any possible alternative, even if it yielded a different sense of the chronological relation) would be useful. Roughly, the reason is that all the $H^1$-causal curves (in particular those with $g(\gamma',\gamma')<0$ a.e., or which were
$H^1$-timelike in any reasonable sense) can be approached by sequences of
$H^1$-timelike curves $\{\gamma_k\}$ with
$g(\gamma'_k,\gamma'_k)<-1/k$ a.e. (namely, locally, if $\gamma=(\tau,y_\gamma(\tau))$ put $\gamma_k=(\tau,\sqrt{1-1/k} \, \cdot y_\gamma(\tau))$), eventually maintaining one fixed endpoint $p\in \M$. For these curves, Prop. \ref{p_futuresmooth} applies; thus, in particular,  the corresponding $H^1$-causal future (or past) of $p$ lie in the closure of its chronological future (or past) computed by means of piecewise smooth timelike curves.
}\end{rem}

Following the discussion of item \ref{i2}, let $J^\pm_{ps}(p)$
be  the causal future/past of $p\in\M$ when computed with piecewise smooth causal curves, and $J^\pm_{H^1}(p)$ when computed with $H^1$-causal ones.

\begin{prop}\label{l}
Let  $(\overline{M},g)$ be a spacetime-with-timelike-boundary. If it satisfies the definition of global hyperbolicity computing the causal futures and pasts by using piecewise smooth causal curves, then $J^{\pm}_{ps}(p)=J^{\pm}_{H^{1}}(p)$ for all $p\in \overline{M}$. So, in this case,
$(\overline{M},g)$ also satisfies global hyperbolicity by using $H^1$-causal curves.
\end{prop}

\begin{proof}
\bcambios{Let us reason for the future. As $
J^{+}_{ps}(p) \subset J^{+}_{H^{1}}(p)$ is trivial, it is sufficient to check:}
\begin{equation}\label{eqq}
J^{+}_{H^{1}}(p) \subset cl(I^{+}(p))\subset  cl(J^{+}_{ps}(p))=J^{+}_{ps}(p) \qquad\forall p \in \overline{M}.
\end{equation}
The second inclusion is trivial (recall that $I^+(p)$ is computed with piecewise smooth curves, as discussed above) and the equality means that, in the piecewise smooth case, the causal simplicity of the spacetime follows from its global hyperbolicity (which is proven  easily as in the case without boundary or the $H^1$ one, see  Lemma  \ref{ppp} below).

 For the  first inclusion in \eqref{eqq}, recall that its local version (namely, $p$ admits an open  neighborhood $V\subset\overline{M}$, such that
$J^{+}_{H^{1}}(q,V) \subset cl(I^{+}(q,V)$ for all $q \in V$) follows from an  observation
in Remark \ref{r_futuresmooth}: whenever $\gamma=(\tau,y_\gamma(\tau))$ is $H^1$-causal, then $\gamma_k=(\tau,\sqrt{1-1/k} \, \cdot y_\gamma(\tau))$ approaches $\gamma$, it  satisfies $g(\gamma_k',\gamma_k')\leq -1/k$ a.e. and, then, Prop.~\ref{p_futuresmooth} applies.
 To go from this local version  to the  global one \eqref{eqq}, the following slight modification of the procedure in the proof of Prop. \ref{p_futuresmooth} is carried out.  Take  any $q \in J_{H^{1}}^{+}(p)$ and let $\gamma: I\rightarrow \M$ be some future-directed $H^{1}$ causal curve joining $p$ with $q$.
From the local result, for any $\gamma(r)$, $r\in I$, there exists some open neighbourhood $V_{\gamma(r)}$ such that $J^{+}(q',V_{\gamma(r)})_{H^{1}} \subset cl(I^{+}(q',V_{\gamma(r)}))$ for all $q'\in V_{\gamma(r)}$.
So, take a Lebesgue number $\delta>0$ for the open covering
$\{\gamma^{-1}(V_\gamma(r)), r\in I\}$ of $I=[a,b]$, and choose a
partition
 $\{r_{0}=a<r_{1},\ldots,<r_{l-1}<r_l=b\}$ with diameter smaller than $\delta$. The case $l=1$ is trivial, and assume by induction that \eqref{eqq} holds  for $l-1$. Let $r$ be so that
$\gamma([r_{l-1},r_{l}]) \subset V_{\gamma(r)}$ and, thus, $\gamma(b) \in cl(I^{+}(\gamma(r_{l-1}),V_{\gamma(r)}))$.
So, there is a sequence $\{q_m\}\rightarrow q$ such that $\gamma(r_{l-1})\in I^{-}(q_m,V_{\gamma(r)})$ for all $m$.
Therefore, for each $q_m$ all the points in some neighborhood $U_m\ni \gamma(r_{l-1})$ lie in $I^-(q_m,V_{\gamma(r)})$. By the hypothesis of induction, some  $u_m\in U_m$  belongs to $I^+(p)$ and, so, $p\ll u_m \ll q_m$ for all $m$, as required.
\end{proof}



\begin{conv}\label{convention} {\em \bcambios{In what follows,  all the causal curves will be $H^1$-causal, all the  futures and pasts are computed  with $H^1$-curves and all the corresponding causal definitions are carried out accordingly. Consistently,  from now on,  $J^\pm(p)$ means $J^\pm_{H^1}(p)$. For timelike curves, classical piecewise smooth regularity will be used.}

\bcambios{As a summary, the consistency of this convention comes from the following facts. About $H^1$-causal curves: (a) they are  equivalent to classical continuous causal curves in manifolds without boundary, (b) they are intrinsic (extensions $\widetilde{M}$ of $\M$ are not required),  (c) they are preserved by limit curves in the same way as in the case without boundary, and (d) they lead to a more general notion of globally hyperbolic \bcambios{spacetime-with-boundary}, where our results apply.
About piecewise smooth timelike curves: (a)~they can connect the same endpoints as  $H^1$-timelike ones, and (b)
 $H^1$-causal curves  can be approched by piecewise smooth timelike  ones, even fixing one of the endpoints.}
} \end{conv}


\section{ Topological splitting  and Geroch's equivalence  }\label{s3}

\subsection{Higher steps of the causal ladder.}\label{s3.2}

Geroch's proof of the topological splitting of any globally hyperbolic spacetime (without boundary) $(M,g)$ is based on the existence of a time function constructed by computing  certain volumes 
 using an appropriate measure $m$. The conditions to be satisfied by $m$ are very mild; indeed, they are satisfied by the measure associated to any semi-Riemannian metric $g^*$ such that the total volume of the manifold is finite (so, one can choose $g^*$ conformal to the original Lorentzian metric $g$). However, following Dieckmann   (\cite{Dieckmann}; see also \cite[section 3.7]{MSCH}) the abstract  properties required for a measure on $\M$ will be recalled first.

Given the \ncambios{spacetime-with-boundary} $(\overline{M},g)$, $\overline{M}=M \cup \partial M$, consider the $\sigma$-algebra $\mathfrak{A}(\tau_{\overline{M}}
\cup \overline{Z})$ \cambios{generated by} the topology $
\tau_{\overline{M}}$  of $
\overline{M}$ in addition to the set
$\overline{Z}$
containing the \cambios{zero-measure} sets of $\overline{M}$.
Since $M$ is an
open subset of $\overline{M}$, the $\sigma$-algebra $\mathfrak{A}(\tau_{M} \cup Z)$ of $M$ 
coincides
with the induced $\sigma$-algebra of $\mathfrak{A}(\tau_{\overline{M}} \cup \overline{Z})$ over $M$, that is, with the set
$\{E \cap M \mid E \in \mathfrak{A}(\tau_{\overline{M}} \cup \overline{Z})\}$.
In a natural way, the {\em measures} on the previous $\sigma$-algebras will be called just measures on $\M$ or $M$, consistently.

It is straightforward to check that if $m$ is a measure on $M$ then it induces naturally  a measure $\overline{m}$ on $\M$ just imposing $\overline{m}(\partial M)=0$, that is,
\[
\overline{m}(A):=m(A \cap M)\quad\hbox{for any $A \in \mathfrak{A}(\tau_{\overline{M}} \cup \overline{Z})$.}
\]

\begin{defi}\label{d_admissible}
A measure  $\overline{m}$ on $\M$ is {\em admissible} when it satisfies:
\begin{enumerate}
 \item $\overline{m}(\overline{M})<\infty$;
 \item $\overline{m}(U)>0$ for any open subset $U \subset \overline{M}$;
 \item $\overline{m}(\dot I^{\pm}(p))=0$ ($\dot I^{\pm}(p)$ denotes the topological boundary of $I^{\pm}(p)$ in $\M$ ),  $\forall p \in \overline{M}$;
 \item \ncambios{$\overline{m}$ is {\em inner regular}, i.e.} for any open subset $U \subset \overline{M}$ there exists a sequence $\{K_{n}\}_{n}\subset \overline{M}$ of compact subsets such that $K_{n} \subset K_{n+1}$, $K_{n} \subset U$ for all $n$ and
 $\overline{m}(U)=\lim_{n} \overline{m}(K_{n})$.
\end{enumerate}
\end{defi}

\begin{prop}\label{p_admiss}
If $m$ is an admissible measure on $M$ then the induced measure $\overline{m}$ on $\M$ is also admissible.
\end{prop}

\begin{proof}
 Properties 1, 2, 4 in Def. \ref{d_admissible} are straightforward. To prove 3, it is enough to check, say, $m(\dot I^{-}(p) \cap M)=0$. This holds because it is an achronal edgeless subset of $M$, since $\dot I^{-}(p)$ can be written locally as a graph by using Cor. \ref{cor}. \ncambios{(Notice also
 the inclusions $I^+(p)\subset cl(I^+(p)) \subset cl(J^+(p))=cl(I^+(p))$, so the measure of $I^+(p)$ coincides with the measure of $cl(J^+(p))$.)}
\end{proof}
\begin{rem}
{\em
An alternative way to define an admissible measure $\overline{m}$ on $\M$ is to consider any admissible measure on an extension $\tilde{M}$ of $(\M, g)$ (as in  Prop. \ref{pp}) and taking the restriction to $\M$.}
\end{rem}
In what follows, an admissible measure $\overline{m}$ is fixed on $\M$.

\begin{defi}
The function $t^{-}(p):=\overline{m}(I^{-}(p))$ (resp. $t^{+}(p):=-\overline{m}(I^{+}(p))$) is  the {\em past} (resp. {\em future}) {\em volume function} associated to $\overline{m}$.
\end{defi}

Trivially, the volume functions are non-decreasing on any future-directed causal,
but they are constant on any closed causal curve. The next two propositions hold as in the case without boundary (we refer to \cite[Sect. 3.2.1]{LuisTesis} for detailed proofs).

\begin{prop}\label{p_dist} If $(\overline{M},g)$ is past (resp. future) distinguishing, the volume function $t^-$ (resp. $t^+$) is (strictly) increasing over any future-directed causal curve $\gamma$.
\end{prop}



\begin{prop}\label{p_creer} For any spacetime with timelike  boundary $(\M,g)$ the following properties are equivalent:

\ben
\item The set valued function $I^+$ (resp. $I^-$) is continuous on $\M$.

\item The volume function $t^+$ (resp. $t^-$) is continuous on $\M$.

\item $(\M, g)$ is {\em past} (resp. {\em future}) {\em reflecting}, that is, for all $p, q\in \M$:

$q\in cl(I^+(p))
\Rightarrow p\in cl(I^-(q)) \qquad\hbox{(resp.
$p\in cl(I^-(q))
\Rightarrow q\in cl(I^+(p))$.)}$
\een
So, a distinguishing spacetime-with-timelike-boundary is causally continuous if and only if some/any of the previous equivalent properties (for the future and the past) hold, or equivalently, if and only if both volume functions $t^+, t^-$  are time functions.
\end{prop}

The following result will allow to complete the implications of the ladder,
taking into account the optimized definitions of global hyperbolicity and causal simplicity used here (consistent with \cite{BS07}).

\begin{lemma}\label{ppp}
Let $(\overline{M},g)$ be a spacetime-with-timelike-boundary.

(a) If it is causally simple then it is causally continuous.

(b) If it is globally hyperbolic then
it is causally simple. 
\end{lemma}

\noindent {\bf Proof:} (a) We have to prove that is both distinguishing and (future and past) reflecting, \ncambios{recall Prop. \ref{p_creer}}.
For the former, following \cite{BS07}, if $p\neq q$ but, say,  $I^+(p)=I^+(q)$  then choose any sequence $\{q_n\}\rightarrow q$ with $q\ll q_n$ and, thus, $p\ll q_n$. Then  $q\in cl(I^+(p))=cl(J^+(p))=J^+(p)$ (the first equality by Prop. \ref{p_opentrans} and the second by hypothesis). Analogously,  $p\in J^+(q)$ and there is a closed causal curve with endpoints at $p$ crossing $q$.
For the latter property, causal simplicity implies
$J^{\pm}(p)=cl(I^{\pm}(p))$ and, thus, the reflectivity becomes equivalent
to the trivially true property $q\in J^+(p) \Leftrightarrow  p\in J^-(q)$.

(b) Following \cite[Props. 3.16, 3.17]{didier}), let us check that, say, $J^+(p)$ is closed. Let $r= \lim_m r_m$ with $r_m\in J^+(p)$. Since $\pM$ is timelike,
there exists some $q\in I^+(r)$ and
$r_m\in I^-(q)$ for large $m$. Then, $r\in cl(J^+(p)\cap J^-(q))= J^+(p)\cap J^-(q)$ (the latter by global hyperbolicity), and thus, $r\in J^+(p)$. \cvd

\begin{thm}\label{t_ladder} Let $(\M,g)$ be a spacetime-with-timelike-boundary.
\ben
\item If $(\M,g)$ is causally continuous then it is stably causal. Moreover, $(M,g|_M)$ is causally continuous and $(\pM,g|_{\pM})$ is stably causal.
\item If $(\M,g)$ is causally simple then it is causally continuous.
\item If $(\M,g)$ is globally hyperbolic then it is causally simple. Moreover, $(\pM,g|_{\pM})$ is globally hyperbolic  too.
\een

\end{thm}

{\em Proof.} 1. The first assertion follows because any of the volume functions $t^+,t^-$ provides the required time function (recall Prop. \ref{p_creer}). Moreover, $\partial M$ is also stably causal because, trivially,  the restrictions of these functions to $\pM$ are also  time functions.
The causal continuity of the interior  $M$ is again a consequence of Prop. \ref{p_creer} taking into account that $t^\pm$ are both continuous on all $\M$ and  their restrictions on $M$  agree with the volume functions for the measure $m$ on $M$. Indeed, as $I^\pm(p,M)=
I^\pm(p)\cap M$ (Prop.~\ref{p_opentrans}), $m(I^\pm(p,M))=\overline{m}(I^\pm(p))$ for all $p\in M$, and the result follows.

2. This is just Lemma \ref{ppp} (a).

3. The first assertion is just Lemma \ref{ppp} (b). The last assertion follows from \cite[Prop. 3.15]{didier}, we include the proof for completeness.
As $\pM$ is strongly causal, from previous items 1 and 2 it is enough to check that  $cl(J^+(p,\pM)\cap J^-(q,\pM))$ is compact (see  \cite[Lemma 4.29]{Beem}).  Now, any sequence in this subset admits a subsequence converging to some $r \in J^+(p) \cap J^-(q)$ and, as $\pM$ is closed in $\M$, $r \in cl(J^+(p,\pM)\cap J^-(q,\pM))$.

 \begin{rem}\label{r_counterexamples}  {\em Thm. \ref{t_ladder} and Prop. \ref{p_ord_lowerlevels} allow us to reobtain   the strict ordering of all the steps in the classical causal ladder in the  case with  boundary.
The following examples show, in particular, that the inherited properties for $M$ and $\pM$ are optimal.

(1) Start with the closed half space of Lorentz-Minkowski 3-space $\{(t,x,y)\in
\LL^3: y\geq 0\}$ and remove the line $L=\{(0,x,0) \mid x \geq 0\}$. Since the interior $M$ is causally continuous and the continuous extension of the $M$-volume functions to the whole $\M$ agree with the volume functions on $\M$, the spacetime $\M$ is causally continuous (recall Prop. \ref{p_creer}), but, clearly, $\pM$ is not (see Fig. \ref{fig1}).

	\begin{figure}[hbt!]
	\centering
	\ifpdf
	\setlength{\unitlength}{1bp}%
	\begin{picture}(240,150)(0,0)
	\put(0,0){\includegraphics[scale=.25]{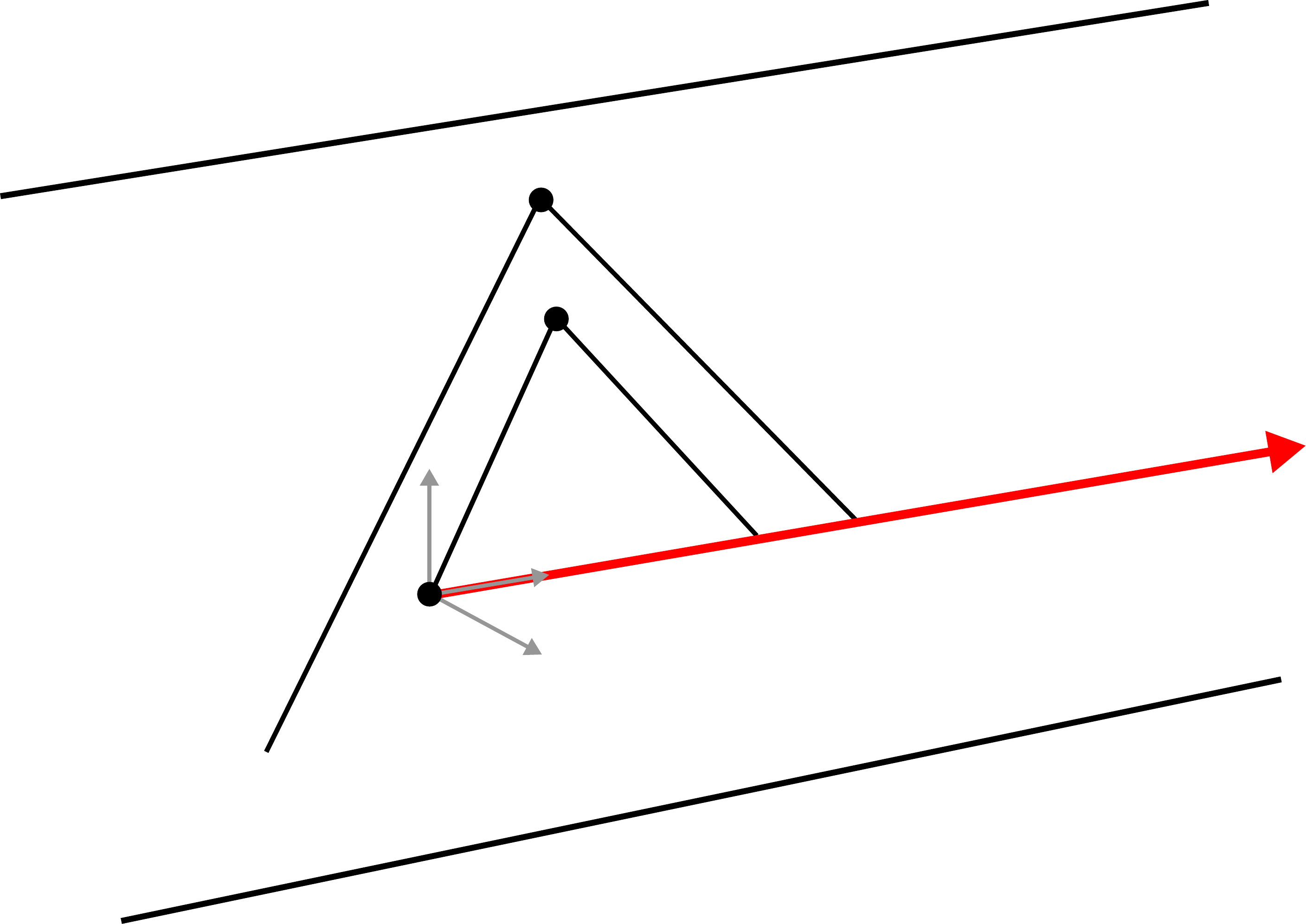}}
	\put(150.00,48.00){\fontsize{14.23}{17.07}\selectfont $L$}
	\put(08.00,98.90){\fontsize{14.23}{17.07}\selectfont $\partial M$}
	\put(83,100.00){\fontsize{09.50}{0.0}\selectfont $p$}
	\put(82.30,118.00){\fontsize{09.50}{0.0}\selectfont $p'$}
	\put(220.30,25.00){\fontsize{14.23}{17.07}\selectfont $M$}
	\put(82,35.00){\fontsize{09.00}{0.0}\selectfont $y$}
	\put(84,55.900){\fontsize{09.00}{0.0}\selectfont $x$}
	\put(68.5,67.50){\fontsize{09.00}{0.0}\selectfont $t$}
	\put(64,42){\fontsize{09.00}{0.0}\selectfont $0$}

	\end{picture}%
		\else
	\setlength{\unitlength}{1bp}%
	\begin{picture}(240,150)(0,0)
	\put(0,0){\includegraphics[scale=.25]{NonCausallyContinuousB}}
	\put(150.00,48.00){\fontsize{14.23}{17.07}\selectfont $L$}
	\put(08.00,98.90){\fontsize{14.23}{17.07}\selectfont $\partial M$}
	\put(83,100.00){\fontsize{09.50}{0.0}\selectfont $p$}
	\put(82.30,118.00){\fontsize{09.50}{0.0}\selectfont $p'$}
	\put(220.30,25.00){\fontsize{14.23}{17.07}\selectfont $M$}
	\put(82,35.00){\fontsize{09.00}{0.0}\selectfont $y$}
	\put(84,55.900){\fontsize{09.00}{0.0}\selectfont $x$}
	\put(68.5,67.50){\fontsize{09.00}{0.0}\selectfont $t$}
	\end{picture}%
	\fi
	\caption{\label{fig1} \bcambios{The spacetime-with-timelike-boundary $\M=\{(t,x,y) \in \LL^3 \mid y \geq 0\} \setminus L$, where $L=\{(0,x,0) \mid x \geq 0\}$, with the Lorentzian metric $g=-dt^2+dx^2+dy^2$ and timelike boundary $\partial M = \{(t,x,y) \in \LL^3  \mid y=0\}$ is depicted. If we take a point $p' \in \pM$ slightly to the future of $p \in \partial M$, then, the volume function restricted to the boundary is not continuous at $p$. So, $(\pM,g\mid_\pM)$ is not causally continuous.
	}}
	
\end{figure}

(2) Consider now $\M={\mathbb L}^3\setminus C$, where
$C$ is the timelike cylinder $\R\times D$, being $D$ the disk $\{x^2+y^2<1\}\subset \R^2$. Clearly, $J^+(p,M)$ is not closed whenever there exists a future-directed lightlike half-line $l$ starting at $p\in M$ and tangent to $C$ at some point $\hat q\in C$; indeed, the points in $l$ beyond $\hat q$ will lie in  $cl(J^+(p,M))\setminus J^+(p,M)$. That is, in general $M$ is not  causally simple, even if $\M$ is globally hyperbolic with timelike boundary.

\begin{figure}[hbt!]
\centering
\ifpdf
\setlength{\unitlength}{1bp}%
\begin{picture}(160,140)(0,0)
\put(0,0){\includegraphics[scale=.25]{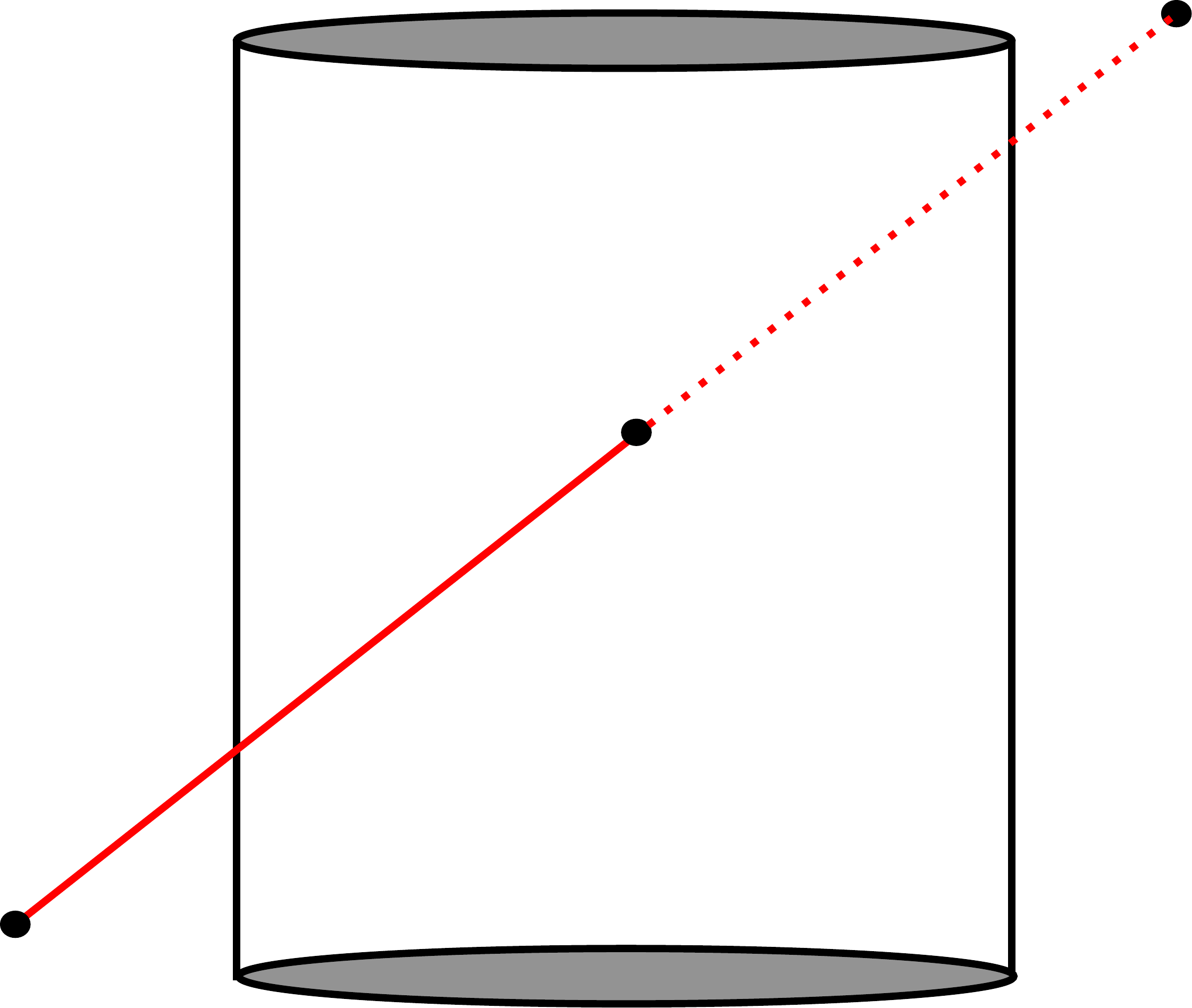}}
\put(-10.00,108.90){\fontsize{14.23}{17.07}\selectfont $M$}
\put(40.00,108.90){\fontsize{14.23}{17.07}\selectfont $\pM$}
\put(0,0){\fontsize{11.50}{0.0}\selectfont $p$}
\put(88,70.00){\fontsize{11.50}{0.0}\selectfont $\hat q$}
\put(158.90,122.90){\fontsize{14.23}{17.07}\selectfont $r$}
\put(45.00,33.90){\fontsize{15.50}{0.0}\selectfont $l$}
\end{picture}%
\else
\setlength{\unitlength}{1bp}%
\begin{picture}(160,140)(0,0)
\put(0,0){\includegraphics[scale=.25]{Causallysimpleint}}
\put(-10.00,108.90){\fontsize{14.23}{17.07}\selectfont $M$}
\put(40.00,108.90){\fontsize{14.23}{17.07}\selectfont $\pM$}
\put(0,0){\fontsize{11.50}{0.0}\selectfont $p$}
\put(88,70.00){\fontsize{11.50}{0.0}\selectfont $\hat q$}
\put(158.90,122.90){\fontsize{14.23}{17.07}\selectfont $r$}
\put(45.00,33.90){\fontsize{15.50}{0.0}\selectfont $l$}
\end{picture}%
\fi
\caption{\label{cylinder02}\bcambios{Observe that the lightlike line $l$ touches the timelike boundary $\pM$ at a unique point $\hat{q}$. So, if we consider the spacetime (without boundary) $(M,g)$ then the points $p$ and $r$ are not causally related in $M$. However, we have that $r \in cl(J^{+}(p,M)) \setminus J^{+}(p,M)$, so, $(M,g)$ is not causally simple.}}
\end{figure}

(3) Finally, consider the closure of the previous cylinder, $\overline{C}= \R\times \overline{D}\subset \LL^3$, take the arc $A=\{(0,\cos \theta,\sin \theta): 0\leq \theta \leq \pi/4\}$ and consider the spacetime $\M = \overline{C}\setminus A$ (see Fig. \ref{cylinder01}). Clearly, $\partial M$ is not causally continuous. However, $\M$ (and also $M$) is causally simple. Indeed, for each $p\in M$, it is obvious not only that $J^\pm(p,M)$ is closed but also that so is $J^{\pm}(p)$. This happens even if $p\in \pM$ because any $q (\neq p)$ in the boundary of $J^{\pm}(p)$ can be joined with $p$ by means of a lightlike segment included in $M$ up to the endpoints. \ncambios{(From a more general viewpoint, this property happens because $\pM$ is strongly light-convex; see \cite[Sect. 3.2 and Thm. 3.5]{CGS} for further background and results.)}

\begin{figure}[hbt!]
\centering
\ifpdf
\setlength{\unitlength}{1bp}%
\begin{picture}(95,140)(0,0)
\put(0,0){\includegraphics[scale=.25]{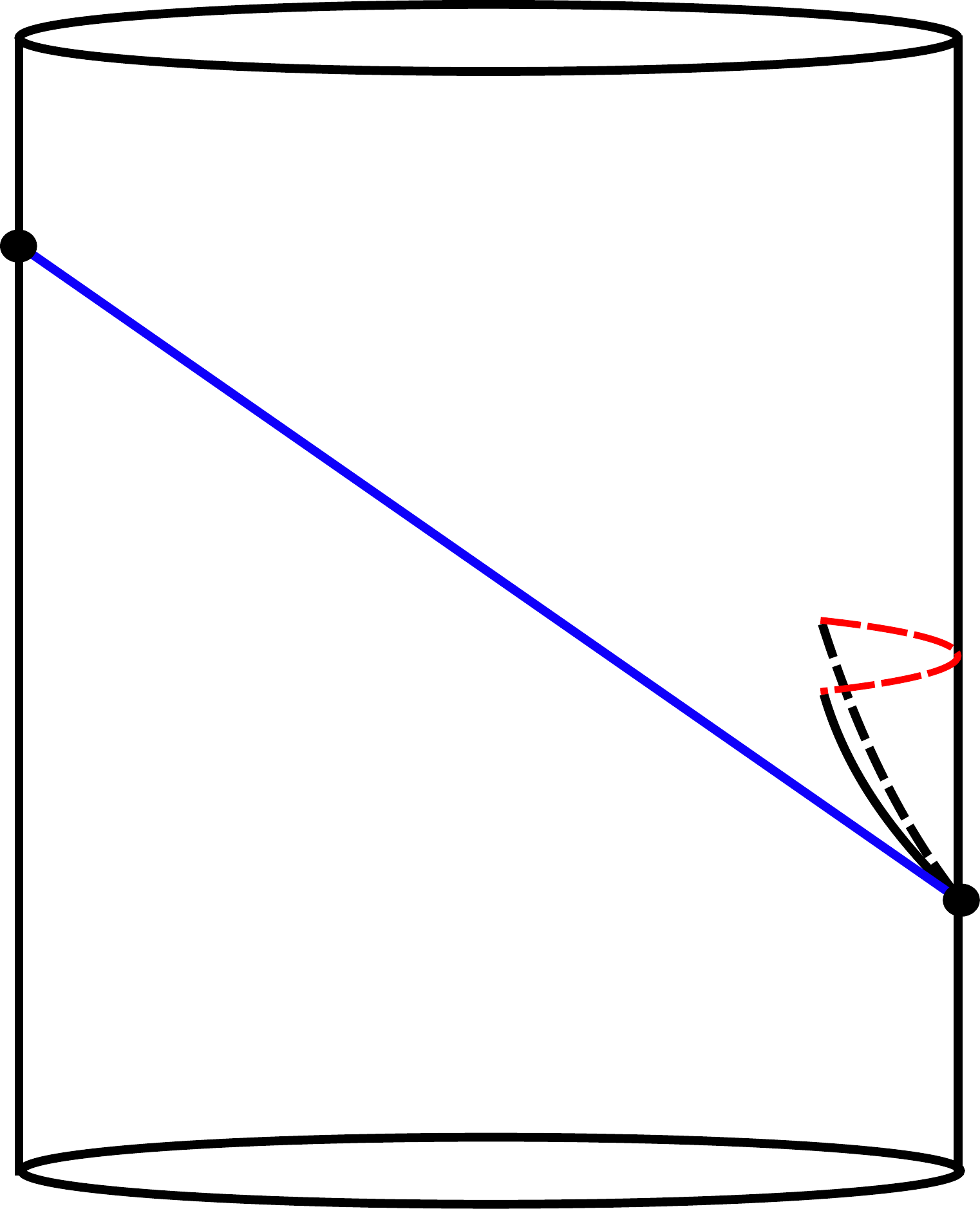}}
\put(-29.00,118.90){\fontsize{14.23}{17.07}\selectfont $\partial M$}
\put(20.00,108.90){\fontsize{14.23}{17.07}\selectfont $M$}
\put(110,30.00){\fontsize{11.50}{0.0}\selectfont $p$}
\put(95,70.00){\fontsize{11.50}{0.0}\selectfont $A$}
\put(-8.00,105.90){\fontsize{11.50}{0.0}\selectfont $q$}
\end{picture}%
	\else
\setlength{\unitlength}{1bp}%
\begin{picture}(95,140)(0,0)
\put(0,0){\includegraphics[scale=.25]{Cylinder01v02}}
\put(-29.00,118.90){\fontsize{14.23}{17.07}\selectfont $\partial M$}
\put(20.00,108.90){\fontsize{14.23}{17.07}\selectfont $M$}
\put(110,30.00){\fontsize{11.50}{0.0}\selectfont $p$}
\put(95,70.00){\fontsize{11.50}{0.0}\selectfont $A$}
\put(-8.00,105.90){\fontsize{11.50}{0.0}\selectfont $q$}
\end{picture}%
\fi
\caption{\label{cylinder01}\bcambios{The timelike boundary $\pM$ is not causally continuous. In fact,  the future volume function of $\pM$ \gcambios{ will not be continuous at $p$ (check this by taking $p' \in \pM$ slightly to the past of $p \in \pM$); so, neither is $\pM$   causally simple. } However, the spacetime-with-timelike-boundary $\M$ is causally simple since any point $q \in \partial J^{+}(p)$ can be connected by a lightlike \gcambios{segment which is included in $M$ except, at most, its endpoints.}
}}
\end{figure}

}
\end{rem}

\subsection{Cauchy hypersurfaces and extended Geroch's  proof.}
\label{s3.3}

%


Formally, our notion of Cauchy hypersurface  for a spacetime-with-timelike-boundary is equal to the minimal one developed in  \cite{O} for  the case without boundary.
\begin{defi}\label{d_CauchyHyp} Let $(\overline{M},g)$ be a spacetime-with-timelike-boundary.
An achronal set $\bar{\Sigma}\subset \overline{M}$ is a {\em Cauchy hypersurface} if it is intersected exactly once by every inextensible timelike  curve.
\end{defi}
Next, let us
 check  some properties of  Cauchy hypersurfaces (in particular, that they are truly    hypersurfaces) by adapting the approach in \cite[pp. 413-415]{O}.

\begin{defi}\label{d_achronal} Let $(\M,g)$ be a spacetime-with-timelike-boundary of dimension $n$.

(1) A subset $\bar S\subset \M$ is a {\em (embedded) topological hypersurface
transverse to $\pM$} if for any $p\in \bar S$ there exists an open neighborhood $V\subset \M$ of $p$ and a homeomorphism
$\phi: V \rightarrow  (-\epsilon,\epsilon)\times N$, where  $N$ is an open subset of $\R^{n-2}\times [0,\epsilon)$, 
such that $\phi(V\cap \bar S)=\{0\}\times N$. In this case, $\bar S$ is {\em locally Lipschitz} if the previous homeomorphism can be chosen Lipschitz when written in smooth coordinates.

(2) The {\em edge} of an achronal set $A \subset \overline{M}$ is the set of all the points $p \in cl(A)$
 such that, for every open neighborhood $U\subset \overline{M}$ of $p$, there exists
a timelike curve contained in $U$ from $I^{-}(p,U)$ to $I^{+}(p,U)$ that does not intersect $A$.
\end{defi}

\begin{prop}\label{pf}
Let $(\overline{M},g)$
be a spacetime-with-timelike-boundary. An achronal set $A$ with $A\cap {\rm edge}(A)=\emptyset$ is a locally Lipschitz  hypersurface transverse to $\pM$.
\end{prop}

\begin{proof} The  condition on ${\rm edge}(A)$ implies that the achronal set $A\cap M$ contains no edge  points, thus,  from \cite[Prop. 14.25]{O})  $A\cap M$ is a topological hypersurface without boundary in $(M,g)$.
So, let $\hat{p} \in A \cap \partial M\neq\emptyset$ and let us  construct a  Lipschitz  topological chart $(V,\phi)$ centered at  $\hat{p}\in V \subset \overline{M}$ as in Defn. \ref{d_achronal} (1). Since $A\cap {\rm edge}(A)=\emptyset$, there exists some open neighbourhood $U$ of $\hat{p}$ such that any timelike curve contained in $U$ going from $I^{-}(\hat{p},U)$ to $I^{+}(\hat{p},U)$ intersects $A$. Without loss of generality, we can assume   that $(U,\psi=(x^{0},\ldots,x^{n-1}))$ is a Gaussian chart centered at $\hat{p}$ such that $\psi(U)$ is a cube.
Consider the subset $(-\epsilon,\epsilon) \times N_0 \subset \psi(U)
$, where $N_0=\{ y\in   \R_+^{n-1}: (0,y)\in \psi(U)\}$.
 The required  neighbourhood $V\ni\hat p$ will be  $V=\psi^{-1}((-\epsilon, \epsilon) \times N)$,  where $N\subset N_0$ is any (half) open ball centered at $0$ for the natural Euclidean metric $|\cdot |$ in the coordinates $\hat x:=(x^1, \dots, x^{n-1})$,
with small radius  so that $N$ is relatively compact in $N_0$ and the slices $x^0=\pm  \epsilon/2$ satisfy:
\[
\{q\in \bv V \ev : x^0(q)=-\epsilon/2\} \subset I^-(\hat{p},U)\quad\hbox{and}\quad \{q\in \bv V \ev : x^0(q)=\epsilon/2\} \subset I^+(\hat{p},U)
\]
(this can be achieved trivially because $\partial_0$ is timelike).
Now,  the integral curve of $\partial_0$ starting at any $(0,y)\in \{0\}\times N$ must intersect both $I^-(\hat{p},U)$ and $I^+(\hat{p},U)$.
Thus, it must intersect $A$ (recall $U\supset V$ and $\hat p\not\in {\rm edge} (A)$) in a  point $y_q$, which is unique by the achronality of $A$. So,  $A \cap V$ can be regarded as the graph of the function $h:N \rightarrow (-\epsilon, \epsilon)$,  $h(q):=x^0(y_{q})$. To show that $h$ is  Lipschitzian will be enough  because, in this case, the desired  chart $\phi$ on $V$ is just:

$$ \phi(\psi^{-1}(x^0,\ldots,x^{n-1}))=  (x^0-h(x^{1},\ldots,x^{n-1}), x^{2},\ldots,x^{n-1})$$
The
  Lipschitz condition will be  checked with respect to 
the  distance
$|\cdot |$ induced in $V$ by $\psi$ from  the Euclidean one.
 Recall first that, as $cl(V)$ is compact and $\partial_0$ is timelike, there exists some small $c>0$ such that the flat Lorentzian metric $g_c=-c^2 (dx^0)^2 + \sum_{i=1}^{n-1}(dx^i)^2$ satisfies $g_c <g$. The Lipschitz condition $|h(x)-h(y)| <  c^{-1}| x-y|$ must hold for all  $x,  y \in N$ because, otherwise, the points $(h(x), x), (h(y),y)$ would be (future or past) causally related for $g_c$ and, thus, for $g$ (in contradiction with the achronality of $A$).
\end{proof}
Other properties of achronal sets (including a converse to Prop. \ref{pf}) can be found in \cite[Sect. 3.2.3]{LuisTesis}. Next, we focus on consequences for Cauchy hypersurfaces.

\begin{cor}\label{co} 
(1) An achronal set $A$ with ${\rm edge}(A)=\emptyset$ is a closed (as a subset of $\M$)   locally Lipschitz hypersurface  transverse to $\partial M$.

(2) Let $F \neq \emptyset, \M$ be a future set (i.e $I^+(F)=F$). Then,  its topological boundary $\dot F$  is an achronal closed locally Lipschitz  hypersurface transverse to $\partial M$.

 (3) Any Cauchy hypersurface $\bar{\Sigma}$ of $(\overline{M},g)$ is an achronal closed locally Lipschitz hypersurface with boundary transverse to $\pM$.

\end{cor}
\begin{proof}
 (1) 
Prop. \ref{pf} implies that $A$ is a locally Lipschitz hypersurface transverse to $\partial M$. Closedness follows directly from the general inclusion
$cl(A) \setminus A \subset {\rm edge}(A)$,
which happens because, if $q \in cl(A) \setminus A$, no timelike curve through $q$ can intersect $A$ ($A$ achronal implies  $cl(A)$ achronal, as the chronological relation is also open in the case with boundary; recall Prop. \ref{p_opentrans} (a))  and, so, $q \in {\rm edge}(A)$.



(2) Taking into account elementary properties of transitivity (Prop. \ref{p_opentrans}),
$$I^{+}(\dot F) \subset \hbox{interior}(F),\quad I^{-}(\dot F) \subset \overline{M} \setminus cl(F), \quad\hbox{in particular,}\quad I^{+}(\dot F) \cap I^{-}(\dot F)=\emptyset .
$$ 
From this  last equality  
$\dot F$ is achronal. From the part (1), to prove  ${\rm edge}(\dot F)=\emptyset$ suffices and, since $\dot F$ is closed, ${\rm edge}(\dot F)
\subset 
\dot F$, that is,  to prove ${\rm edge}(\dot F)\cap \dot F=\emptyset$  suffices too. Assuming   $p\in {\rm edge}(\dot F)\cap \dot F$, there exists a timelike curve starting at $I^-(p)$ (thus, in $\overline{M} \setminus cl(F)$) and ending at  $I^+(p)$ (in ${\rm int}(F)$)  without crossing $\dot F$, a contradiction.

(3) Clearly, $\M$ is the disjoint union
$\overline{M}=I^{+}(\bar{\Sigma}) \cup \bar{\Sigma} \cup I^{-}(\bar{\Sigma})$.
So, $\bar{\Sigma}$ is the boundary of the future set $F=I^{+}(\bar{\Sigma})$ and the part (2) applies.
\end{proof}
Next, let us re-take the existence of a Cauchy time function.
Let
\begin{equation}\label{geroch}
t: \overline{M} \rightarrow \mathbb{R},\quad t(p):=\ln\left(-\frac{t^{-}(p)}{t^{+}(p)}\right)
\end{equation}
be a {\em Geroch function}.
From Thm. \ref{t_ladder} and Prop. \ref{p_creer}, this function is continuous for any globally hyperbolic spacetime-with-timelike-boundary. Thus, we have the elements to extend Geroch's equivalence to the case with boundary.

\begin{thm}\label{t_geroch}
For any  globally hyperbolic spacetime-with-timelike-boundary $(\overline{M},g)$, the
Geroch function $t$  in \eqref{geroch} is a (acausal)  Cauchy time function, that is, $t$ is a time function and all its levels are  acausal  Cauchy hypersurfaces.  Thus,
 if  $\bar{\Sigma}_0=t^{-1}(0)$ then $\M$ is homeomorphic to $\R\times \bar{\Sigma}_0$, and
 any other Cauchy hypersurface $\bar{\Sigma}$ is homeomorphic to $\bar{\Sigma}_0$.

 Conversely, any spacetime-with-timelike-boundary admitting a Cauchy hypersurface is globally hyperbolic.

\end{thm}

\begin{proof}
As the sum $t$ of the time functions $\ln t^-$ and $\ln (-t^+)$ is also a time function, its levels are acausal, and $t$  will be Cauchy if, for any
 inextensible past-directed causal curve  $\gamma:(a,b) \rightarrow \overline{M}$, 
\[
\lim_{s \rightarrow a} t^{-}(\gamma(s))=0, \qquad ( \hbox{and analogously} \; \lim_{s \rightarrow b} t^{+}(\gamma(s))=0),
\]
so that $\lim_{s \rightarrow a} t(\gamma(s))=-\infty$ and $\lim_{s \rightarrow b} t(\gamma(s))=\infty$ (as in the case without boundary).
To check  $\lim_{s \rightarrow a} t^{-}(\gamma(s))=0$, recall that the measure of $\M$ can be approximated by compact subsets (Prop. \ref{p_admiss}). So, for any $\epsilon>0$ there exists some compact  $K\subset \overline{M}$ with $\overline{m}(\M\setminus K)<\epsilon$ and  one has just to show that there exists $s_{0} \in (a,b)$ such that $I^{-}(\gamma(s_{0})) \cap K=\emptyset$ (as this implies
$
t^-(\gamma(s))=\overline{m}(I^{-}(\gamma(s))) \leq \overline{m}(\M)-m(K)<\epsilon,
$  for all $s \leq s_{0}$).
Assuming by contradiction the existence of a sequence $\{s_{m}\} \rightarrow a$ with $I^{-}(\gamma(s_{m})) \cap K \neq \emptyset$ for all $m$, there exists a sequence $\{r_{m}\}_{m} \subset K$ with $r_{m} \in I^{-}(\gamma(s_{m}))$; by the compactness of $K$,  $\{r_{m}\} \rightarrow r \in K$ up to a subsequence. Taking $p \in I^{-}(r)$ and some fixed $q=\gamma(c)$, one has $p \ll r_{m} \ll \gamma(s_{m}) \ll q$ for large $m$. This implies  $\gamma(a,c] \subset J^{+}(p) \cap J^{-}(q)$, that is, a past inextensible causal curve is imprisoned in the compact subset $J^{+}(p) \cap J^{-}(q)$, which contradicts \ncambios{Cor. \ref{imprisonedb}}.

To obtain the homeomorphism,  Prop. \ref{extendfield} ensures the existence of a future-directed timelike vector field $T \in \mathfrak{X}(\overline{M})$ whose restriction to $\partial M$ is tangent to $\partial M$. With no loss of generality, $T$ can be chosen unit for some auxiliary complete Riemannian metric $g_{R}$ on $\overline{M}$. Then,  $T$ is complete and, so,  its integral curves are inextensible timelike curves in $\overline{M}$.
 From the part (a), $t$ diverges along the integral curves of $T$, and each integral curve of $T$ intersects  $\bar{\Sigma}_{0}=t^{-1}(0)$ at a unique point $x\in \bar{\Sigma}_0$, which will be regarded as the initial point of each integral  curve $\gamma_x$ of $T$. 
 So, every  $p\in \overline{M}$ can be written univocally as $\gamma_x(t)$ for some  $x\in \bar{\Sigma}_0$ and $t\in \R$. Therefore, the map  $\Psi: \M\rightarrow \R\times \bar{\Sigma}_0$,
 $p\mapsto (t,x)$ is continuous, bijective and maps boundaries into boundaries. For the continuity of  $\Psi^{-1}$, recall that,  in the case without boundary, it is straightforward  by  the theorem of invariance of the domain. In the case with boundary just apply it to the double manifold $\M^d$. That is,  notice that the identification of homologous points of the two copies of  $\M$ yields also an extended hypersurface without boundary $\bar{\Sigma}_0^d$ (recall that Cor. \ref{co} (3) ensures that the boundary of  $\bar{\Sigma}_0$  is included in $ \partial M$); so, extend naturally $\Psi$ to a continuous bijective map $\Psi^d: \M^d \rightarrow \R\times \bar{\Sigma}_0^d$, which will be then a homeomorphism.

 Moreover, for any Cauchy hypersurface $\bar{\Sigma}$  clearly $\Psi(\bar{\Sigma})$ can be regarded  as the graph of a locally Lipschitz function $h:\bar{\Sigma}_0\rightarrow \R$. So, the continuos map $\bar{\Sigma}_0\ni x \mapsto (h(x),x) \in \Psi(\bar{\Sigma})$ is the required homeomorphism.

 Finally, the converse follows as in the case without boundary  (see \cite[Th.8.3.10]{Wald}, \cite[Cor. 14.39]{O} for the proof)  just taking into account that the limit curve theorem (which is the essential tool for that case) still holds here (Prop. \ref{p2.11}).
\end{proof}

\section{Stability}\label{s_Stab}

\subsection{Results for
global hyperbolicity and Cauchy temporal functions}\label{s_Stab1}
Next, our aim is to prove the following two theorems.

\begin{thm}[Stability of global hyperbolicity]
\label{T0}
Let $(\M, g)$ be a globally hyperbolic spacetime-with-timelike-boundary. Then, there exists a metric $g'>g$  (thus, necessarily with timelike boundary) which is also globally hyperbolic. \end{thm}

This result will be used to prove that such a spacetime $(\M, g)$ admits a Cauchy temporal function $\tau$ {\em with gradient} $\nabla\tau$ {\em tangent to} $\pM$ by reducing the problem to the simple case when $\pM$ is totally geodesic. However, once the existence of such functions is established, the following stronger stability result for the restrictive properties of the obtained $\tau$ will become interesting in its own right.

\begin{thm}[Stability of Cauchy temporal functions] \label{T1}
If $\tau$ is a Cauchy
temporal  function for a
globally hyperbolic metric $g$ with $\nabla\tau$ tangent to $\partial M$, then there
exists $g'$ with wider timecones, $g< g'$, such that $\tau$ is also a Cauchy temporal function for
$g'$ with $g'$-gradient tangent to $\pM$.

Moreover, $\tau$ is Cauchy temporal for any $g''\leq g'$ (with $g''$-gradient not necessarily tangent to $\pM$). In particular, its level sets  $\bSigma_{\tau_0}$, $\tau_0\in \R$,  are spacelike Cauchy hypersurfaces for any such $g''$.

 \end{thm}

\begin{rem} {\em (1) As in the case without boundary, the last assertion of Thm. \ref{T1} does not hold for an arbitrary Geroch time function $t$ (even if $t$ is smooth) because its levels may be degenerate.
 That is, independently of the presence of the boundary, we can say that a smooth time function $\tau$ for $g$ is also a smooth time function for some $g'>g$ if and only if $\tau$ is a temporal function for $g$.


(2) We will give a direct complete proof of Thm. \ref{T0} for the sake of completeness,  which may have interest to compare with previous techniques in \cite{ geroch, BM, fathi}. However, it is worth emphasizing that most of this proof
 can be skipped  (in fact, this was done in the unpublished reference \cite{S}, whose techniques are incorporated here).  Indeed, from the hypotheses of Thm. \ref{T0} one can construct a Cauchy temporal function $\tau_0$ on $\overline{M}$ {\em with no restriction on} $\pM$ just by working exactly as in the case without boundary (say, as in \cite{BernalSanchez} or \cite{MS}, see also Remark \ref{r_final}). From here, one can reproduce the proof of Thm. \ref{T0} developed in the following subsections, but skipping the subtle details associated to working there with a Geroch time function instead of a temporal one (see Remark \ref{r_A} for further details).

Summing up,  Thms. \ref{T0} and \ref{T1} (with Remark \ref{r_A}) yield  two proofs of stability in both the case with and without boundary, one of them direct (and more technical) and the other one by using an auxiliary Cauchy temporal function $\tau_0$
\ncambios{(compare with \cite[Sect. 6]{geroch} and the Section 3 in \cite{BM}, arxiv version).}
}\end{rem}

\subsection{Techniques to perturb $g$ maintaining Cauchy hypersurfaces}\label{s_Stab2}

To prove Thm. \ref{T0}, let $t$ be a  prescribed Geroch's time function on $(\M,g)$, so that, topologically $\M \equiv \R\times \bSigma$ with $\bSigma$ Cauchy  and acausal. We
will also consider a complete Riemannian metric $g_R$ on $\overline{M}$ with associated distance $d_R$. For any two non-empty
subsets $A,B\subset \M$ we denote $J(A,B):=J^+(A)\cap J^-(B)$. Moreover,  $d_R(A,B)$ will denote the $d_R$-distance between $A$, $B$ (i.e., the infimum  of the set $\{d_R(x,y): x\in A,\, y\in B\}$) and
 $d_H(A,B)$  the Hausdorff distance associated with $d_R$ between the two sets
 (i.e., the infimum of the $r\geq 0$ such that $d_R(x,B)\leq r$ for all $x\in A$ and $d_R(A,y)\leq r$ for all $y\in B$). The latter  will be applied to relatively compact subsets $A,B\subset \M$ (recall that $d_H$ becomes a true distance in the set of compact non-empty subsets). When $A=\{p\}, p\in \M$, the notation will be simplified to $A=p$.

Let $\omega$ be the 1-form metrically associated with any 
timelike vector  field  on $\overline{M}$.
For any  function $\alpha \geq 0$ on $\overline{M}$, consider the metric:
\begin{equation} \label{ealpha}
g_\alpha=g-\alpha \cdot  \omega ^2.
\end{equation}
Whenever $\alpha>0$, one has  $g< g_\alpha$.
We will add the subscript $\alpha$ to denote elements computed
with $g_\alpha$ (for example, $J_\alpha^+(p)$ is the
$g_\alpha$-causal future of $p$).  Our final aim is  to prove that some $\alpha>0$ can be chosen
small enough so that $g_\alpha$ remains causal with compact $J_\alpha (p,q)$ for all $p,q$, and thus, $(\overline{M},g_{\alpha})$ is globally hyperbolic with timelike boundary. Let us introduce a working definition.

\begin{defi}\label{dper}
 Fix an open neighborhood $U_0\subset  \M  $
bounded by two $t$-levels, i.e.  $U_0\subset J(\bSigma_{t_-},\bSigma_{t_+})$,
for some ${t_-},{t_+}$, which  will be assumed to be optimal\footnote{That is, with no loss of generality we will always assume that no $t'_-\geq
t_-$, $t'_+\leq t_+$ satisfy $U_0\subset
J(\bSigma_{t'_-},\bSigma_{t'_+})$ if some of the two inequalities is strict.}. Let
$C_0\subset  \M$ be a closed subset included in $U_0$, and choose
$\epsilon_0>0$.

A $t$-{\em perturbation of $g$ with wider timecones on $C_0$, support
in $U_0$ and $d_H$-distance smaller than $\epsilon_0$}, is any
smooth function $\alpha\geq 0$ satisfying:

(i) $\alpha>0$ on $C_0$ (thus, $g< g_\alpha$ on $C_0$),

(ii) $\alpha \equiv 0$ outside $U_0$ (i.e., $g \equiv g_\alpha$ on
$\overline{M}\backslash U_0$),

(iii) $g_\alpha$ is globally hyperbolic and it admits
 $\bSigma_{t_-}$, $\bSigma_{t_+}$ (and, thus,  all $\bSigma_t$, $t\in \R\setminus [t_-,t_+]$) as Cauchy hypersurfaces.

(iv) the following bounds hold for all $p\in J(\bSigma_{t_-},\bSigma_{t_+})$:

\begin{center}
\bcambios{ $d_H(J(p,\bSigma_{t'_+}), J_\alpha(p,\bSigma_{t'_+}))<\epsilon_0  \quad\hbox{if $t'_+\in [t(p),t_+]$}$, }

\bcambios{ $d_H(J(\bSigma_{t'_-},p), J_\alpha(\bSigma_{t'_-},p))<\epsilon_0  \quad\hbox{if $t'_-\in [t_-, t(p)]$}$, }

\bcambios{ $d_H(p, J_\alpha(p,\bSigma_{t'_+}))<\epsilon_0   \quad  \hbox{if $t'_+\in [t_-,t(p)]$ and $J_\alpha(p,\bSigma_{t'_+})\neq \emptyset$}$, }

\bcambios{ $d_H(p, J_\alpha(\bSigma_{t'_-},p))<\epsilon_0  \quad
\hbox{if $t'_-\in [t(p),t_+]$ and $J_\alpha(\bSigma_{t'_-},p)\neq \emptyset$}$. }
\end{center}


\em
\end{defi}

\begin{rem}\label{r_zcompact} {\em
The conditions in this definition imply strong restrictions on the sets \bcambios{of} type
$ J_\alpha(p,\bSigma_{t'_+}))$ such as {\em compactness}. Indeed,
$d_H(J(p,\bSigma_{t'_+}), J_\alpha(p,\bSigma_{t'_+}))$ $<\epsilon_0$ implies that
$J_\alpha^+(p) \cap \bSigma_{t'_+}$ lies in a compact set $D$ of $\bSigma_{t'_+}$. Hence
\begin{equation}\label{b}
J_{\alpha}(p,\bSigma_{t'_+})=J_{\alpha}(p,J_{\alpha}^+(p)\cap\bSigma_{t'_+})=J_{\alpha}(p,D).
\end{equation}
As the global hyperbolicity of $g_\alpha$ implies that
$J_\alpha(A,B)$ is compact  for any compact $A,B\subset \M$, the set
$J_\alpha(p,\bSigma_{t'_+})=J_\alpha(p,D)$ in (\ref{b}) is compact too. In particular, $J_\alpha^+(p) \cap \bSigma_{t'_+}=J_\alpha(p,\bSigma_{t'_+})\cap \bSigma_{t'_+}$ is also compact.

Notice also that, if $d_H(J(p,\bSigma_{t'_+}), J_\alpha(p,\bSigma_{t'_+}))<\epsilon_0$, then for any  $\alpha'$ with $0\leq\alpha'\leq \alpha$ one has $J(p,\bSigma_{t'_+}) \subset J_{\alpha'}(p,\bSigma_{t'_+})\subset J_{\alpha}(p,\bSigma_{t'_+})$, and thus, $d_H(J(p,\bSigma_{t'_+}), J_{\alpha'}(p,\bSigma_{t'_+}))<\epsilon_0$..

Finally, recall also  that if $\alpha$ is such a perturbation then so is
$\alpha/N$ for any $N>1$.
}\end{rem}


\begin{lemma}\label{L1} {\em (Existence of a
$t$-perturbation for a
compact set)}. For any compact $K_0\subset \overline{M}$, any open neighborhood
$U_0\subset J(\bSigma_{t_-},\bSigma_{t_+})$ of $K_0$ bounded by two $t$-levels $t_\pm$ and  any  $\epsilon_0>0$,
there exists a $t$-perturbation of $g$ as in Defn. \ref{dper} \bcambios{with $C_0=K_0$}.
\end{lemma}

\begin{proof}
With no loss of generality, we can assume
that the closure $cl(U_0)$ is compact  (otherwise, take any compact $K_0'$ with $K_0\subset {\rm int}(K_0')$, redefine $U_0$ as $U_0\cap {\rm int} (K_0)$ and, eventually, take a bigger $t_-$ and smaller $t_+$ to ensure that they remain optimal).
Let  $\alpha\geq 0$ be any function with support in $U_0$ and $\alpha>0$ on $K_0$. We will check that the required properties hold for
some $\alpha_m:=\alpha/m, m\in \N$ with large $m$. Recall  that, necessarily,
$K:=J_\alpha(cl(U_0), \bSigma_{t_+})\cup
J_\alpha(\bSigma_{t_-},cl(U_0))$ is equal  to $J(cl(U_0), \bSigma_{t_+})\cup
J(\bSigma_{t_-},cl(U_0))$ and, so, it  is compact.

Let us see first that $g_{\alpha_m}$ is  strongly causal for large $m$.
Assume by contradiction that the strong causality of each $g_{\alpha_m}$ is violated at some $p_m\in \M$. Necessarily,\footnote{\label{foota}Otherwise, trivially, $p_m\in J(\Sigma_-,\Sigma_+)$ and $p_m$ would admit  neighborhood $V$ in $ J(\Sigma_-,\Sigma_+)\setminus K$ where strong causality is violated. However, no causal curve starting at some $q\in V$ would cross $U_0$  (as $q\in K$ otherwise). So, strong causality would be violated in $(\bar M, g)$, a contradiction.} $p_m\in K$ and, up to a subsequence, it will converge to some $p\in K$.
Around $p$,  choose a small neighborhood $\tilde W$  with $(\tilde W,g_\alpha)$ intrinsically strongly causal,
and a smaller one  $W\subset\subset \tilde W$  which is $g_{\alpha}$-causally convex in  $\tilde W$.  Up to a finite number, $p_m\in W$ and, thus, there exists some $g_{\alpha_m}$-causal curve $\rho_m$ with endpoints in $W$ not included in $\tilde W$; indeed, all $\rho_m$ must escape $\tilde W$ and come back.
Necessarily, $\rho_m$ is entirely contained in $K$  (recall footnote \ref{foota}).  As  all  $\rho_m$  are $g_\alpha$-causal, the  sequence $\{\rho_m\}$
will
have a limit curve $\rho$ (also included in $K$) starting at $p$  \gcambios{which is ($H^1$-) causal } for all\footnote{Here, a small subtlety appears. In a strongly causal spacetime, any limit
curve of causal curves is causal. \bcambios{Instead, when the spacetime is
not strongly causal, limit curves may be non-causal. Even if one can
assure in this last case that at least one causal limit curve exists
\cite[Prop. 3.31]{Beem}, this result is not enough for our purposes, namely:
given any $h$, the subsequence $\{\rho_{h+m}\}_m$ contains
$g_{\alpha_h}$-causal curves, and the result provides a limit curve $\rho^h_\infty$
which is $g_{\alpha_h}$-causal, {\em nevertheless it depends on} $h$. Note,
however, that the {\em procedure} to find such a causal limit curve
implies directly that
the $g_\alpha$-causal limit curve (say, obtained for $h=1$) is also
$g_{\alpha_h}$-causal for all $h$. Indeed, notice first that, for each
point $p$, a ball of small radius  in any  chart is convex  for both $g$
and  all $g_{\alpha_m}$ with large $m$ (see the proof of \cite[Prop.
1.2]{P}); then, reason as in the proof of \cite[Prop. 3.31]{Beem}, that
is, apply Arzela's theorem to the sequence of curves parametrized with
respect to an auxiliary complete Riemannian metric, and work locally
around each limit point on a convex neighborhood $U$ as above.}} $g_{\alpha_m}$, which implies that $\rho$ is  $H^1$-causal for $g$ too. Indeed, the velocity of $\rho$ (whenever it is differentiable)  must be causal for all $g_m$
 and, as
 $\{g_{\alpha_m}\}\rightarrow g$ (uniformly on all $\M$
  by construction), it is causal for $g$.
 Moreover, since $\rho$ remains in the compact set $K$, a contradiction with the strong causality of $g$ appears: either  $\rho$ is closed or it is inextensible but imprisoned in a compact subset.

 Now, let us check that, for any $m$ such that  $g_{\alpha_m}$ is strongly causal, the property {\em (iii)} in Defn. \ref{dper} holds.  Indeed, no inextensible future-directed  $g_{\alpha_m}$ -causal curve  $\gamma: \R\rightarrow \overline{M}$ can be partially  imprisoned to the future (resp. past) in the compact set $K$. So, either $\gamma$ does not
 intersect $K$ (and, trivially, will cross $\bSigma_{t_+}, \bSigma_{t_-}$ once) or there is a last point $\gamma(s_0^+)$ (resp. $\gamma(s_0^-)$) such that $\gamma|_{(s_0^+,\infty)}$ (resp. $\gamma|_{(-\infty,s_0^-)}$) does not intersect $K$. In this case, as $\gamma(s_0^+)\in J^-(\bSigma_{t_+})$ (resp. $\gamma(s_0^-)\in J^+(\bSigma_{t_-})$), then $\gamma|_{[s_0^+,\infty)}$ (resp. $\gamma|_{(-\infty, s_0^-]}$) touches $\bSigma_{t_+}$ (resp. $\bSigma_{t_-}$) exactly once.
 Clearly, $\gamma|_{(s_0^-, s_0^+)}$ cannot touch $\bSigma_{t_+}$ (resp. $\bSigma_{t_-}$)
 at some $\gamma(s_0)$ with $s_0\in (s_0^-, s_0^+)$ because in this case $\gamma$ would be contained in $I^+(\bSigma_{t_+})$ (resp. $I^-(\bSigma_{t_-})$) beyond $s_0$ towards $+\infty$ (resp. $-\infty$).

\bcambios{To prove that  {\em (iv)} in Defn. \ref{dper}  also holds for  $\alpha_m$ with large $m$, in the remainder the value of  $m$ will be bigger than the value $m_0$ obtained above, which  ensures the strong causality and, then,  the global hyperbolicity, of $g_{\alpha_m}$.}

\bcambios{Assume by contradiction that the first inequality in {\em (iv)} does not hold for  $\alpha_m$ with big $m$,
that is,  there exists $q_m\in J(\bSigma_{t_-},\bSigma_{t_+})$ and
  $t^m_+ \in [t(q_m), t_+]:$
\begin{equation}
\label{e_ef2}
 d_H(J(q_m,\bSigma_{t^m_+}),J_{\alpha_m}(q_m,\bSigma_{t^m_+}))\geq\epsilon_0, 
 \end{equation}
for  diverging values of  $m\geq m_0$. Necessarily, $q_m\in K$ and there exists
$r_m \in J(q_m,\bSigma_{t^m_+})$,
$r_m' \in J_{\alpha_m}(q_m,\bSigma_{t^m_+})$ such that $ d_R(r_m,r'_m)  =
d_H(J(q_m,\bSigma_{t^m_+}),J_{\alpha_m}(q_m,\bSigma_{t^m_+}))\geq\epsilon_0$
for each $m$. Up to  subsequences, we can assume that $\{q_m\}, \{r_m\}, \{r_m'\}$ converge   to
some $q,r,r' \in K$, resp., and $\{t^m_+\}\rightarrow t_+'  \in [t(q), t_+]$.
Necessarily, $ d_R(r,r')  \geq \epsilon_0$ and
$r'\not\in J(q,\bSigma_{t'_+})$. However, the
$g_{\alpha_m}$-causal curves $\gamma_m$ from $q_m$ to $r'_m$ will
have a $g$-causal limit curve $\gamma$ starting at $q$ and arriving  at $r'$, that is,
 $r'\in J^+(q)$.  As $t(r')\leq \limsup t(r'_m)\leq \limsup t^m_+=t'_+$, the contradiction
$r'\in J^+(q,\bSigma_{t'_+})$) follows.}


\bcambios{Now, assume that the third  inequality in {\em (iv)} does not hold (the other two cases would be analogous), that is,  there exists
$
t^m_+ \in [t_-,t(q_m))$ (so, $J(q_m,\bSigma_{t^m_+})\subset \{q_m\}$) with
$$ d_H(q_m,J_{\alpha_m}(q_m,\bSigma_{t^m_+}))\geq\epsilon_0,
$$ 
for diverging values of $m\geq m_0$. Up to subsequences, $\{q_m\}$ converges to
some
$q \in K$ and $\{t^m_+\}\rightarrow t_+'$ ($\in [t_-,t(q)]$).
Take a point $r'_m\in  J_{\alpha_m}(q_m,\bSigma_{t^m_+})))$
which realizes the Hausdorff distance to
 $\{q_m\}$, thus
$d_R(q_m,r'_m)\geq\epsilon_0$.
Choose any second point $r''_m\in J_{\alpha_m}^+(r'_m) \cap \bSigma_{t^m_+}$ (which must exist from the definition of $J_{\alpha_m}(q_m,\bSigma_{t^m_+})$, possibly $r_m''=r_m'$). Up to  subsequences, $\{r_m'\}\rightarrow r'\in K$ and $\{r_m''\}\rightarrow r''\in \bSigma_{t'_+} \cap K$; moreover,
\begin{equation}
\label{e_ef}
d_R(q,r')\geq \epsilon_0>0.
\end{equation}
Now, the sequence  $\{\gamma_m\}$ of future-directed $g_{\alpha_m}$-causal curves from
$q_m$ to $r_m''$ through $r_m'$ must have a  limit curve $\gamma$  which is
$H^1$-causal for $g$. As $g$ is
globally hyperbolic, $\gamma$ goes from $q$ to $r''$ through $r'$. Moreover, as $t$ is a $g$-time function,  $t(q)\geq t(r'')=t'_+$ and the equality must hold (recall
$t(r''_m)=t^ m_+ \leq t(q_m)$ by hypothesis). As $\bSigma_{t(q)}$ is $g$-acausal; $\gamma$ must be a constant and $q=r'=r''$, in contradiction with \eqref{e_ef}.}
\end{proof}

\begin{prop}\label{L2} {\em (Existence of a
$t$-perturbation for a
strip)}. Let $C$ be the closed strip  $J(\bSigma_{t_1},\bSigma_{t_2})$, $t_1<t_2$, let   $U=(t_-,t_+)\times \bSigma$ be a neighborhood of $C$ (i.e. $t_-<t_1, t_2<t_+$), and choose
$\epsilon>0$. Then, there exists a $t$-perturbation of $g$ for $C, U$ and $\epsilon$ as in
Defn.~\ref{dper}. 
\end{prop}

\begin{proof}
 We can assume that  $\bSigma$ is not compact (otherwise, the result would follow directly from Lemma \ref{L1}). Let $\{B_m\}$ be  an exhaustion of $\bSigma$ by compact subsets, i.e. $B_m \subset {\rm int}(B_{m+1})$ (where the interior of $B_{m+1}$ is regarded as a subset of the topological space $\bSigma$)
and $\bSigma  =\cup_m B_m$. By convenience, put $B_0  :=B_{-1}  :=\emptyset$, and let $K_m:=[t_1,t_2]\times B_m$ 
and
$U_m:=(t_-,t_+)\times {\rm int}(B_{m+1})$ for all $m\in \N \cup \{0,-1\}$; recall $K_m\subset U_m$.

Let us construct $\alpha$ inductively as follows.
Consider first a perturbation
$\alpha_1$ obtained by applying Lemma \ref{L1} to the compact set
$K_1$,  its neighborhood $U_1$ and putting $\epsilon_0  =:\epsilon_1$  equal to $\epsilon/2$. Set $\sigma(1)=1$.  Assuming
inductively that $\alpha_m$ has been defined (and, thus, the corresponding globally hyperbolic metric $g_{\alpha_m}$), let $\sigma(m+1)$ be the first integer greater
than
$\sigma(m)$ such that
\begin{equation}
\label{Ej}
J_{\alpha_m}(cl(U_{\sigma(m)}),\bSigma_{t_+})\cup
J_{\alpha_m}(\bSigma_{t_-},cl(U_{\sigma(m)}))\subset U_{\sigma(m+1)-1}\cup \bSigma_{t_-}\cup \bSigma_{t_+}
\end{equation}
($\sigma(m+1)$ exists as the subset of the left-hand side is compact). Now, obtain $\alpha_{m+1}$ by
applying Lemma \ref{L1} to the metric $g_{\alpha_m}$, the compact
set $K_{\sigma(m+1)}\setminus {\rm int}(K_{\sigma(m)})$,  its
neighborhood $U_{\sigma(m+1)}\setminus cl(U_{\sigma(m)- 2  })$
and choosing some positive
$\epsilon_0  =: \epsilon_{m+1}   \leq \epsilon/2^{m+1}$ small so that, the relations  \eqref{Ej} still hold  when the causal futures and pasts are computed with $g_{\alpha_{m+1}}$ instead of $g_{\alpha_{m}}$. Namely, $\epsilon_{m+1} $
   is taken also  smaller than
 the minimum of:
\begin{eqnarray}\label{e_arr1} d_R\left([t_-,t_+]\times \dot B_{\sigma(m+1)-1}
\, , \; J_{\alpha_m}(cl(U_{\sigma(m)}),\bSigma_{t_+})\cup
J_{\alpha_m}(\bSigma_{t_-},cl(U_{\sigma(m)}))\right) , \\ \label{e_arr2}
d_R\left(
[t_-,t_+]\times \dot B_{\sigma(m)-1}
\, , \; J_{\alpha_{m}}(cl(U_{\sigma(m-1)}),\bSigma_{t_+})\cup
J_{\alpha_{m}}(\bSigma_{t_-},cl(U_{\sigma(m-1)}))\right) \end{eqnarray}
 (where $\dot U$ denotes the
  topological boundary of the
corresponding subset, and the distance in the second line is taken into account only for $m>1$). Notice that both (\ref{e_arr1}) and (\ref{e_arr2}) are positive, as both are distances between   compact  disjoint sets; in particular, the second distance is lower bounded
by
$$
d_R\left(
[t_-,t_+]\times \dot B_{\sigma(m)-1}
\, , \; J_{\alpha_{m-1}}(cl(U_{\sigma(m-1)}),\bSigma_{t_+})\cup
J_{\alpha_{m-1}}(\bSigma_{t_-},cl(U_{\sigma(m-1)}))\right)
- \epsilon_m>0.
$$
\ncambios{In fact,  if $A, B, C$ are three compact subsets of $\M$ then
$
d_R(A,C)\leq d_R(A,B)+d_H(B,C)$,
and apply this to
$A=[t_-,t_+]\times \dot B_{\sigma(m)-1}$,
$B=J_{\alpha_{m}}(cl(U_{\sigma(m-1)}),\bSigma_{t_+})\cup
J_{\alpha_{m}}(\bSigma_{t_-},cl(U_{\sigma(m-1)})) $ and
$C=J_{\alpha_{m-1}}(cl(U_{\sigma(m-1)}),\bSigma_{t_+})\cup
J_{\alpha_{m-1}}(\bSigma_{t_-},cl(U_{\sigma(m-1)}))$.}

These  requirements for $\epsilon_{m+1}$ ensure, for all
$k\geq 1$:

\begin{equation}\label{Ej2}
\begin{array}{c}
J_{\alpha_{m+ 2 }}(cl(U_{\sigma(m)}),\bSigma_{t_+})=
J_{\alpha_{m+ 2 +  k}}(cl(U_{\sigma(m)}),\bSigma_{t_+}), \\
J_{\alpha_{m+ 2 }}(\bSigma_{t_-},cl(U_{\sigma(m)}))=
J_{\alpha_{m+  2+ k}}(\bSigma_{t_-},cl(U_{\sigma(m)})).
\end{array}
\end{equation}
\ncambios{Recall, for the first eqn. \eqref{Ej2}: $J_{\alpha_{m+1}}(cl(U_{\sigma(m)}),\bSigma_{t_+}) \subset U_{\sigma(m+1)-1}$ (once determined the value of $\sigma(m+1)$, by using the bound (\ref{e_arr1}) for $\epsilon_{m+1}$), $J_{\alpha_{m+2}}(cl(U_{\sigma(m)}),\bSigma_{t_+}) \subset U_{\sigma(m+1)-1}$ (use (\ref{e_arr2})) and $J_{\alpha_{m+2+k}}(cl(U_{\sigma(m)}),\bSigma_{t_+}) = J_{\alpha_{m+2}}(cl(U_{\sigma(m)}),\bSigma_{t_+})$ (because
$U_{\sigma(m+1)-1}\cap (U_{\sigma(m+3)}\setminus cl(U_{\sigma(m+2)-2}) =\emptyset$).}

Each metric $g_{\alpha_m}$ satisfies the properties (iii) and (iv)  in Def. \ref{dper} and it can be regarded as a $t$-perturbation of the original metric $g$ with perturbing function $\alpha_1+\dots +\alpha_m$. Moreover, clearly, for each compact subset $Z\subset \M$ there exists some $m_0\in \N$ such that (a) $\alpha_m(Z)\equiv 0$ for all $m\geq m_0$ and (b) the property \eqref{Ej2} holds if $cl(U_{\sigma(m)})$ is replaced by $Z$.
Assertion (a) allows to define $\alpha$   as the locally finite sum $\alpha:= \sum_m \alpha_m$, and  let us check that, then, $g_\alpha$ satisfies the properties stated in Defn. \ref{dper}.
By construction,  $g_\alpha$ satisfies (i) and (ii). For (iii), apply the assertion (b) to $Z=\{p\}$, $p\in U_0$ and recall that all $g_{\alpha_m}$ satisfied (iii).
Finally, for
(iv), putting  $g_{\alpha_0}:=g$,
$$d_H(J(p,\bSigma_{t'_+}), J_{\alpha}(p,\bSigma_{t'_+}))\leq \sum_{k=0}^{m-1} d_H(J_{\alpha_{k}}(p,\bSigma_{t'_+}), J_{\alpha_{k+1}}(p,\bSigma_{t'_+})) \leq  \sum_{k=0}^{m-1} \epsilon_k  <\epsilon
$$
and the result follows taking into account the assertion (b) again.
\end{proof}



\subsection{Proofs of the main results} \label{s_Stab3}

\noindent {\em Proof of Thm. \ref{T0}.}
Consider the closed strips
$C_{k}=[k-1, k+2]\times \bSigma$ and open ones $U_{k}=$  $(k-2, k+3)\times \bSigma$ for each integer $k\in \Z$, and choose the complete Riemannian metric $g_R$ so that $d_R(\bSigma_{k+k'/3}, \bSigma_{k+(k'+1)/3})>1$ for all $k\in\Z, k'\in\{0,1,2\}$.
 Construct a
 $t$-perturbation $\alpha_k$ for each  $C_k$, $U_k$ as in Prop. \ref{L2} with $\epsilon_0=1/2$ for all $k$.
Now, take any smooth function $\alpha >0$ on $\M$ such that $\alpha<\alpha_k$ on each
strip $C_k$. \ncambios{
To construct such an $\alpha$, consider  the (continuous, positive) function $\hat \alpha_k$ obtained as the minimum of $\{\alpha_{k-1}, \alpha_{k}, \alpha_{k+1}\}$
on each strip
$[k,k+1]\times \bar\Sigma$, $k\in\Z$. Merge all the functions $\hat\alpha_k$'s on the strips into a single function $\hat \alpha$ on all $\M$ just by choosing the minimum between $\hat\alpha_k$ and $\hat\alpha_{k+1}$ on each slice
$\bSigma_{k}$. As $\hat \alpha$ is positive and lower semi-continuous, one can choose a smooth $\alpha$ satisfying $0<\alpha<\hat\alpha$ by means of a standard partition of unity argument.}


Let us check that $g_\alpha$ can be chosen as the required metric $g'$. Trivially, $g_\alpha>g$. In order to prove that $(\overline{M},g_{\alpha})$ is globally hyperbolic, consider first the following.

\bigskip
\noindent
{\em Claim.} Let $p\in \M$ and $q\in J_\alpha^+(p)$.
Any future-directed $g_{\alpha}$-causal curve $\rho$ from $p$ to $q$ is included in the strip $(t(p)-1,t(q)+1))\times \bSigma$. In particular, $
J_\alpha^+(p)\cap \bSigma_{t_0}=\emptyset$ if $t_0\leq t(p)-1$.

Moreover, for each  $t_0\in \R$,
$
J_\alpha(p,\bSigma_{t_0})$ and $
J_\alpha(\bSigma_{t_0}, p)$ are compact.

 \smallskip
Assuming the claim, the result follows easily. Indeed, (strong) causality holds at any $p\in \M$ because, otherwise, any closed causal loop $\rho$ would be contained in the strip
$(t(p)-1,t(p)+1))\times \bSigma$. However, this strip is included in the closed strip
$C_{k_p}$, where $k_p$ is the integer part of $t(p)$. As $g_\alpha \leq g_{\alpha_{k_p}}$ on $C_{k_p}$, the closed $g_{\alpha}$-causal curve $\rho$ is also causal for $g_{\alpha_{k_p}}$, in contradiction with the global hyperbolic character of $g_{\alpha_{k_p}}$. Finally, notice that $J_\alpha^+(p)\cap J_\alpha^-(q)$ is the intersection of the compact sets $J_\alpha(p,\bSigma_{t_q+1})$ and $
J_\alpha(\bSigma_{t(p)-1}, q)$.

\bigskip
\noindent
{\em Proof of the Claim.} For the first assertion, assume by contradiction that there exists a first point $q'$ with $t(q')=t(p)-1$
crossed by $\rho$ (a similar contradiction would appear at the last point where $\rho$ left the region $t\leq t(q)+1$).
Let $p'$ be the last point before $q'$ crossed by $\rho$ with $t(p')=t(p)$; trivially,
$q'\in J_{\alpha}^+(p')$. The portion of $\rho$ from $p'$ and $q'$ is entirely contained in the strip $C_{k_p}$, thus $q'\in J_{\alpha_{k_p}}^+(p')$.
So, putting $t_0=t(q')$ ($=t(p)-1$) and taking into account Defn. \ref{dper} (iv):
\begin{equation}
\label{e_lat0}
d_H(p', J_{\alpha_{k_p}}(p',\bSigma_{t_0}))<\epsilon_0=1/2.
\end{equation}
This is a contradiction because $q'\in  J_{\alpha_{k_p}}(p',\bSigma_{t_0})$ and, by our choice of $g_R$,
$d_R(p',q')>1$.

 For the last assertion, let us reason for $J_\alpha(p,\bSigma_{t_0})$. If $t_0 \leq t(p)$, then $J_\alpha(p,\bSigma_{t_0})$ is entirely included in the strip $[t(p)-1,t(p)+1]\times \bSigma$, which is included in the strip $C_{k_p}$. By applying Remark \ref{r_zcompact} to $g_{\alpha_{k_p}}$ we deduce that $J_{\alpha_{k_p}}(p,\bSigma_{t_0})$ is compact. Then, taking into account that $\alpha\leq \alpha_{k_p}$, necessarily $J_{\alpha}(p,\bSigma_{t_0})\subset J_{\alpha_{k_p}}(p,\bSigma_{t_0})$, and so, $J_{\alpha}(p,\bSigma_{t_0})$ is also compact. Assume inductively that compactness hold when $t_0\leq t(p)+k$, and consider
$t(p)+k<t_0\leq t(p)+k+1$. Then, both
$ J_\alpha(p,\bSigma_{t(p)+k})$ and
$D_0:= J^+_\alpha(p) \cap \bSigma_{t(p)+k}$
are compact by hypothesis, and
$J_\alpha (D_0, \bSigma_{t_0})$ must lie in the
strip $C_{k_p+k}$. So (using analogously the
global hyperbolicity of $g_{\alpha_{k_p+k}}$ and
$\alpha \leq \alpha_{k_p+k}$),
$J_\alpha (D_0, \bSigma_{t_0})$ is compact and
the result follows by noticing
$J_\alpha(p,\bSigma_{t_0})=
J_\alpha(p,\bSigma_{t(p)+k})\cup
J_\alpha (D_0, \bSigma_{t_0})
$.
 \cvd

\begin{rem}\label{r_bound_alpha}{\rm
Recall that all the perturbed metrics constructed so far are of the form~\eqref{ealpha}. Indeed, the globally hyperbolic metric has been obtained for some $\alpha>0$. Obviously, any other metric as in \eqref{ealpha} constructed with some function $\alpha'>0$ such that $\alpha'\leq \alpha$ would be globally hyperbolic too (the Cauchy hypersurfaces for $g_\alpha$ would be also Cauchy for $g_{\alpha'}$ as $g_{\alpha'}\leq g_{\alpha}$).
}\end{rem}


Next, let us focus on the stability of Cauchy temporal functions.

\smallskip

\noindent {\em Proof of Thm. \ref{T1}}.
The proof of Thm. \ref{T0} will be mimicked by  choosing $\omega = d\tau$,  which  is metrically equivalent to the timelike vector field $-\partial_t/\Lambda$ in the  orthogonal splitting $(\R\times \bSigma, g\equiv -\Lambda d\tau^2 + g_\tau)$ associated with the Cauchy temporal function $\tau$. 
So,
\begin{equation}\label{E4}
g_\alpha= g-\alpha d\tau^2 =-(\Lambda+\alpha)d\tau^2 + g_\tau.
\end{equation}
 Recall that, now, {\em for all the metrics $g_\alpha$, $\alpha \geq 0$, the slices $\bSigma_\tau$ remain spacelike}. This is an important  difference with Thm. \ref{T0}, where all the slices for Geroch function $t$ might be non $g_\alpha$-acausal for any $\alpha$. So, in both Lemma \ref{L1} and Lemma \ref{L2},  not only the slices $\bSigma_t$, $t\in \R\setminus (t^-,t^+)$ are Cauchy but also all the slices
 with
 $t\in (t^-,t^+)$ \ncambios{(this follows from the proof of those lemmas, but can be also noticed because any closed spacelike hypersurface contained in the region between two disjoint Cauchy hypersurfaces  is Cauchy too, see for example \cite[Corollary 11]{BS03}).}
 Therefore, we can follow  the proof of Thm. \ref{T0} but we do not need to replace $\tau$ as was done for the function $t$ (inductively replaced by some $t_m$ to construct a Cauchy time $t_\infty$): simply, $\tau$ remains equal at all the steps of the proof and becomes a Cauchy temporal function for $g_\alpha$. Recall also that, as the expression \eqref{E4} remains valid for any perturbed metric $g_\alpha$, all the $g_\alpha$-gradients of $\tau$ remain tangent to $\pM$.
$\Box$

\begin{rem}\label{R_que_pena} {\em
 Notice that the metrics $g_\alpha$ in \eqref{E4} \bcambios{(or more generally (\ref{ealpha}))} with $\alpha>0$ not only satisfy $g<g_\alpha$ but also $-g_\alpha(v,v)>-g(v,v)$ whenever $v$ is $g$-causal.
\gcambios{Morevoer, for} any $\alpha\geq 0$, one has:
$$
\nabla \tau=-\frac{\partial_{\tau}}{\Lambda}, \quad
\nabla_{g_\alpha}\tau=-\frac{\partial_{\tau}}{\Lambda+\alpha}, \quad \hbox{thus,} \quad
 -g(\nabla_{g_\alpha}\tau, \nabla_{g_\alpha}\tau)= -d\tau(\nabla_{g_\alpha}\tau)=\frac{1}{\Lambda+\alpha}.$$
Therefore, if the Cauchy temporal function $\tau$ is steep for $g_\alpha$ (i.e.,  $\Lambda + \alpha\leq 1$) then, it is also  steep for $g$
($\Lambda\leq 1$). Conversely, if $\tau$ is steep for $g$ and $\alpha$ is chosen bounded $\alpha\leq C$ then the Cauchy temporal function $\sqrt{1+C} \, \tau$
is steep for $g_\alpha$.  }

\end{rem}

\begin{rem}\label{r_stably} {\em
  In connection with the proof of Prop. \ref{p_stable} (implication $2 \Rightarrow 3$), it is clear now that, if one considers any stably causal spacetime and $\tau$ is a temporal function for $g$, then $\tau$ is also temporal for all the metrics with wider cones $g_\alpha$ constructed in \eqref{E4}, for any $\alpha\geq 0$.
}\end{rem}

\section{The Cauchy orthogonal decomposition.} \label{s4}
 Throughout this section,
$(\overline{M},g)$ will be
 globally hyperbolic 
 with timelike boundary.

\subsection{Simplification of the problem by using stability  }
Next, our aim is to show that the problem of existence of $\tau$ in Thm. \ref{t0} can be reduced to the case of a new metric $g^*$ with a simple product behaviour
around $\pM$.  First,  let us globalize Gaussian coordinates. 

\begin{lemma}[Existence of a global tubular neighborhood] There exists a smooth function $\rho: \partial M \rightarrow \R$, $\rho>0$,
such that the orthogonal exponential map
\begin{equation}
\label{Exp}
\exp^\perp: \{(\hat p, s)\in \partial M\times [0,\infty): 0\leq s < \rho(\hat p)\}\rightarrow\overline{M} , \qquad (\hat p,s)\mapsto \exp_{\hat p}(s\,N_{\hat{p}})
\end{equation}
is a diffeomorphism onto its image $E$, where $N_{\hat{p}}$ is the pointing-inward $\partial M$-orthogonal unit vector at $\hat{p}$. This will be called {\em tubular neighborhood} of $\pM$.
\end{lemma}
\begin{proof}
The technique is standard. Start with a Gaussian neighborhood around each $\hat p\in \M$ with normal coordinate $s\in [0,\epsilon_{\hat p})$, take an exhaustion by compact sets $\{\hat K_m\}_m$, $\hat K_m\subset $ int$(\hat K_{m+1})$
(even if $\pM$ has infinitely many connected components, this can be done by  intersecting  $\pM$ with  closed balls of radius $m$ for a  complete Riemannian metric on $\M^d$ centered in a freely chosen point), determine  $\epsilon_{\hat p_m}$ such that  $\exp^\perp$ satisfies the required properties on $\hat{K}_m\times [0,\epsilon_{\hat p_m})$ and choose  $\rho>0$ such that $\rho<\epsilon_{\hat p_m}$ on $\hat{K}_m\setminus \hat K_{m-1}$.
\end{proof}

Elements on $E$ and their preimages by $\exp^\perp$ will be identified with no further mention. In particular, $E$ is endowed with the ($\pM$-orthogonal, geodesic) vector field $\partial_s$.



\begin{prop}\label{R1} Let  $g'>g$ be globally  hyperbolic (as obtained in Thm. \ref{T0}). Then, there exist another Lorentzian metric $g^*$ on $\M$ satisfying:

(a) $g\leq g^* <g'$ on $\M$, and so, $g^*$ is also globally hyperbolic.

(b) On some $g^*$-tubular neighborhood $E$ of $\pM$, each $p\in E$ admits a neighborhood of type $\hat U \times [0,\epsilon)\subset E$, where $g^*$ is a product metric $g^*= \hat g_0 +ds^2$, being $\hat g_0$  a Lorentzian metric on $\hat U\subset \pM$.
\end{prop}

\begin{proof}
As $g'>g$, we can choose a metric $\hat g_0$ on $\pM$ such that
\begin{equation}\label{e_z1}
g|_{\pM} < \hat g_0 < g'|_{\pM}
\quad \hbox{and}
\quad  -\hat g_0(v,v)>-g(v,v)
\; \hbox{for all} \,
g\hbox{-causal} \; v\in T(\pM ),
\end{equation}
and such that the (locally) product metric
\begin{equation}\label{e_z2}
g_0:= \hat g_0 + ds^2
\end{equation}
defined in some tubular neighborhood $E'$ of $\pM$,  satisfies
\begin{equation}\label{e_z3}
g<g_0< g' \quad  \hbox{on} \, E' \;
\quad \hbox{and}
\quad  -\hat g_0(v,v)>-g(v,v)
\; \hbox{for all} \,
g\hbox{-causal} \; v\in TE'.
\end{equation}
Indeed, notice that $g|_{\pM} < g'|_{\pM}$ are globally hyperbolic metrics on the manifold  $\pM$ (Thm. \ref{t_ladder} (3)) and $\hat g_0$ can be constructed as a metric $g_\alpha$, $\alpha>0$, in Remark \ref{R_que_pena} (1), ensuring both assertions in \eqref{e_z1}. Moreover, as $\partial_s$ is  unit and $g$-orthogonal with both $g_0$ and $g$, this also ensures $g<g_0$ on all the points of $\pM$ and (choosing a smaller $\alpha$ if necessary) $g_0<g'$ on $\pM$. Thus, reducing $E'$ if necessary, both assertions in \eqref{e_z3} also hold.

Now, take any smaller tubular neighborhood $E$ with $cl(E)\subset E'$ (say, the associated with the function $\rho/2$ in equation \eqref{Exp}). Consider the covering $\{E',\M\setminus E\}$ of $\M$, choose a subordinate partition of unity $\{\mu, 1-\mu\}$ with supp$(\mu)\subset E'$ and construct
 \begin{equation}
 \label{EG} g^*=\mu g_0 + (1-\mu)g .
 \end{equation}
By the second assertion in \eqref{e_z3}, $g\leq g* \leq g_0$; thus (by the first assertion),  $g^*$ fulfills (a). The \bcambios{requirement (b)}
\bcambios{holds} because $g^*=g_0$ on the whole $E$.
 \end{proof}

\begin{rem}{\em
On $E$, both $g$ and $g^*$ can be written in Gaussian coordinates as in \eqref{e_Gaussian}; however, the dependence on $s$ of $g_{ij}(\hat x,s)$ is dropped for $g^*$.}
\end{rem}


 \begin{lemma}[Reduction to a local product around $\pM$]\label{L_reduction} For any  $g^*$ as in Prop.~\ref{R1}:

 (a) if $\tau$ is  Cauchy temporal for $g^*$, then so is it for $g$.

(b) 
if the $g^*$-gradient $\nabla^* \tau$ of $\tau$ is tangent to $\pM$, then so is the $g$-gradient $\nabla \tau$.
\end{lemma}

\begin{proof}
(a) $g\leq g^*$ implies that $\tau$ is $g$-Cauchy temporal
 (the $g$-orthogonal to $\nabla \tau$ is tangent to a $\tau$-slice, which is $g^*$-spacelike, and thus, $g$-spacelike).

(b) On $\pM$,
$g(\nabla \tau,\partial_s)=d\tau(\partial_s)= g^*(\nabla^*\tau , \partial_s)=0$, thus $\nabla \tau$ is tangent to $\pM$.
\end{proof}

\begin{rem}\label{R_pena} {\em As a summary, Thm. \ref{T0}, Prop. \ref{R1} and Lemma \ref{L_reduction} yield:
{\em in order to obtain a Cauchy temporal function $\tau$ for $g$ with $\nabla \tau$ tangent to $\pM$ (as required for Thm. \ref{t0}), one can assume, with no loss of generality, that $g$ satisfies the local product property stated for $g^*$ in Prop. \ref{R1} \bcambios{(b)} } (otherwise, work with $g^*$ itself).
}\end{rem}

\subsection{Proof of Theorem \ref{t0}}\label{s4.1}

%
%
%
%
%
%
%

\noindent {\it Proof}. \bcambios{First, let us prove the first assertion of Thm. \ref{t0}, that is, any globally hyperbolic spacetime-with-timelike-boundary admits a temporal function $\tau: \M \rightarrow \R$ whose gradient $\nabla \tau$ is tangent to $\pM$. Take} two copies of $(\M,g)$ and identify their homologous points along the boundary $\pM$ in order to obtain the double manifold $\M^d$.
Now, $\M^d$ inherits not only a structure of smooth manifold (without boundary) but also a smooth metric $g^d$. Indeed,  $g$ is assumed to satisfy the  properties stated for $g^*$ in Prop. \ref{R1} (by  Remark \ref{R_pena}). So,  the local product structure of any tubular neighborhood $E$ of $\pM$ yields directly the smoothness of both the exponential and the metric around   $\pM \subset M^d$ (just work in  Gaussian coordinates). Moreover,  the natural reflection $i: \M^d \rightarrow \M^d$ (which maps each point in its homologous one) becomes an isometry for $g^d$. Notice that $g^d$ is globally hyperbolic because any pair of homologous Cauchy hypersurfaces with boundary in each copy merges into a continuous Cauchy hypersurface without boundary of $(\M^d,g^d)$.  So, the main result in \cite{muller2016} (applied with 0 invariant Cauchy hypersurfaces) ensures the existence of a Cauchy temporal function $\tau^d$ on $\M^d$ which is invariant by $i$. Therefore, the restriction $\tau$  of $\tau^d$ to $\M$ becomes a Cauchy temporal function, and its gradient must be tangent to the boundary as $i_*(\nabla \tau)=\nabla\tau$ on the set  $\pM$ (of fixed points for $i$), \bcambios{as required}.

\bcambios{
As stated in Theorem \ref{t0},  some consequences of the existence of a Cauchy temporal function $\tau$ with gradient $\nabla \tau$ tangent to $\pM$ are: (a) $\M$ splits as a product manifold $\R \times \bSigma$, where $\bSigma$ is a spacelike Cauchy hypersurface-with-boundary, and (b)
	the metric $g$ splits globally $g=-\Lambda d \tau^2+g_\tau$ as in Eq. \ref{e_split}.
In fact, choosing a slice $\bSigma=\tau^{-1}(0)$ and moving it with the flow of $-\nabla \tau/|\nabla\tau|^2$ yields $(a)$ and $(b)$ easily (as in Cor. \ref{cor} or  \cite[Prop.~2.4]{BernalSanchez}).
}


 \bcambios{Next, let us check the isometric embeddability of $(\M,g)$ into $\LL^N$}. Notice that the technique in \cite{MS} works whenever a steep Cauchy function $\tilde \tau$ is obtained (that is,  the requirement of tangency of $\nabla \tilde \tau$ to $\pM$ necessary for the orthogonal splitting can be dropped now).
In order to find such a $\tilde \tau$, the same procedure as in \cite{MS} can be used. \ncambios{This is straightforward because the technique in  \cite{MS} is based in local functions type $\j_{p}(\cdot):=\exp(-1/d(p,\cdot)^2)$, where $d$ is the Lorentzian distance. Moreover, the required local functions can be chosen in a very flexible way (for example, taking the distance associated with a $C^0$-close flat metric $h_0$ with $-h_0(v,v)>-g(v,v)$ on any $g$-causal $v$). So, the constructive technique works in the same way even if, now,  $p$ may belong to $\pM$.}

\bcambios{The remainder of the assertions in Thm. \ref{t0} concerns $M$ and $\pM$.
By  Thm. \ref{t_ladder}, $M$ is causally continuous and  $\pM$ is globally hyperbolic. Morevoer, the restriction of $\tau:\M \rightarrow \R$ to $M$ also determines an orthogonal splitting on $(M,g)$; so,  the spacetimes in the class of the causally continuous ones that can be regarded as the interiors of globally hyperbolic spacetimes-with-timelike-boundary, do admit a global orthogonal splitting too. Finally, the restriction of $\tau$  to the boundary becomes a Cauchy temporal function on  $\pM$ trivially and its levels $\tau^{-1}(\tau_0)\cap \pM$ are acausal in $\M$, as so are the whole levels $\tau^{-1}(\tau_0)$. } $\Box$

\begin{rem}\label{r_final} {\em   As emphasized above, the proof of isometric embeddability in $\LL^n$ does not require the steep Cauchy temporal  function $\tau$ for $g$ having gradient tangent to $\pM$ and, thus, it can be proven by using the same techniques as in the case without boundary.
 Analogously, other problems on smoothability can be proved directly  as in the case without boundary. For example:

 (a) Geroch's topological splitting (Thm. \ref{t_geroch}) can be improved into a smooth one just noticing that the same procedure as in \cite{BS03} gives a spacelike  Cauchy hypersurface $\bSigma$, and moving it by using the flow of any timelike vector field $T$ tangent to $\pM$ (as obtained in Prop. \ref{extendfield} (iii), and choosing $T$ complete).

 (b) The technique in \cite{BernalSanchez} directly gives also a Cauchy temporal function $\tau_0$ and, this, a smooth foliation of $\M$ by Cauchy hypersurfaces (however, they are not necessarily orthogonal to $\pM$).

 (c) The extensions of the classical equivalences of the notion of
  stably causal spacetime to the case with boundary rely also on the case without boundary (as described in Prop. \ref{p_stable}).

  (d) As in the case without boundary (see \cite{MS}), stably causal spacetimes-with-boundary can be  conformally embedded in some Lorentz-Minkowski spacetime $\LL^N$ for big $N$.

}
\end{rem}

\begin{rem}\label{r_A} {\em
Finally, let us revisit how, if one assumed the existence of the Cauchy temporal function $\tau_0$ (as in Remark \ref{r_final} (b) above), the proof of Thm. \ref{T0} is widely simplified.
First, one would take $\omega=d\tau_0$ in \eqref{ealpha}, so that $g_\alpha$ can be written  as in \eqref{E4}. Thus, $\tau_0$ is directly a temporal function for $g_\alpha$. This implies that, in Def. \ref{dper} (iv), \bcambios{neither  the case $ t'_+\in [t_-,t(p))$ and $J_\alpha(p,\bSigma_{t'_+})\neq \emptyset$ (nor  $J_\alpha(p,\bSigma_{t(p)})\neq \{p\}$),
nor the case $t'_-\in (t(p),t_+]$ and $J_\alpha(\bSigma_{t'_-},p)\neq \emptyset$ (nor  $J_\alpha(\bSigma_{t(p)},p)\neq \{p\}$) can hold.} So, they should not be taken into account in the proof of
 Lemma \ref{L1}; what is more, in this lemma it is not necessary to prove (strong) causality, and all the slices $\tau_0=$ constant  become directly   spacelike and, then, acausal and Cauchy  for $g_\alpha$ (this was not true even for smooth Geroch's functions). These properties also hold in Prop. \ref{L2}; so, the claim in the proof of Thm. \ref{T0} is not needed.
}\end{rem}

\ncambios{We are now in conditions to sketch the proof of Proposition \ref{p_stable}.}

\medskip

\ncambios{
\noindent {\em Proof of Prop. \ref{p_stable}.}
	1 $\Rightarrow$ 2. The same smoothing procedure as in the case without boundary \cite{BernalSanchez, SanchezSaoPaulo} holds, as no additional condition is required for $\nabla\tau$ on the boundary (see Remark~\ref{r_final}).
	
	2 $\Rightarrow$ 3. Even if $\overline{M}$ does not split globally as a product, the temporal function $\tau$ still allows to write $g$ as
	$-\Lambda d\tau^2+g_\tau$ where $\Lambda>0$ is a function on $\M$ and $g_\tau$ a Riemannian metric on the bundle Ker$(d\tau)$. So, for any positive function $\alpha>0$ the metric $g_\alpha=-(\Lambda+\alpha)d\tau^2+g_\tau $ satisfies $g_\alpha>g$ and it is also stably causal. Indeed, $\tau$ is a temporal function also for $g_\alpha$ because the gradients of $\tau$ for $g$ and $g_\alpha$ are pointwise proportional and their $g$ and $g_\alpha$-orthogonal bundles are equal to Ker$(d\tau)$, that is, positive definite (see
	also Remark \ref{r_stably}).
	
	3 $\Rightarrow$ 1. Hawking's proof for the case without boundary \cite[Prop. 6.4.9]{Hawking} also works here, because it is based on the integration of chronological futures and pasts type $I^\pm(p)$ by using and admissible measure $m$ independent of the metric  such that both $\pM$ and the boundaries of $I^\pm(p)$ have zero measure (these are the same reasons why Geroch's construction of Cauchy time functions also work in the case with boundary, as detailed in Section \ref{s3}).
	
	The last assertion is trivial from the assertion 3; however, a direct proof from 1 is easy to obtain (see \cite[Remark 3.9(b)]{LuisTesis}). $\Box$}

\medskip

 As a simple consequence of our techniques, let us check that any globally hyperbolic spacetime-with-timelike-boundary can be regarded as the closure of an open subset in a globally hyperbolic spacetime without boundary (thus, strengthening Prop. \ref{pp}).

\begin{cor}\label{c_extension_globhip}
For any globally hyperbolic spacetime-with-timelike-boundary $(\M,g)$ there exists a globally hyperbolic spacetime without boundary $(\mathring{M},\mathring{g})$ with the same dimension 
and an isometric embedding $i:\overline{M} \hookrightarrow \mathring{M}$.
\end{cor}

\begin{proof}
Consider first the globally hyperbolic metric $g^{d}$ on all the double manifold satisfying $g<g^{d}$ on $\M$ (see the beginning of this subsection) as well as the extension $\tilde{g}$ of $g$ to some $\widetilde{M} \subset \M^{d}$  as in Prop. \ref{pp}. Taking a smaller $\widetilde{M}$ if necessary,
we can assume that
\begin{equation}
\tilde{g} < g^{d} \;\;  \hbox{on} \;\; \widetilde{M}\quad\hbox{and}\quad \M\subset\widetilde{N}\subset {\rm cl}(\widetilde{N})\subset\widetilde{M}\subset \M^{d}\quad\hbox{for some open subset $\widetilde{N}$ of $\widetilde{M}$.}
%
%
\end{equation}
Let $\tau^{d}$ be the Cauchy temporal function on $(\M^{d},g^{d})$ (see the beginning of this subsection), and consider the corresponding orthogonal splitting $\M=\R \times \bSigma \subset \widetilde{M} \subset \R \times \Sigma^{d}\equiv \M^{d}$. Then, on $\widetilde{M}$ we can write:
\begin{equation}
\tilde{g}=-\tilde{\Lambda}(d\tau^{d})^{2} + \tilde{g}_{\tau^{d}}.
\end{equation}
Let $\mu$ be a bump function satisfying $\mu(\M)\equiv 1$ and $\mu(\M^{d} \setminus \widetilde{N})\equiv 0$. Let us check that the required extension is $\mathring{M}:=\mu^{-1}((0,1])$ (note that $\M \subset \mathring{M} \subset \widetilde{N}\subset {\rm cl}(\widetilde{N})\subset \widetilde{M} \subset \M^{d}$; in particular, ${\rm cl}(\mathring{M})\subset \widetilde{M}$), and $\mathring{g}=-\mathring{\Lambda}(d\tau^{d})^{2} + \tilde{g}_{\tau^{d}}$ with $\mathring{\Lambda}:=\mu  \tilde{\Lambda}$.

Clearly, $(\overline{M},g)$ is isometrically embedded into the spacetime of the same dimension $(\mathring{M},\mathring{g})$. Moreover, $\tau^{d}$ is a temporal function for $(\mathring{M},\mathring{g})$ as $\mathring{g} (\leq \tilde{g}) < g^{d}$. So, it is enough to check that the slices $\{\tau^{d}=c \mid c\in \R \} \cap \mathring{M}$ are Cauchy, that is, we will check that $\tau^{d}\mid_{\mathring{M}}$ is Cauchy temporal for $(\mathring{M},\mathring{g})$.

Let $\gamma$ be an inextensible future-directed causal curve in $(\mathring{M},\mathring{g})$. Then, $\gamma$ is also causal in $(\M^{d},g^{d})$ and we can reparametrize it with $\tau^{d}$, that is, $\gamma:(a,b) \rightarrow \mathring{M}\subset \R \times \Sigma^{d}= \overline{M}^d$, $\gamma(\tau)=(\tau,x(\tau))$. If, say $b<\infty$ (the case $-\infty < a$ is analogous) then necessarily $\exists \, x_{b} \in \Sigma^{d}$ and, taking into account that
$\{b\} \times \Sigma^{d}$ is Cauchy for $(\overline{M}^d,g^d)$,
\begin{equation}
\lim_{\tau \rightarrow b} \gamma(\tau)=(b,x_{b}) \in cl(\mathring{M}) \setminus \mathring{M} \subset \R \times \Sigma^{d}\equiv \M^{d}\quad\hbox{($cl(\mathring{M})\equiv$closure of $\mathring{M}$ in $\M^{d}$).}
\end{equation}
As $\gamma$ is $\mathring{g}$-causal one has:
\[
\mu(\gamma(\tau)) \tilde{\Lambda}(\gamma(\tau)) \geq \tilde{g}_{\tau^{d}}(x'(\tau),x'(\tau)).
\]
Thus, writing $\tilde{g}_{\tau^{d}}$ in coordinates $\tilde{g}_{ij}$ around $(b,x_b)$:
\[
\tilde{\Lambda}(\gamma(\tau)) \geq \tilde{g}_{ij}(\gamma(\tau)) \frac{\dot{x}^{i}(\tau) \dot{x}^{j}(\tau)}{\mu(\gamma(\tau))}.
\]
If $\epsilon>0$ is the minimum eigenvalue of $\tilde{g}_{ij}(\gamma(\tau))$ in a neighbourhood of $(b,x(b))$ and $\Lambda_{0}(>0)$ is the maximum of $\tilde{\Lambda}$ in this neighbourhood:
\begin{equation}
\label{eqcor}
\Lambda_{0} > \epsilon \, \frac{||\dot{x}(\tau)||_{0}^{2}}{\mu(\gamma(\tau))},
\end{equation}
where $|| \cdot ||_{0}$ denotes the standard Euclidean metric in the coordinates $(x^{1},\ldots,x^{n})$.
However, as $\mu(b,x_b)=0$ one has a contradiction with (\ref{eqcor}) by expanding in a series:
\[
\mu(\gamma(\tau))= \mu((b,x_b))+
{\frac{\partial \mu}{\partial x_{i}}(\mu(b,x_b))} \dot{x}^{i}+
{\frac{\partial^{2} \mu}{\partial x_{i} x_{j}}(\mu(b,x_b))} \dot{x}^{i} \dot{x}^{j}+o(||\dot{x}||_0^{3})= o(||\dot{x}||_0^{3}).
\]
The last equality \bcambios{holds} because $(b,x_b)$ belongs either to the boundary or to the support of~$\mu$.
\end{proof}

\section{Appendix A: naked singularities and the causal boundary}

Clearly, all the points of the timelike boundary $\pM$ of $\M$ correspond to (conformally invariant) naked singularites of $M$. In this Appendix, globally hyperbolic spacetimes-with-timelike-boundary are characterized as those containing all its naked singularities.
Before proceeding, we need to recall some basic notions and properties associated to the causal boundary (c-boundary) of spacetimes without boundary.

\subsection{Brief review on the c-completion of spacetimes}

We refer to  \cite{causalb, Flores} for further details and proofs (see also the original article \cite{GerochIdealPointsSpaceTime1972}).

A past set $P\subset M$ (i.e., $P\neq\emptyset$, $I^-(P)=P$) that cannot be written as the
union of two proper subsets, both of which are also past sets, is
said to be an {\em indecomposable past} set (IP). It can be shown that an IP either coincides with the past of some point of the spacetime, i.e., $P=I^{-}(p)$ for $p\in M$, or else $P=I^{-}(\gamma)$ for some inextendible future-directed
timelike curve $\gamma$. In the former case, $P$ is said to be a {\em proper indecomposable past
set} (PIP), and in the latter case $P$ is said to be a {\em
terminal indecomposable past set} (TIP). These two classes of IPs are disjoint.

The {\em common past} of a given set $S\subset
M$ is defined by \[\downarrow S:=I^{-}(\{p\in M:\;\; p\ll
q\;\;\forall q\in S\}).\]
The corresponding
definitions for {\em future sets}, IFs, TIFs,
PIFs, {\em common future}, etc., are obtained just by interchanging the roles of past and
future, and will always be understood.

The set of all IPs constitutes the so-called {\it future $c$-completion} of $(M,g)$, denoted by $\hat{M}$. If $(M,g)$ is strongly causal, then $M$ can naturally be viewed as a subset of $\hat{M}$ by identifying every point $p\in M$ with its respective PIP, namely $I^-(p)$.
The {\it future $c$-boundary} $\hat{\partial} M$ of $(M,g)$ is defined as the set of all its TIPs. Therefore, upon identifying $M$ with its image in $\hat{M}$ by the natural inclusion as outlined above,
\[
\hat{\partial} M \equiv \hat{M} \setminus M.
\]
The definitions of \textit{past $c$-completion} $\check{M}$ and {\it past $c$-boundary} $\check{\partial}M$ of $(M,g)$ are readily defined in a time-dual fashion using IFs.

Next, we introduce the so-called {\em Szabados relation} (or \textit{$S$-relation}) between IPs and IFs: an IP $P$ and an IF $F$ are {\em S-related}, denoted $P\sim_{S}F$, if $P$ is a maximal IP inside
$\downarrow F$ and $F$ is a maximal IF inside $\uparrow P$. In particular, for any $p \in M$, it can be shown that $I^-(p) \sim_{S} I^+(p)$.
\begin{defi}\label{d1} The {\em (total) c-completion} $\overline{M}^c$ is
composed by all
the pairs $(P,F)$ formed by $P\in \hat{M}\cup\{\emptyset\}$ and
$F\in \check{M}\cup\{\emptyset\}$ such that either
\begin{itemize}
\item[i)] both $P$ and $F$ are non-empty and $P\sim_{S}F$; or
\item[ii)] $P=\emptyset$, $F \neq \emptyset$ and there is no $P'\neq \emptyset$ such that $P'\sim_{S}F$; or
\item[iii)] $F=\emptyset$, $P \neq \emptyset$ and there is no $F'\neq\emptyset$ such that $P\sim_{S}F'$.
\end{itemize}
The original manifold $M$ is then identified with the set $\{(I^{-}(p),I^{+}(p)): p\in M\}$, and the {\em c-boundary} is defined as $\partial^c M\equiv\overline{M}^c\setminus M$.
\end{defi}

\begin{rem}
\label{nakedsingu}
{\em Any pair $(P,F)\in \partial^c M$, with $P \neq \emptyset \neq F$, will be called a {\em naked singularity}. Notice that, necessarily $P=I^-(\gamma)$ for some inextendible future-directed timelike curve, and $\gamma$ must lie in the past of any $z\in F$ (a dual assertion follows by interchanging the role of $P$ and $F$). According to a classical physical interpretation, $\gamma$ may represent a particle dissapearing of the spacetime, and all this process can be seen at $z$. Conversely, whenever such a $z, \gamma$ ($\gamma\subset I^-(z)$) exist, a pair $(P,F)\in\partial^c M, P \neq \emptyset \neq F$ must appear.

Clearly, such a notion of naked singularity is  conformally invariant. However, one expects that the physically relevant representatives of the conformal  class will present  curvature-related divergences along such $\gamma$'s which make the spacetime inextensible. } \end{rem}

Having defined the set structure of the $c$-completion, the next step is to extend the chronological relation in $(M,g)$ to the $c$-completion as follows:
\begin{equation}
(P,F)\ll
(P',F')\;\;\iff\;\; F\cap P'\neq\emptyset. \label{eq:7}
\end{equation}

Next, let us define the {\em future chronological limit operator} $\hat{L}$ on $\hat{M}$ as follows.
%
Given a sequence $\sigma=\{P_{n}\}_{n}\subset \hat{M}$ of IPs and $P\in \hat{M}$, we set
\begin{equation}
\label{eq:3}
P\in \hat{L}(\sigma)\iff \left\{\begin{array}{l}
                                 P\subset \mathrm{LI}(\sigma)\\
                                 P \hbox{ is a maximal IP in }\mathrm{LS}(\sigma),
\end{array}\right.
\end{equation}
\bcambios{
\[
\hbox{where}\quad
\begin{array}{lll}
LI(\sigma) &:=& \{x \in M \, : \, \mbox{  $x$ belongs to all but finitely many $P_n$'s}\}, \\
LS(\sigma) &:=& \{x \in M \, : \, \mbox{ $x$ belongs to infinitely many $P_n$'s}\}.
\end{array}
\]
}
Again, by simply interchanging past and future sets we may analogously define the {\it past chronological limit operator} $\check{L}$ on $\check{M}$. Then, the {\em future (resp. past) chronological topology on $\hat{M}$ (resp. $\check{M}$)} is the derived topology associated to the limit operator $\hat{L}$ (resp. $\check{L}$), that is, the topology whose {\em closed sets} are those subsets $C \subset \hat{M}$ (resp. $C \subset \check{M}$) such that $\hat{L}(\sigma) \subset C$ (resp. $\check{L}(\sigma) \subset C$) for any sequence $\sigma$ of elements of $C$.

\smallskip

 In order to define the chronological topology on the full $c$-boundary, first define a limit operator $L$ on $\overline{M}^c$ as follows: given a sequence
$\sigma=\{(P_{n},F_{n})\}\subset\overline{M}^c$, put

\begin{equation}
\label{eq:4}
(P,F)\in L(\sigma)\iff \left\{
  \begin{array}{l}
    P\in \hat{L}(P_{n}) \hbox{ if $P\neq \emptyset$}\\
F\in\check{L}(F_n) \hbox{ if $F\neq \emptyset$}.
  \end{array}
\right.
\end{equation}
By definition, the {\em chronological topology on $\overline{M}^c$} is the derived topology $\tau_L$ associated to the limit operator $L$ defined in (\ref{eq:4}), that is, the topology whose {\em closed sets} are those subsets $C \subset \overline{M}^c$ such that $L(\sigma) \subset C$ for any sequence $\sigma$ of elements of $C$.

The following result (see \cite[Thm. 3.27]{causalb}) summarizes the key properties of the chronological topology.

\begin{thm}
\label{thm:mainc-completion}
Let $(M,g)$ be a strongly causal spacetime and consider its associated $c$-completion $\overline{M}^c$ endowed with the chronological relations and chronological topology defined in (\ref{eq:7}) and (\ref{eq:4}), respectively. Then, the following statements hold.

\begin{itemize}
\item[1.] \label{item:mismatopologia}The inclusion $M\hookrightarrow \overline{M}^c$ is continuous, with an open dense image. In particular, $\partial^c M$ is closed in $\M^c$ and the topology induced on $M$ by the chronological topology on $\overline{M}$ coincides with the original manifold topology.
\item[2.] The chronological topology is second-countable and $T_1$
(but not necessarily $T_2$).
\item[3.] \label{item:chains} Let $\{x_n\}\subset M$ be a {\em future (resp. past) chain}, i.e., a sequence satisfying that $x_n\ll x_{n+1}$ (resp. $x_{n+1}\ll x_{n}$) for all $n$. Then,
\[
\begin{array}{c}
L(\left\{ x_{n} \right\})=\left\{ (P,F)\in \overline{M}^c: P=I^-(\left\{ x_n \right\}) \right\}
\\
\hbox{(resp. $L(\left\{ x_{n} \right\})=\left\{ (P,F)\in \overline{M}^c: F=I^+(\left\{ x_n \right\}) \right\})$.}
\end{array}
\]

\item[4.] \label{item:c-completioncompleta} The c-completion is {\em complete} in the following sense: given any (future or past) chain $\{x_n\}\subset M$, necessarily $L(\left\{ x_{n} \right\}) \neq \emptyset$, i.e. any (future or past) chain converges in $\overline{M}^c$. 

\item[5.] The sets $I^{\pm}((P,F))\subset\overline{M}^c$ are open for all $(P,F)\in \overline{M}^c$.
\end{itemize}

 \end{thm}

\subsection{Main result}

It is worth emphasizing the following result in the case without boundary \cite[Cor. 4.34]{causalb}:
\begin{quote}
{\em Let $(M,g)$ be a spacetime which admits a conformal boundary $\pM$ such
that $\M=M\cup \pM$ is $C^1$ and $\pM$ {\em chronologically complete} (i.e., each inextensible future/past-directed timelike curve in $M$ has an endpoint in $\partial M$).

$M$ is
globally hyperbolic if and only if $\pM$ {\em does not admit timelike points}
(i.e., $T_{\hat p}(\pM)$ is everywhere either a spacelike or a lightlike
hyperplane).}
\end{quote}
This result
suggests that, for a globally hyperbolic \ncambios{spacetime-with-timelike-boundary}, the boundary should contain all the naked singularities.
The formal statement and proof of this assertion will be given in Th. \ref{naked} below. Previously, Lemma \ref{diamondclosure} will extend the basic results \cite[Lemma 3.4]{BE},  \cite[Thm. 2.3]{Harris} to the case of a strongly causal spacetime-with-timelike-boundary.

\begin{rem} 
{\em The hypotheses which ensure that the (conformal) boundary $\pM$ of a \ncambios{spacetime-with-boundary} $\M$ can be identified with the c-boundary $\partial^cM$ of its interior $M$ and, moreover, that the c-completion $\M^c$ agrees with the \ncambios{spacetime-with-boundary} $\M$, were analyzed in \cite[Section 4]{causalb}, see specially Thm. 4.26, Cor. 4.28 and Thm. 4.32 therein.
In particular, it is straightforward to check that, the boundary $\pM$ of any spacetime-with-timelike-boundary $\pM$  is  {\em regularly accessible}.  Essentially, this means that $\pM$
cannot present a variety of pathologies which occur for arbitrary conformal embeddings, and $\pM$ can be regarded as a part of the causal boundary. Indeed,  $\pM$ can be identified with the whole $\partial^cM$ (and $\M$ with $\M^c$) when $\pM$ is both regularly accessible and chronologically complete, see \cite[Th. 4.16]{causalb}. Anyway, here we will make a self-contained development adapted to our purposes.
}
\end{rem}

\ncambios{We will abuse of notation in the remainder by writing  $I^{\pm}(p,M)$ instead of $I^{\pm}(p)\cap M$ for any $p\in \overline{M}$ (recall Remark \ref{nuevo}).}

\begin{lemma}
\label{diamondclosure}
Let $(\overline{M},g)$ be a strongly causal spacetime-with-timelike-boundary.

(1) If $cl(J^{+}(p) \cap J^{-}(q))$ is
compact for all $p,q \in \overline{M}$ then $(\overline{M},g)$ is globally hyperbolic with timelike boundary.

(2) Let $\hat x\in \overline{M}$ and $\{x_n\}$ be a sequence in $\overline{M}$ such that $I^-(x,M)\in \hat L (\{I^-(x_n,M)\})$. Then $\{x_n\}$ converges to $\hat x$ with the manifold topology of $\M$.
\end{lemma}

 \begin{proof} (1) Essentially, we will  follow the proof of \cite[Lemma 4.29]{Beem}.
  It suffices to show that $J^{+}(p) \cap J^{-}(q)\subset\overline{M}$ is closed for every $p,q \in \overline{M}$. By contradiction, suppose
 that $r \in cl(J^{+}(p) \cap J^{-}(q)) \setminus J^{+}(p) \cap J^{-}(q)$. Then, there exists
 $\{r_{n}\} \subset J^{+}(p) \cap J^{-}(q)$ such that $r_{n} \rightarrow r$. Let $\gamma_{n}:[0,1) \rightarrow \overline{M}$ be
 inextensible future-directed causal curves such that $\gamma_{n}(0)=p$, $r_{n} \in \gamma_n$ and $q \in \gamma_{n}$ for all $n$. By Prop. \ref{p2.11}
 there exists a future-directed causal limit curve $\gamma:[0,1) \rightarrow \overline{M}$ of $\{\gamma_{n}\}$
 such that $\gamma(0)=p$. Since $(\overline{M},g)$ is strongly causal, the inextensible causal curve $\gamma:[0,1) \rightarrow \overline{M}$ is not imprisoned on the compact subset
 $cl(J^{+}(p) \cap J^{-}(q))$. Hence, there exists $x \in {\rm Im}(\gamma)$ such that
 $x \not \in cl(J^{+}(p) \cap J^{-}(q))$. From the notion of limit curve, any neighborhood of $x\in Im(\gamma)$
 intersects all but finitely many of the $\gamma_n$'s. So, we can assume without restriction the existence of a sequence $x_n \in Im(\gamma_n)$
 such that $\{x_{n}\}$ converges to $x$. Since $x \not \in cl(J^{+}(p) \cap J^{-}(q))$
 we also have that $x_{n} \not \in cl(J^{+}(p) \cap J^{-}(q))$ for all $n$ large enough. So taking into account that $\gamma_{n} \subset J^{+}(p)$, it follows
 that $x_{n} \not \in J^{-}(q)$ for large $n$. Hence $q$ lies between the points $p$ and $x_{n}$ on $\gamma_{n}$, i.e, $p \leq r_{n} \leq q \leq x_{n}$
 for large $n$. Denote by $\gamma \mid_{[p,x]}$ the portion of $\gamma$ between the point $p$ and $x$, and by $\gamma_{n} \mid_{[p,x_{n}]}$
 the portion of $\gamma_{n}$ between the points $p$ and $x_{n}$. From Prop. \ref{otro} (1),
 we may assume, by taking a subsequence of
 $\{\gamma_{n}\mid_{[p,x_{n}]}\}$  if necessary, that $\{\gamma_{n}\mid_{[p,x_{n}]}\}$ converges to $\{\gamma\mid_{[p,x]}\}$ in the $C^0$ topology of curves.
 Since $q \in \gamma_{n} \mid_{[p,x_{n}]}$, necessarily $q \in \gamma \mid_{[p,x]}$. On the other hand, since $r_{n} \leq q$
 and $r_{n} \rightarrow r$, necessarily $r \in \gamma\mid_{[p,x]}$ and $r\leq q$. Therefore, $r \in J^{+}(p) \cap J^{-}(q)$, a contradiction.

 (2) We will follow the reasoning in the implication to the left of \cite[Thm. 2.3]{Harris}. Assume by contradiction that $x_{n} \not \rightarrow \hat x$. Then, there exists
a relative compact open neighbourhood $U \ni \hat x$ ($U\subset \M$), and a subsequence $\{x_{n_{k}}\}$ with $x_{n_{k}} \not \in U$ for all $k$; by strong causality, $U$ can be assumed causally convex (Prop.  \ref{imprisoned}). Consider a future chain $\{z_{n}\} \subset I^{-}(\hat{x},M)$ such that $z_{n} \rightarrow \hat{x}$.
For $n$ sufficiently large, $z_{n} \in U$. Since $z_{n} \ll \hat{x}$, there exists $K_{n}\in \mathbb{N}$ such that $z_{n} \ll x_{n_{k}}$ for $k \geq K_{n}$. So, there is a timelike curve $\gamma_{n}^{k}$ in $\overline{M}$ from $x_{n}$ to $z_{n_{k}}$. Since $x_{n_{k}} \not \in U$, $\gamma^{k}_{n}$ exits $U$ at some point $y_{n_{k}} \in \dot{U}$. For each $n$, the curves $\{\gamma_{n}^{k} \mid k \geq K_{n}\}$ have a future-directed causal limit curve
$\gamma_{n}$ in $\overline{M}$ from $z_{n}$ to some point $y_{n} \in \dot{U}$. Moreover, the sequence of curves $\{\gamma_{n}\}$ has a future-directed causal limit curve
$\gamma$ in $\overline{M}$ from $\hat{z}$ to some $y \in \dot{U}$ as $\gamma$ cannot remain imprisoned in $cl(U)$.

Let $I^{-}(y,M)=I^{-}(\gamma,M)$ and note that $I^{-}(y,M) \subset LI(\{I^{-}(x_{n_{k}},M)\})$. In fact, take $w \in I^{-}(y,M)$, then $y \in I^{+}(w,M)$, so, for large $n$, we have that
$y_{n} \in I^{+}(w,M)$. Therefore, for $k$ large enough, $y_{n_k} \in I^{+}(w,M)$. Since $y_{n_k} \ll x_{n_{k}}$, necessarily $x_{n_{k}} \in I^{+}(w,M)$ for large $k$. Hence, $(I^-(\hat{x},M)\subsetneq)I^{-}(y,M)\subset LI(I^-(x_{n_k},M))$, in contradiction with $I^-(\hat{x},M)\in \hat{L}(I^-(x_n,M))\subset\hat{L}(I^-(x_{n_k},M))$.
\end{proof}


\begin{thm}
\label{naked}
Let $(\overline{M},g)$ be a strongly causal spacetime-with-timelike-boundary. The following sentences are equivalent:
\begin{itemize}
 \item [(a)] $(\overline{M},g)$ is a globally hyperbolic spacetime-with-timelike-boundary.

 \item [(b)] $\partial M$ contains all the naked singularities, i.e., for any pair
 $(P,F)$ in the c-boundary $\partial^{c}M$, with $P\neq \emptyset \neq F$, there exists $\hat{z} \in \partial M$ such that $(P,F)=(I^{-}(\hat{z},M),I^{+}(\hat{z},M))$.
\end{itemize}
\end{thm}

\begin{proof}
(a) $\Rightarrow$ (b). Let $(P,F) \in \partial^{c}M$ with $P \neq \emptyset \neq F$. Take chains $\{p_{n}\}$, $\{q_{n}\} \subset M$ generating $P$ and $F$, respectively. Then, for some $p_0, q_0\in M$ (which could be included as the first element of each chain) we have
$$p_{n}, q_{n} \in I^{+}(p_0,M) \cap I^{-}(q_0,M)
\subset I^{+}(p_0) \cap I^{-}(q_0) \subset J^{+}(p_0) \cap J^{-}(q_0).$$
The compactness of $J^{+}(p_0) \cap J^{-}(q_0)$
implies the existence of $\hat{z}, \hat{z}'\in \partial{M}$ such that $p_{n} \rightarrow \hat{z}$ and $q_{n} \rightarrow \hat{z}'$ respectively (since $(\overline{M},g)$ is strongly causal we can ensure that previous limits hold for the entire sequences, not only for some subsequences of them). Assume by contradiction that $\hat{z} \neq \hat{z}'$. Since $p_{n} \ll q_{n}$, there exists future-directed timelike curves $\gamma_{n}:[0,1] \rightarrow M$ joining each $p_{n}$ and
$q_{n}$ with $\gamma_{n}(0)=p_{n} \rightarrow \hat{z}$ and $\gamma_{n}(1)=q_{n} \rightarrow \hat{z}'$.
Since $(\overline{M},g)$ is globally hyperbolic, Prop. \ref{otro} (2)
implies the existence of a causal limit curve $\gamma:[0,1] \rightarrow \overline{M}$ so that $\gamma(0)=\hat{z}$ and $\gamma(1)=\hat{z}'$. Take some $0<s_0<1$ such that $\hat{z} < \gamma(s_{0}) < \hat{z}'$.
Note that $I^{-}(\gamma(s_{0}),M) \subset \downarrow F$. In fact, if $p \in I^{-}(\gamma(s_{0}),M)$ then $\gamma(s_{0}) \in I^{+}(p)$. So, taking into account that
$\gamma$ is a limit curve and $I^{+}(p)$ is an open set, there exists a subsequence $\{\gamma_{n_{k}}\}$ and
a subsequence $\{s_{n_{k}}\} \subset [0,1]$ such that $\gamma_{n_{k}}(s_{n_{k}}) \in I^{+}(p)$. This implies
$p \ll \gamma_{n_{k}}(s_{n_{k}}) \ll q_{n_{k}}$, and so, $p\ll q_{n_{k}}$ for all $k$, which implies $I^{-}(\gamma(s_{0}),M) \subset \downarrow F$, as required. In conclusion, to prove
$P \subsetneq I^{-}(\gamma(s_{0}),M)(\subset \downarrow F)$ suffices, as this would contradict $P\sim_s F$.

First, let us justify the identity $P=I^{-}(\hat{z},M)$ by proving the following

\smallskip

\noindent {\em Claim 1:} If $\beta:[a,b) \rightarrow \overline{M}$ is a future-directed timelike curve such that $\beta(t) \rightarrow p$ ($\beta$ is continuously extendible to $p$) then
$I^{-}(p)=I^{-}(\beta)$.
\smallskip

\noindent {\em Proof of the Claim 1.} For the inclusion to the right, take any $w \in I^{-}(p)$. Then, $p \in I^{+}(w)$ and, taking into account that it is an open, necessarily $\beta(t) \in I^{+}(w)$ for large
$t$. So, $w \in I^{-}(\beta)$. For the inclusion to the left, take now $w \in I^{-}(\beta)$. Then,
$w \ll \beta(t_{0})$ for some $t_{0} \in [a,b)$ and, taking into account that $\beta$ is continuously extendible to $p$, necessarily
$w \ll \beta(t_{0}) \ll p$. So, $w \in I^{-}(p)$. $\Box$

\smallskip

Taking into account the previous claim, to show $I^{-}(\hat{z},M)(=P)\subsetneq I^{-}(\gamma(s_0),M)$ suffices. We already know that $I^{-}(\hat{z}) \subsetneq I^{-}(\gamma(s_0))$ (from Prop. \ref{p_ord_lowerlevels}, $(\overline{M},g)$ is distinguishing).
Assume by contradiction that
$I^{-}(\hat{z},M)=I^{-}(\gamma(s_{0}),M)$. Since
$\partial M$ is smooth, we can take future-directed timelike curves $\sigma_{i}:[0,1) \rightarrow M \subset \overline{M}$, $i=1,2$, such that
$\sigma_{1}(1)=\hat{z} \in \partial M$ and
$\sigma_{2}(1)=\gamma(s_{0}) \in \partial M$. Moreover, these curves satisfy $I^{-}(\hat{z},M)=I^{-}(\sigma_1,M)$ and $I^{-}(\gamma(s_0),M)=I^{-}(\sigma_2,M)$. Since each curve $\sigma_{i}$, $i=1,2$, is future-directed timelike, necessarily $\sigma_{i} \subset I^{-}(\sigma_{i},M)$ for $i=1,2$. Moreover, from the initial assumption, $\sigma_{1} \subset I^{-}(\sigma_{2},M)$ and $\sigma_{2} \subset I^{-}(\sigma_{1},M)$. Hence
$I^{-}(\sigma_{1}) \subset I^{-}(\sigma_{2})$ and $I^{-}(\sigma_{2}) \subset I^{-}(\sigma_{1})$, and thus,
$I^{-}(\sigma_{1})=I^{-}(\sigma_{2})$. Indeed, let us prove that, say, $I^{-}(\sigma_{1}) \subset I^{-}(\sigma_{2})$. If
$w_{0} \in I^{-}(\sigma_{1})$ then $w_{0} \ll \sigma_{1}(t)$ for some $t \in [0,1)$. Since $\sigma_{1} \subset I^{-}(\sigma_{2},M)$, necessarily
$w_{0} \ll \sigma_{1}(t) \ll \sigma_{2}(s)$, which implies that $w_{0} \ll \sigma_{2}(s)$, and thus, $w_{0} \in I^{-}(\sigma_{2})$. The proof of $I^{-}(\sigma_{2}) \subset I^{-}(\sigma_{1})$ is analogous. In conclusion, we have proved that $I^{-}(\sigma_{2})=I^{-}(\sigma_{2})$ whenever
$\sigma_{1}(t) \rightarrow z$ and $\sigma_{2}(t) \rightarrow \gamma(s_{0})$, which implies $I^{-}(z)=I^{-}(\gamma(s_{0}))$, a contradiction.

\medskip

(a) $\Leftarrow$ (b). From Lemma \ref{diamondclosure} (1), it suffices to show that
$cl(J^{+}(p) \cap J^{-}(q))\subset \overline{M}$ is (sequentially) compact for any $p,q \in \overline{M}$. To this aim, consider any sequence $\{z_{n}\} \subset cl(J^{+}(p) \cap J^{-}(q))$. It is not a restriction to assume that $\{z_{n}\} \subset J^{+}(p) \cap J^{-}(q)$ (otherwise, replace $\{z_{n}\}$ by some other sequence $\{z'_n\}$ in $J^{+}(p) \cap J^{-}(q)$ such that $d_{R}(z'_{n},z_{n})<1/n$ for some auxiliary Riemannian metric $g_{R}$ on $\overline{M}$).
In order to prove that $\{z_{n}\}$ converges to some $z \in cl(J^{+}(p) \cap J^{-}(q))$, first note that every \cambios{$z_{n} \in J^{+}(p) \cap J^{-}(q)$}
satisfies $I^{-}(p,M) \subset I^{-}(z_{n},M)$ and $I^{+}(q,M) \subset I^{+}(z_{n},M)$. Hence, \cite[Prop. 5.3]{causalb2} ensures the existence of
 some IP $P$ and some IF $F$ containing $I^{-}(p,M)$ and $I^{+}(q,M)$, respectively, such that, up to a subsequence,
\begin{equation}\label{eq}
P\in \hat{L}(\{I^{-}(z_{n},M)\})\quad\hbox{and}\quad F\in \check{L}(\{I^{+}(z_{n},M)\}).
\end{equation}
In particular, $\emptyset\neq P \subset \downarrow F$ and $\emptyset\neq F \subset \uparrow P$.

In the case that $P=I^{-}(z,M)$ for some $z \in M$ (and analogously for $F$), the closedness of $\pM$ (see Thm. \ref{thm:mainc-completion}) and (\ref{eq:3})
imply that $z_{n} \in M$ for large $n$ and
$z_{n} \rightarrow z$ with the manifold topology.  So, only the case when both $P$ and $F$ are terminal (in $M$) must be taken into account. In this case, it suffices to show that $P\sim_S F$, since then, by hypothesis, $P=I^-(z,M)$ for some $z\in\pM$, and Lemma \ref{diamondclosure} (2) gives the result. So, assume by contradiction that $P\not\sim_S F$. Choose any $F'$ which is a maximal IF in $\uparrow P$ containing $F$. If $P\sim_S F'$ holds, again by hypothesis we have $F=I^+(z,M)$ for some $z\in\pM$, and Lemma \ref{diamondclosure} (2) implies that $z_n\rightarrow z$ (anyway, this case could not hold because, then, one would also have $F'\in \check{L}(\{I^+(z_n,M)\})$, in contradiction with the second expression in (\ref{eq})). So, the problem is reduced to the following claim.


\smallskip
\noindent {\em Claim 2:} If $F'$ is a maximal IF in $\uparrow P$ then $P$ must be a
 maximal IP in $\downarrow F'$. In particular, $P\sim_S F'$.
\smallskip

\smallskip

\noindent {\em Proof of the Claim 2.}
Assume by contradiction that $P \subsetneq P' \subset \downarrow F'$. Then, $(P',F') \in \overline{M}^{c}$.
 Let $z' \in \overline{M}$
 be the point such that $(P',F')=(I^-(z'),I^+(z'))$ (which exists by hypothesis) and consider future-directed chains $\{p_{i}\}$ and $\{p_{i}'\}$ generating
 $P$ and $P'$, respectively, such that $p_{i} \ll p_{i}'$ for all $i$. So, there exists some inextensible past-directed timelike curves $\gamma_i$ joining $p_{i}'$ with $p_{i}$ for all $i$. Since $\{p_{i}'\}$ converges to $z'$, from Prop. \ref{p2.11}
 there exists a past-directed
 causal limit curve $\gamma:[0,1) \rightarrow \overline{M}$ with $\gamma(0)=z'$.
 Next, we consider two excluding cases:
 \begin{itemize}
  \item  There exists $q_0\in {\rm Im}(\gamma)$ and some subsequence $\{p_{i_k}\}$ such that $p_{i_k}\rightarrow q_0$. Note that we can directly exclude the case $q_0=z'$, since, otherwise, $p_{i_k}\rightarrow z'$, and thus, $P=I^{-}(z',M)=P'$ (see Claim 1), a contradiction. In this case, take $s_{0}\in (0,1)$ and $\{s_{i_k}\}$ such that $\gamma(s_0)=q_0\neq z'$ and $p_{i_k}=\gamma_{i_k}(s_{i_k})$ for all $k$.

 \item  There is no subsequence of $\{p_i\}$ converging to some point of $\gamma$; in this case, take any $s_0\in (0,1)$, and define $q_0:=\gamma(s_0)$. Since $\gamma$ is a limit curve, there exists $\{s_{i_{k}}\}$ such that $\gamma_{i_k}(s_{i_k})\rightarrow q_0$. Moreover, we can assume additionally $p_{i_k} \ll \gamma_{i_k}(s_{i_k})$ for all $k$.
 Indeed, otherwise we can assume, up to a subsequence, $\gamma(s_{i_k}) \ll p_{i_k}$. By Prop. \ref{otro} (1) $\{\gamma_{i_{k}}\mid_{[0,s_{i_{k}}]}\}$ converges in the $C^0$-topology to $\gamma \mid_{[z^{'},q_{0}]}$. Thus, any relatively compact neighborhood of $Im(\gamma_{[z',q_{0}]})$ contains all the points $p_{i_k}$ up to a subsequence, in contradiction with the hypothesis of non-convergence.


 \end{itemize}

\noindent Trivially, $I^{+}(z',M) \subset I^{+}(q_{0},M)$ and we will prove $I^{+}(q_{0},M) \subset \uparrow P$ with a reasoning valid in both cases. If $w \in I^{+}(q_{0},M)$ then $q_{0} \in I^{-}(w)$, and taking into account that $\gamma_{i_k}(s_{i_k})\rightarrow q_0$, we have that $\gamma_{i_{k}}(s_{i_{k}}) \in I^{-}(w)$. So, $p_{i_{k}} \ll w$. Therefore, $p \ll p_{i_{k}} \ll w$ for any $p\in P$ and any $k$ large enough, and thus, $I^{+}(q_{0},M) \subset \uparrow P$. Summarizing, we have proved $I^{+}(z',M) \subset I^{+}(q_{0},M)\subset \uparrow P$. On the other hand, since $z' \neq q_{0}$, \cite[Lemma 4.6]{causalb} ensures that $F'(=I^{+}(z',M)) \subsetneq I^{+}(q_{0},M)$, in contradiction with the maximality of $F'$ into $\uparrow P$. $\Box$

 \end{proof}

\section{Appendix B: Piecewise smooth and $H^1$ causal relations are different}

Next, our aim will be to show that there exists a
spacetime-with-timelike-boundary  which can be written as a product $(\overline{M}= \R\times \bSigma, g=-d\tau^2 +g_0)$
such that the causal future of some point  computed by using piecewise $C^{\infty}$ smooth causal curves is {\em strictly contained} in its causal future computed with $H^1$-causal curves. Moreover,  $\M$  will be globally hyperbolic or not, depending on whether the causal relation is computed by using $H^1$-curves or piecewise smooth ones.

The problem is reduced to construct a complete $C^\infty$-Riemannian manifold with boundary $(\bSigma,g_0)$ containing two points $P_0$, $P_\infty$, which cannot be connected by means of any minimizing path with  piecewise smooth regularity.
This problem was studied by R. Alexander and S. Alexander \cite{AA}, who proved that the completeness of $g_0$ implies the existence of a minimizing path $\alpha$, but only $C^1$ regularity is assured for $\alpha$. So, assuming that such a $\alpha:[0,T]\rightarrow \bSigma$ cannot be piecewise smooth and it is  parametrized by arc-length, the curve $\gamma(t)=(t,\alpha(t)), t\in [0,T]$, becomes $H^1$-causal (indeed, $C^1$ and lightlike) and it connects $(0,P_0)$ with $(T,P_\infty)$. However, these points cannot be connected by means of a piecewise smooth  lightlike curve $\tilde \gamma$, as the projection of $\tilde \gamma$ on $\bSigma$ would be a piecewise smooth minimizer.

It is worth pointing out that such $P_0$ and $P_\infty$ can be connected in $\bSigma$ by means of arbitrarily close piecewise smooth curves $\gamma_k$    with lengths $T_k \searrow T$. However, this does not permit to connect $(0,P_0)$ and $(T,P_0)$  by means of piecewise smooth lightlike curves, but only  $(0,P_0)$ and each $(T_k,P_0)$ (these points can also be  joined by smooth timelike ones). 
Being $\M$ a product, its  global hyperbolicity  in the $H^1$-causal case  can be proven by checking that the slices $\tau=$ constant are Cauchy (exactly as in the case without boundary \cite[Theorem 3.66]{Beem}), while its   non-global hyperbolicity in the piecewise smooth case (indeed, its non-causal simplicity) follows just  because $J^+_{ps}((0,P_0))$ ($\not\ni (T,P_\infty)$) is not closed.

In order to find $(\bSigma,g_0)$ as claimed above,  the proof in \cite{AA} makes apparent that there is no any reason to expect piecewise smooth differentiability for the minimizers. Anyway,  the explicit counterexample is depicted in Fig. \ref{boundary} is sketched next.


\begin{figure}[hbt!]
	\centering
	\ifpdf
	\setlength{\unitlength}{1bp}%
	\begin{picture}(150,150)(0,0)
	\put(0,0){\includegraphics[scale=.25]{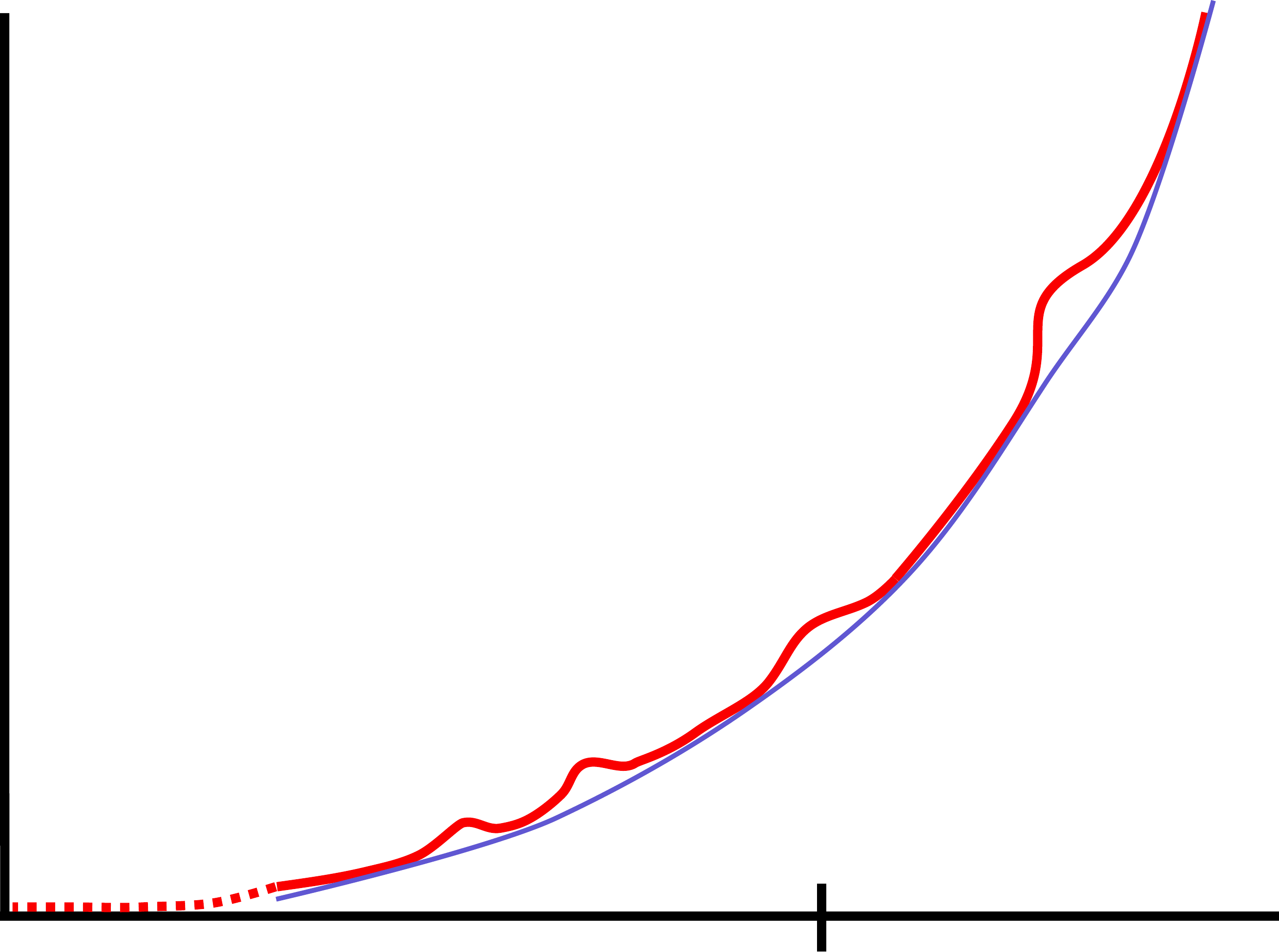}}
	\put(105.00,55.00){\fontsize{10.25}{0.0}\selectfont $\color{red} f$}
	\put(170.00,95.00){\fontsize{10.25}{0.0}\selectfont $\color{blue} h$}
	\put(115.00,-07.00){\fontsize{10.25}{0.0}\selectfont $x_m$}

	\end{picture}%
		\else
	\setlength{\unitlength}{1bp}%
	\begin{picture}(150,150)(0,0)
	\put(0,0){\includegraphics[scale=.25]{Boundary.pdf}}
	\put(105.00,55.00){\fontsize{10.25}{0.0}\selectfont $\color{red} f$}
	\put(170.00,95.00){\fontsize{10.25}{0.0}\selectfont $\color{blue} h$}
	\put(115.00,-07.00){\fontsize{10.25}{0.0}\selectfont $x_m$}
	\end{picture}%
	\fi
	\caption{\label{boundary} \gcambios{Function $f$ that determines the boundary of $(M_0,g_0)$. Note that $f$ is obtained by adding to the graph of $h$ infinitely many humps of decreasing size towards 0.}}
	
\end{figure}

Concretely, $(\bSigma,g_0)$ will be a closed subset of $\R^2$ (in natural coordinates $(x,y)$) type
\[
\bSigma:=\{(x,f(x))\in\R^2:\; x\leq f(x)\},\quad
\hbox{for some $C^\infty$-smooth function $f:\R\rightarrow \R$. }
\]
First, consider the smooth function $h:{\mathbb R}\rightarrow {\mathbb R}$,
\[
h(x):=\left\{\begin{array}{cc} 0 & \hbox{if $x\leq 0$} \\ e^{-1/x^2} & \hbox{if $x>0$.}\end{array}\right.
\]
Take the sequence of points $(P_m)$ with $P_m:=(x_m,h(x_m))\in {\rm graph}(h)$ \gcambios{for all $m> 5$} with $x_m=1/m$. For each $m$, consider the function $h_m:\mathbb{R}\rightarrow \mathbb{R}$ such that $h_m(x_m)=0$ and whose first derivative coincides with the following bump-cavity function (see Fig. \ref{hmprima})
\[
h'_m(x)=\left\{\begin{array}{ll} K_m e^{\lambda_{m}^+(x)} & \hbox{if $x\in (x_m-\epsilon_m,x_m)$} \\ -K_m e^{\lambda_{m}^-(x)} & \hbox{if $x\in (x_m,x_m+\epsilon_m)$} \\ 0 & \hbox{otherwise,}\end{array}  \right.
\]
where
\[
\lambda_m^{\pm}(x)=-\frac{1}{1-\left(\frac{2}{\epsilon_m}x\pm 1-\frac{2x_m}{\epsilon_m}\right)^2},\quad K_m=m^3e^{-m^2}\quad \hbox{and}\quad \epsilon_m=1/m^2.
\]

\begin{figure}[hbt!]
	\centering
	\ifpdf
	\setlength{\unitlength}{1bp}%
	\begin{picture}(150,150)(0,0)
	\put(0,0){\includegraphics[scale=.25]{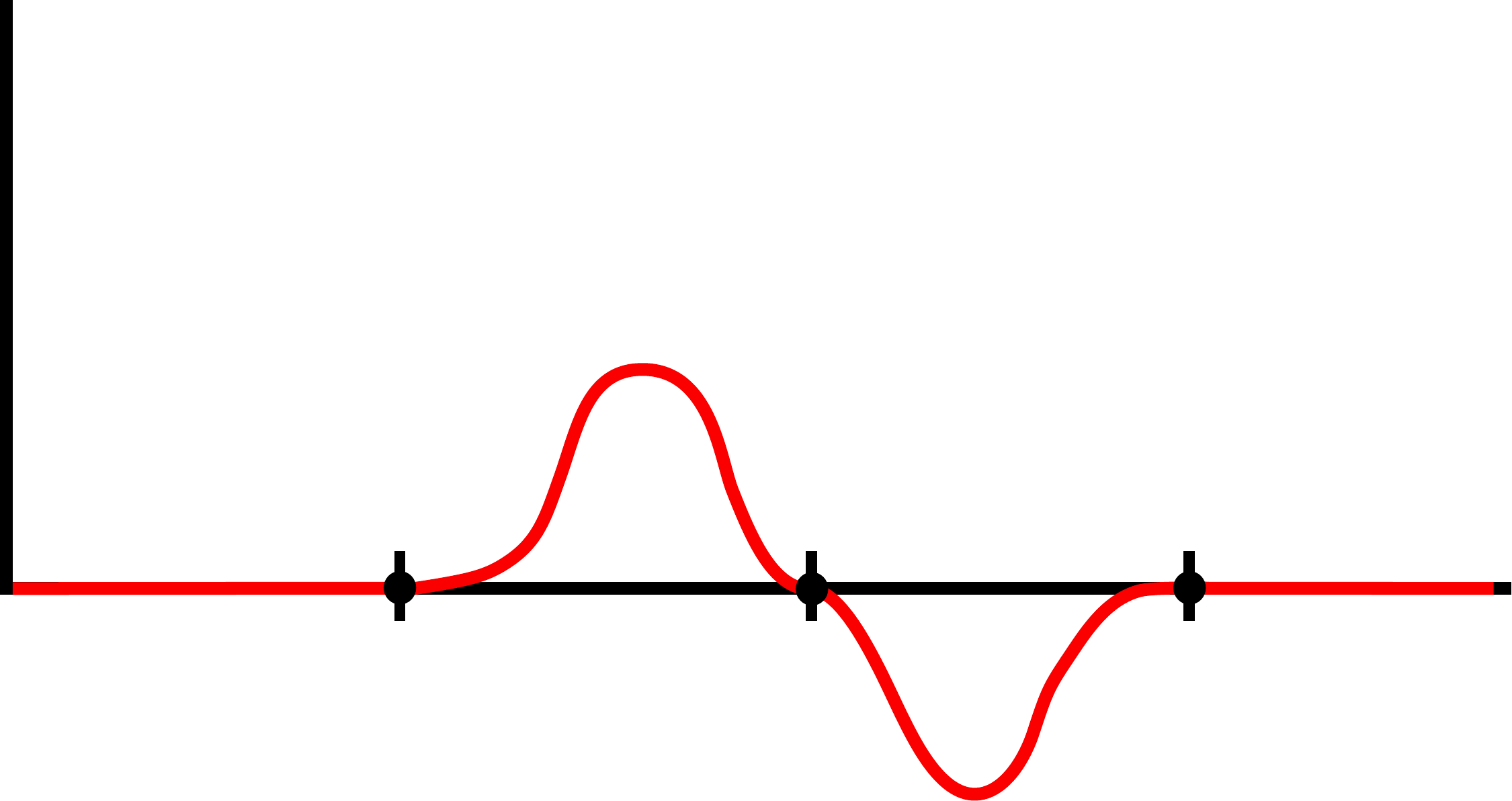}}
		\put(60.00,60.00){\fontsize{10.25}{0.0}\selectfont $h_{m}'$}
		\put(79.99,4.00){\fontsize{10.25}{0.0}\selectfont $x_m$}
		\put(115.00,4.00){\fontsize{10.25}{0.0}\selectfont $x_m+\epsilon_m$}
		\put(25.00,4.00){\fontsize{10.25}{0.0}\selectfont $x_m-\epsilon_m$}
	\end{picture}%
		\else
	\setlength{\unitlength}{1bp}%
	\begin{picture}(150,150)(0,0)
	\put(0,0){\includegraphics[scale=.25]{Funcionhmprima.pdf}}
	\put(60.00,60.00){\fontsize{10.25}{0.0}\selectfont $h_{m}'(x)$}
	\put(79.99,-07.00){\fontsize{10.25}{0.0}\selectfont $x_m$}
	\put(115.00,-07.00){\fontsize{10.25}{0.0}\selectfont $x_m+\epsilon_m$}
	\put(25.00,-07.00){\fontsize{10.25}{0.0}\selectfont $x_m-\epsilon_m$}
	\end{picture}%
	\fi
	\caption{\label{hmprima}\gcambios{Function $h_m'$. The derivatives of any order to $h'_m$ vanish at $x_{m}$ and $x_{m}\pm\epsilon_m$. }}
	
\end{figure}

Then, the required function $f$ is just (see Fig. \ref{boundary}):
\[
f:\mathbb{R}\rightarrow\mathbb{R},\qquad f(x):=h(x)+\sum_{m=1}^{\infty}h_m(x).
\]
Notice that $f$ is $C^\infty$ because so are $h$ and all $h_m$, and: (a) at each $x\in \R$ there is at most one non-vanishing $h_m$ because of our choice of $\epsilon_m$, and (b) for all $k\in \N$, the $k$-th derivative of $f$ satisfies  $\lim_{x\searrow 0}f^{(k)}(x)=0$ \gcambios{because $K_m/\epsilon_m^{2k}$ bounds $h_m^{(k)}$} and our choice of $K_m$ ensures $\lim_{m\rightarrow\infty} K_m/\epsilon_m^{2k}=0$.

For any point $P_0=(x_0, h(x_0))$  with $x_0>0$ and $P_\infty=(0,0)$ the minimizer $\alpha$ will not be piecewise smooth at $P_0$, because it will have $C^2$-breaks which accumulate at $P_\infty$. The idea is the following. Outside the regions $I_m\times \R$, with  $I_m:=(x_m-\epsilon_m,x_m+\epsilon_m)$, $\alpha$ must remain in the boundary $\partial \Sigma$ (i.e., $\alpha$ parametrize graph$(f)=$ graph$(h)$), because of the \gcambios{convexity of $\partial\Sigma$ respect to $\Sigma$}. However, $\alpha$ must abandon $\partial \Sigma$ at  each interval $I_m$ (whenever  $x_m+\epsilon_m< x_0$). Indeed, $\alpha$ must remain in $\Sigma$ and, thus, be a  segment of $\R^2$ in some maximal region $(x_m^-,x_m^+)\times \R$ with $x_m-\epsilon <x_m^-<x_m^+<x_m+\epsilon$. Even more, at  $x_m^-$, $\alpha$ must be tangent to $\partial\Sigma$ and it must remain in the boundary for close $x<x_m^-$. Nevertheless, as \gcambios{one will have} $f''(x_m^-)\neq 0$  , $\alpha$ is not $C^2$-smooth at $x_m^-$.

In order to check these properties, notice:

\begin{itemize}
\item[(A)] The line \gcambios{$l_1$} tangent to the graph of $f$ at $x_m-\epsilon_m$ remains strictly below of the graph of $f$ \gcambios{ for all $x\geq x_m-\epsilon_m$} (see Fig. \ref{lineastf}); \gcambios{in particular $(x_m+\epsilon_m,f(x_m +\epsilon_m))$ lies above $l_1$. }

Indeed, this follows because $h$ is convex, $f(x_{m}-\epsilon_m)=h(x_{m}-\epsilon_m)$ and $f\geq h$.

\item[(B)] The line \gcambios{$l_2$} tangent to the graph of $f$ at $x_{m}-3\epsilon_m/4$ passes strictly over the graph of $h$, and thus, of $f$, at $x_m+\epsilon_m$ (see Fig. \ref{lineastf});
\gcambios{in particular $(x_m+\epsilon_m,f(x_m +\epsilon_m))$ lies below $l_2$.}
This follows from:

$
f(x_m-3\epsilon_m/4)=h(x_m-3\epsilon_m/4)+h_m(x_m-3\epsilon_m/4)>h(x_m-3\epsilon_m/4)$,

$f'(x_m-3\epsilon_m/4)>h_m'(x_m-3\epsilon_m/4) \gcambios{> K_m/e^2=m^3e^{-m^2-2}>\frac{h(x_m+\epsilon_m)-h(x_m-3\epsilon_m/4)}{7\epsilon_m/4}.}
$
\gcambios{(for the last inequality,
$4m^2h(x_m+\epsilon_m)/7< 6m^2e^{-m^2-2}\leq K_m/e^2$ when  $m>5$).}
\end{itemize}
\begin{figure}[hbt!]
	\centering
	\ifpdf
	\setlength{\unitlength}{1bp}%
	\begin{picture}(150,200)(0,0)
	\put(0,0){\includegraphics[scale=.25]{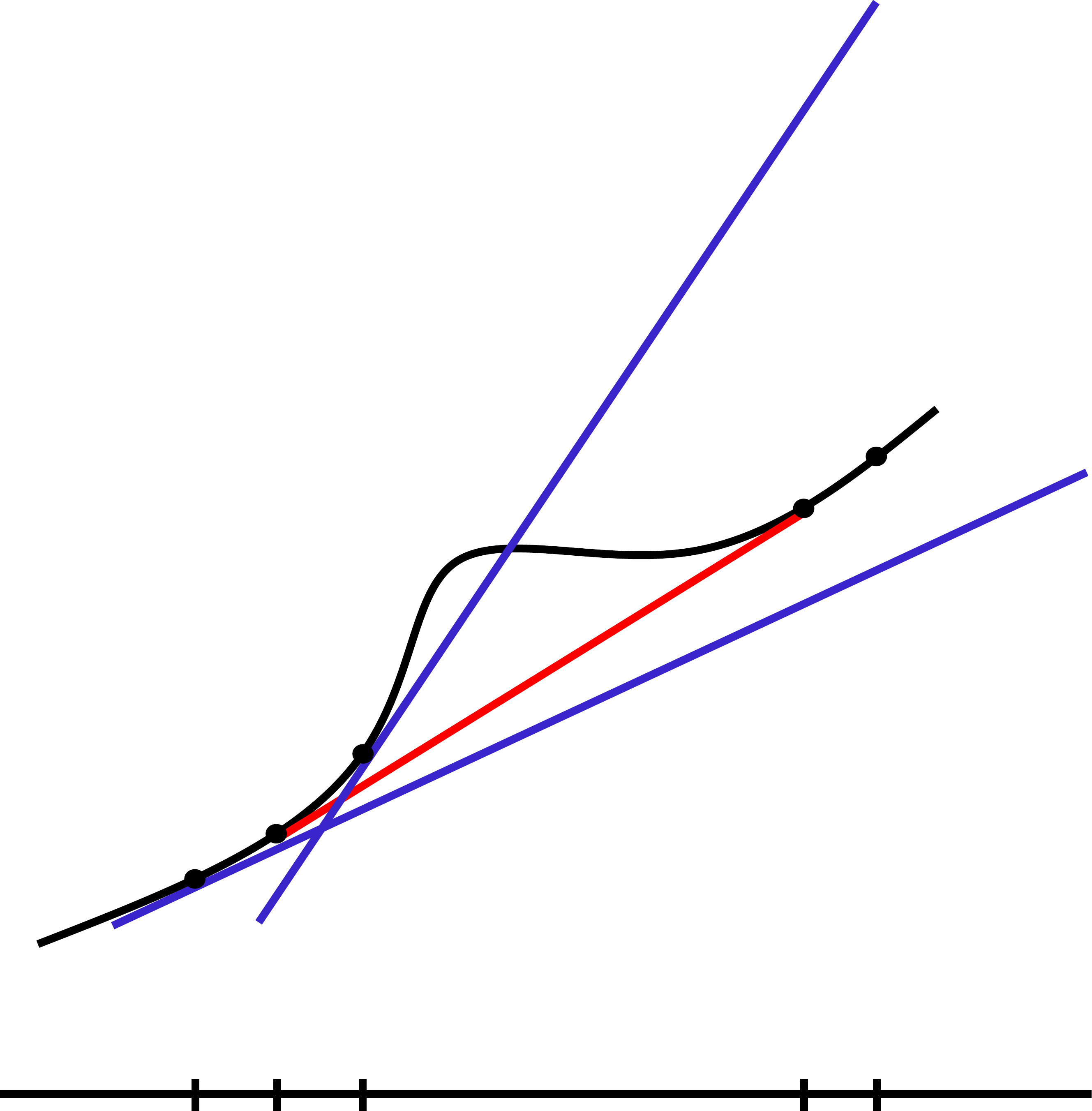}}
  	\put(45,8.00){\fontsize{09.25}{0.0}\selectfont $\color{red} x_{m}^-$}
	\put(15.00,-6.00){\fontsize{09.25}{0.0}\selectfont $x_m-\epsilon_m$}
	\put(55.00,-7.00){\fontsize{09.25}{0.0}\selectfont $x_m-\frac{3 \epsilon_m}{4}$}
	\put(135,8.00){\fontsize{09.25}{0.0}\selectfont $\color{red} x_{m}^+$}
	\put(145,-7.00){\fontsize{09.25}{0.0}\selectfont $x_{m}+\epsilon_m$}
	\put(145,80.00){\fontsize{09.25}{0.0}\selectfont $l_1$}
	\put(145,165.00){\fontsize{09.25}{0.0}\selectfont $l_2$}
	\end{picture}%
		\else
	\setlength{\unitlength}{1bp}%
	\begin{picture}(150,200)(0,0)
	\put(0,0){\includegraphics[scale=.25]{FuncionFlineas.pdf}}
	\put(45,8.00){\fontsize{09.25}{0.0}\selectfont $\color{red} x_{m}^-$}
	\put(15.00,-6.00){\fontsize{09.25}{0.0}\selectfont $x_m-\epsilon_m$}
	\put(55.00,-7.00){\fontsize{09.25}{0.0}\selectfont $x_m-\frac{3 \epsilon_m}{4}$}
	\put(135,8.00){\fontsize{09.25}{0.0}\selectfont $\color{red} x_{m}^+$}
	\put(145,-7.00){\fontsize{09.25}{0.0}\selectfont $x_{m}+\epsilon_m$}
	\put(145,80.00){\fontsize{09.25}{0.0}\selectfont $l_1$}
		\put(145,165.00){\fontsize{09.25}{0.0}\selectfont $l_2$}
	\end{picture}%
	\fi
	\caption{\label{lineastf}\gcambios{The line $l_1$ tangent to $f$ at $x_m-\epsilon_m$ lies strictly under the graph of $f$ (in black) for $x \geq x_m+\epsilon_m$. On the other hand, the line $l_2$ tangent to $f$ at $x_m-\frac{3\epsilon_m}{4}$ lies strictly over the graph of $f$. So, some line (segment in red) exists tangent to $f$ at two different points $x_m^-\in (x_m-\epsilon_m,x_m-3\epsilon_m/4)$, $x_m^+\in (x_m-3\epsilon_m/4,x_m+\epsilon_m )$.}}	
\end{figure}
 From (A) and (B), there exists some (minimum) $x_{m}^-\in (x_{m}-\epsilon_m,x_{m}-3\epsilon_m/4)$ such that the line tangent to $f$ at $(x_m^-,f(x_m^-))$ is also tangent to $f$ at $(x_m^+,f(x_m^+))$, for some  \gcambios{ first $x_m^+> x_m-3\epsilon/4$ } (in particular, $x_m^+>x_m^-$). This implies that the minimizing curve $\alpha$ 
 leaves the graph of $f$ between $(x_m^+,f(x_m^+))$ and $(x_m^-,f(x_m^-))$, and there $\alpha$ turns into a segment tangent to $f$ at those points (see Fig. \ref{lineastf}). Moreover, since
\[
f''(x)>0\quad\hbox{for all $x\in (x_m-\epsilon_m,x_m-3\epsilon_m/4)$},
\]
in particular,
$f''(x_m^-)\neq 0$,
and $\alpha$ is not $C^2$ at $(x_{m}^-,f(x_m^-))$, as required.

 \section*{Acknowledgments} Part of the results of this article are included in the PhD thesis  \cite{LuisTesis}, and the comments
  by I. Costa e Silva (U. F. Santa Catarina, Florianopolis), O.  M\"uller (Humboldt U. Berlin), D. Sol\'{\i}s
 (U. F. Yucat\'an) and  S. Suhr (U. Bochum), on the whole PhD thesis  are warmly acknowledged.  We also acknowledge the assessment  by D. Azagra (U. Complutense de Madrid) on section \ref{s2.4} and the comments by F. Finster (U. Regensburg) on
 reference~\cite{friedrichs}.
  \gcambios{Finally, we are very grateful to the referees for their careful reading of the paper and valuable
suggestions and comments.}
All the authors are partially supported by the coordinated research projects MTM2016-78807-C2-1-P (MS) and MTM2016-78807-C2-2-P (LA \& JLF) funded by  MINECO (Spanish Ministerio the Econom\'{\i}a y Competitividad) and ERDF (European Regional Development Fund). LA has also enjoyed a grant funded by the Consejo Nacional de Ciencia y Tecnolog\'{\i}a (CONACyT), M\'exico.

\end{document}